\documentclass[amsart]{article}
\usepackage{bm}

\usepackage{color}
\usepackage{amsmath}
\usepackage{amssymb}\usepackage{amsthm}
\usepackage{amsfonts}
\usepackage{amscd}
\begin{document}
\theoremstyle{plain}
\newtheorem*{ithm}{Theorem}
\newtheorem*{idefn}{Definition}
\newtheorem{thm}{Theorem}[section]
\newtheorem{lem}[thm]{Lemma}
\newtheorem{dlem}[thm]{Lemma/Definition}
\newtheorem{prop}[thm]{Proposition}
\newtheorem{set}[thm]{Setting}
\newtheorem{cor}[thm]{Corollary}
\newtheorem*{icor}{Corollary}
\theoremstyle{definition}
\newtheorem{assum}[thm]{Assumption}
\newtheorem{notation}[thm]{Notation}
\newtheorem{setting}[thm]{Setting}
\newtheorem{defn}[thm]{Definition}
\newtheorem{clm}[thm]{Claim}
\newtheorem{ex}[thm]{Example}
\theoremstyle{remark}
\newtheorem{rem}[thm]{Remark}
\newcommand{\unit}{\mathbb I}
\newcommand{\ali}[1]{{\mathfrak A}_{[ #1 ,\infty)}}
\newcommand{\alm}[1]{{\mathfrak A}_{(-\infty, #1 ]}}
\newcommand{\nn}[1]{\lV #1 \rV}
\newcommand{\br}{{\mathbb R}}
\newcommand{\dm}{{\rm dom}\mu}
\newcommand{\lb}{l_{\bb}(n,n_0,k_R,k_L,\lal,\bbD,\bbG,Y)}
\newcommand{\Ad}{\mathop{\mathrm{Ad}}\nolimits}
\newcommand{\Proj}{\mathop{\mathrm{Proj}}\nolimits}
\newcommand{\RRe}{\mathop{\mathrm{Re}}\nolimits}
\newcommand{\RIm}{\mathop{\mathrm{Im}}\nolimits}
\newcommand{\Wo}{\mathop{\mathrm{Wo}}\nolimits}
\newcommand{\Prim}{\mathop{\mathrm{Prim}_1}\nolimits}
\newcommand{\Primz}{\mathop{\mathrm{Prim}}\nolimits}
\newcommand{\ClassA}{\mathop{\mathrm{ClassA}}\nolimits}
\newcommand{\Class}{\mathop{\mathrm{Class}}\nolimits}
\newcommand{\diam}{\mathop{\mathrm{diam}}\nolimits}
\def\qed{{\unskip\nobreak\hfil\penalty50
\hskip2em\hbox{}\nobreak\hfil$\square$
\parfillskip=0pt \finalhyphendemerits=0\par}\medskip}
\def\proof{\trivlist \item[\hskip \labelsep{\bf Proof.\ }]}
\def\endproof{\null\hfill\qed\endtrivlist\noindent}
\def\proofof[#1]{\trivlist \item[\hskip \labelsep{\bf Proof of #1.\ }]}
\def\endproofof{\null\hfill\qed\endtrivlist\noindent}
\numberwithin{equation}{section}

\newcommand{\oo}{{\boldsymbol\omega}}
\newcommand{\ctv}{\caC_{(\theta,\varphi)}}
\newcommand{\btv}{\caB_{(\theta,\varphi)}}
\newcommand{\amf}{\mathfrak A}
\newcommand{\at}{\caA_{\bbZ^2}}
\newcommand{\oz}{\caO_0}
\newcommand{\caA}{{\mathcal A}}
\newcommand{\caB}{{\mathcal B}}
\newcommand{\caC}{{\mathcal C}}
\newcommand{\caD}{{\mathcal D}}
\newcommand{\caE}{{\mathcal E}}
\newcommand{\caF}{{\mathcal F}}
\newcommand{\caG}{{\mathcal G}}
\newcommand{\caH}{{\mathcal H}}
\newcommand{\caI}{{\mathcal I}}
\newcommand{\caJ}{{\mathcal J}}
\newcommand{\caK}{{\mathcal K}}
\newcommand{\caL}{{\mathcal L}}
\newcommand{\caM}{{\mathcal M}}
\newcommand{\caN}{{\mathcal N}}
\newcommand{\caO}{{\mathcal O}}
\newcommand{\caP}{{\mathcal P}}
\newcommand{\caQ}{{\mathcal Q}}
\newcommand{\caR}{{\mathcal R}}
\newcommand{\caS}{{\mathcal S}}
\newcommand{\caT}{{\mathcal T}}
\newcommand{\caU}{{\mathcal U}}
\newcommand{\caV}{{\mathcal V}}
\newcommand{\caW}{{\mathcal W}}
\newcommand{\caX}{{\mathcal X}}
\newcommand{\caY}{{\mathcal Y}}
\newcommand{\caZ}{{\mathcal Z}}
\newcommand{\bba}{{\mathbb a}}
\newcommand{\bbA}{{\mathbb A}}
\newcommand{\bbB}{{\mathbb B}}
\newcommand{\bbC}{{\mathbb C}}
\newcommand{\bbD}{{\mathbb D}}
\newcommand{\bbE}{{\mathbb E}}
\newcommand{\bbF}{{\mathbb F}}
\newcommand{\bbG}{{\mathbb G}}
\newcommand{\bbH}{{\mathbb H}}
\newcommand{\bbI}{{\mathbb I}}
\newcommand{\bbJ}{{\mathbb J}}
\newcommand{\bbK}{{\mathbb K}}
\newcommand{\bbL}{{\mathbb L}}
\newcommand{\bbM}{{\mathbb M}}
\newcommand{\bbN}{{\mathbb N}}
\newcommand{\bbO}{{\mathbb O}}
\newcommand{\bbP}{{\mathbb P}}
\newcommand{\bbQ}{{\mathbb Q}}
\newcommand{\bbR}{{\mathbb R}}
\newcommand{\bbS}{{\mathbb S}}
\newcommand{\bbT}{{\mathbb T}}
\newcommand{\bbU}{{\mathbb U}}
\newcommand{\bbV}{{\mathbb V}}
\newcommand{\bbW}{{\mathbb W}}
\newcommand{\bbX}{{\mathbb X}}
\newcommand{\bbY}{{\mathbb Y}}
\newcommand{\bbZ}{{\mathbb Z}}
\newcommand{\str}{^*}
\newcommand{\lv}{\left \vert}
\newcommand{\rv}{\right \vert}
\newcommand{\lV}{\left \Vert}
\newcommand{\rV}{\right \Vert}
\newcommand{\la}{\left \langle}
\newcommand{\ra}{\right \rangle}
\newcommand{\ltm}{\left \{}
\newcommand{\rtm}{\right \}}
\newcommand{\lcm}{\left [}
\newcommand{\rcm}{\right ]}
\newcommand{\ket}[1]{\lv #1 \ra}
\newcommand{\bra}[1]{\la #1 \rv}
\newcommand{\lmk}{\left (}
\newcommand{\rmk}{\right )}
\newcommand{\al}{{\mathcal A}}
\newcommand{\md}{M_d({\mathbb C})}
\newcommand{\ainn}{\mathop{\mathrm{AInn}}\nolimits}
\newcommand{\id}{\mathop{\mathrm{id}}\nolimits}
\newcommand{\Tr}{\mathop{\mathrm{Tr}}\nolimits}
\newcommand{\Ran}{\mathop{\mathrm{Ran}}\nolimits}
\newcommand{\Ker}{\mathop{\mathrm{Ker}}\nolimits}
\newcommand{\Aut}{\mathop{\mathrm{Aut}}\nolimits}
\newcommand{\spn}{\mathop{\mathrm{span}}\nolimits}
\newcommand{\Mat}{\mathop{\mathrm{M}}\nolimits}
\newcommand{\UT}{\mathop{\mathrm{UT}}\nolimits}
\newcommand{\DT}{\mathop{\mathrm{DT}}\nolimits}
\newcommand{\GL}{\mathop{\mathrm{GL}}\nolimits}
\newcommand{\spa}{\mathop{\mathrm{span}}\nolimits}
\newcommand{\supp}{\mathop{\mathrm{supp}}\nolimits}
\newcommand{\rank}{\mathop{\mathrm{rank}}\nolimits}
\newcommand{\idd}{\mathop{\mathrm{id}}\nolimits}
\newcommand{\ran}{\mathop{\mathrm{Ran}}\nolimits}
\newcommand{\dr}{ \mathop{\mathrm{d}_{{\mathbb R}^k}}\nolimits} 
\newcommand{\dc}{ \mathop{\mathrm{d}_{\cc}}\nolimits} \newcommand{\drr}{ \mathop{\mathrm{d}_{\rr}}\nolimits} 
\newcommand{\zin}{\mathbb{Z}}
\newcommand{\rr}{\mathbb{R}}
\newcommand{\cc}{\mathbb{C}}
\newcommand{\ww}{\mathbb{W}}
\newcommand{\nan}{\mathbb{N}}\newcommand{\bb}{\mathbb{B}}
\newcommand{\aaa}{\mathbb{A}}\newcommand{\ee}{\mathbb{E}}
\newcommand{\pp}{\mathbb{P}}
\newcommand{\wks}{\mathop{\mathrm{wk^*-}}\nolimits}
\newcommand{\mk}{{\Mat_k}}
\newcommand{\mnz}{\Mat_{n_0}}
\newcommand{\mn}{\Mat_{n}}
\newcommand{\dist}{\dc}
\newcommand{\braket}[2]{\left\langle#1,#2\right\rangle}
\newcommand{\ketbra}[2]{\left\vert #1\right \rangle \left\langle #2\right\vert}
\newcommand{\abs}[1]{\left\vert#1\right\vert}
\newcommand{\trl}[2]
{T_{#1}^{(\theta,\varphi), \Lambda_{#2},\bar V_{#1,\Lambda_{#2}}}}
\newcommand{\trlz}[1]
{T_{#1}^{(\theta,\varphi), \Lambda_{0},\unit}}
\newcommand{\trlt}[2]
{T_{#1}^{(\theta,\varphi), \Lambda_{#2}+t_{#2}\bm e_{\Lambda_{#2}},\bar V_{#1,\Lambda_{#2}+t_{#2}\bm e_{\Lambda_{#2}}}}}
\newcommand{\trltj}[4]
{T_{#1}^{(\theta,\varphi), \Lambda_{#2}^{(#3)}+t_{#4}\bm e_{\Lambda_{#2}^{(#3)}},\bar V_{#1,\Lambda_{#2}^{(#3)}+t_{#4}\bm e_{\Lambda_{#2}^{(#3)}}}}}
\newcommand{\trltjp}[4]
{T_{#1}^{(\theta,\varphi), {\Lambda'}_{#2}^{(#3)}+t_{#4}'\bm e_{{\Lambda'}_{#2}^{(#3)}},\bar V_{#1,{\Lambda'}_{#2}^{(#3)}+t_{#4}'\bm e_{{\Lambda'}_{#2}^{(#3)}}}}}
\newcommand{\trlta}[2]
{T_{#1}^{(\theta,\varphi), \Lambda_{#2}^{t_{#2}},\bar V_{#1,\Lambda_{#2}^{t_{#2}}}}}
\newcommand{\trltb}[2]
{T_{#1}^{(\theta,\varphi), \Lambda_{#2}+t\bm e_{\Lambda_{#2}},\bar V_{#1,\Lambda_{#2}+t\bm e_{\Lambda_{#2}}}}}

\newcommand{\trlpt}[2]
{T_{#1}^{(\theta,\varphi), \Lambda_{#2}'+t_{#2}'\bm e_{\Lambda_{#2}'},\bar V_{#1,\Lambda_{#2}'+t_{#2}'\bm e_{\Lambda_{#2}'}}}}

\newcommand{\trll}[3]
{T_{#1, #3}^{(\theta,\varphi), \Lambda_{#2},\bar V_{#1,\Lambda_{#2}}}}
\newcommand{\trlp}[2]
{T_{#1}^{(\theta,\varphi), \Lambda_{#2}',\bar V_{#1,\Lambda_{#2}'}}}
\newcommand{\trlpp}[2]
{T_{#1}^{(\theta,\varphi), \Lambda_{#2}'',\bar V_{#1,\Lambda_{#2}''}}}
\newcommand{\trlj}[3]
{T_{\rho_{#1}}^{(\theta,\varphi), \Lambda_{#2}^{(#3)}, V_{\rho_{#1},\Lambda_{#2}^{(#3)}}}}
\newcommand{\trljp}[3]
{T_{{\rho'}_{#1}}^{(\theta,\varphi), {\Lambda'}_{#2}^{(#3)},V_{\rho'_{#1},{\Lambda'}_{#2}^{(#3)}}}}
\newcommand{\wod}[3]
{W_{#1\Lambda_{#2}\Lambda_{#3}}}
\newcommand{\wodt}[3]
{{W^{\bm t}}_{#1\Lambda_{#2}\Lambda_{#3}}}
\newcommand{\comp}[2]
{{(\theta_{#1},\varphi_{#1}), \Lambda_{#2},\{\bar V_{\eta,\Lambda_{#2}}\}_\eta}}
\newcommand{\ltj}[2]{\Lambda_{#1}+{#2} \bm e_{\Lambda_{#1}} }
\newcommand{\ltjp}[2]{{\Lambda'}_{#1}+{#2} \bm e_{\Lambda_{#1}} }
\newtheorem{nota}{Notation}[section]
\def\qed{{\unskip\nobreak\hfil\penalty50
\hskip2em\hbox{}\nobreak\hfil$\square$
\parfillskip=0pt \finalhyphendemerits=0\par}\medskip}
\def\proof{\trivlist \item[\hskip \labelsep{\bf Proof.\ }]}
\def\endproof{\null\hfill\qed\endtrivlist\noindent}
\def\proofof[#1]{\trivlist \item[\hskip \labelsep{\bf Proof of #1.\ }]}
\def\endproofof{\null\hfill\qed\endtrivlist\noindent}
\newcommand{\wrl}[2]{Y_{#1}^{\Lambda_0^{(#2)}}}
\newcommand{\wrlt}[2]{\tilde Y_{#1}^{\Lambda_0^{(#2)}}}
\newcommand{\ZZ}{\bbZ_2\times\bbZ_2}
\newcommand{\SSS}{\mathcal{S}}
\newcommand{\cs}{S}
\newcommand{\ct}{t}
\newcommand{\hS}{S}
\newcommand{\vv}{{\boldsymbol v}}
\newcommand{\ala}{a}
\newcommand{\bet}{b}
\newcommand{\gam}{c}
\newcommand{\alphas}{\alpha}
\newcommand{\alphai}{\alpha^{(\sigma_{1})}}
\newcommand{\alphan}{\alpha^{(\sigma_{2})}}
\newcommand{\betas}{\beta}
\newcommand{\betai}{\beta^{(\sigma_{1})}}
\newcommand{\betan}{\beta^{(\sigma_{2})}}
\newcommand{\alphass}{\alpha^{{(\sigma)}}}
\newcommand{\uu}{V}
\newcommand{\vp}{\varsigma}
\newcommand{\vpr}{R}
\newcommand{\tg}{\tau_{\Gamma}}
\newcommand{\sgg}{\Sigma_{\Gamma}}
\newcommand{\nh}{t28}
\newcommand{\rk}{6}
\newcommand{\nii}{2}
\newcommand{\nhh}{28}
\newcommand{\sjt}{30}
\newcommand{\sjtg}{30}
\newcommand{\bcg}{\caB(\caH_{\alpha})\otimes  C^{*}(\Sigma)}
\newcommand{\pza}[1]{\pi_0\lmk\caA_{\Lambda_{#1}}\rmk''}
\newcommand{\pzac}[1]{\pi_0\lmk\caA_{\Lambda_{#1}}\rmk'}
\newcommand{\pzacc}[1]{\pi_0\lmk\caA_{\Lambda_{#1}^c}\rmk'}
\newcommand{\trlzi}[2]{T_{#1}^{(\theta,\varphi) \Lambda_0^{(#2)}\unit}}

\title{A derivation of braided $C^*$-tensor categories from
gapped ground states satisfying the approximate Haag duality}

\author{Yoshiko Ogata \thanks{ Graduate School of Mathematical Sciences
The University of Tokyo, Komaba, Tokyo, 153-8914, Japan
Supported in part by
the Grants-in-Aid for
Scientific Research, JSPS.}}
\maketitle

\begin{abstract}
We derive 
 braided $C^*$-tensor categories from 
gapped ground states on two-dimensional quantum spin systems
satisfying
some additional condition which we call the
 approximate Haag duality.
\end{abstract}

\section{Introduction}
Recently, the classification problem of topologically ordered gapped
systems has attracted
a lot of attention.
It is frequently said that topological order is characterized by the
{\it existence of  anyons}
and {\it long-range entanglement}.

It is well known that anyons show up in AQFT \cite{BDMRS} \cite{BF} \cite{DHRI}\cite{FRS}\cite{K}\cite{BKLR}
surprisingly naturally.
In Kitaev's quantum double model, a famous model
of gapped ground state phase, 
anyons also show up.
Applying the techniques of AQFT, 
a detailed study of anyons
in Kitaev's quantum double models in the operator algebraic setting
 was carried out by P. Naaijkens and his coauthors 
 \cite{N1}\cite{N2}\cite{FN}\cite{CNN1}.
 Furthermore, in \cite{CNN2} Cha-Naaijkens-Nachtergaele
 derived a braided $C^*$-tensor category in a general setting of
 semi-group of almost localized endomorphisms.
  
 On the other hand, in \cite{NaOg}, we proved that superselection sectors are topological invariant.
 In particular, 
 the existence of a non-trivial
superselection sector implies the long-range entanglement of the system.
There is folklore saying that the existence of
anyons implies long-range entanglement of the state.
This motivates us to derive braided $C^*$-tensor categories out of 
superselection sectors, where, unlike endomorphisms, the multiplication rule is not apriori given.
In this paper, we derive braided $C^*$-tensor categories out of
 non-trivial
sector theory provided an additional condition,
the approximate Haag duality.
In AQFT, a braided $C^*$-tensor category
is usually derived
under Haag duality  \cite{BF} \cite{DHRI}\cite{FRS}\cite{K}\cite{BKLR}, except for \cite{BDMRS}.
Haag duality itself looks too strong to be required in gapped ground phases of 
quantum spin systems, although in very nice models
like Kitaev quantum double models, it holds \cite{N1}\cite{N2}\cite{FN}.
In particular, it does not look plausible that this property is stable under 
quasi-local automorphisms, which are fundamental operations
in our classification problem \cite{bmns}\cite{NSY}.
The approximate Haag duality is a relaxation of Haag duality.
The good point about this approximate version is that it is stable under quasi-local automorphisms.
 Our derivation of the braided $C^*$-tensor category basically
 goes parallel to the recipe from AQFT. However, the ``approximateness'' of the Haag duality 
 and the other differences of our setting from AQFT require additional arguments.

\subsection{Two-dimensional quantum spin systems}
We start by summarizing standard setup of $2$-dimensional quantum spin systems on the two dimensional lattice $\bbZ^{2}$ \cite{BR1,BR2}. 
Throughout this paper, we fix some $2\le d\in\nan$.
We denote the algebra of $d\times d$ matrices by $\Mat_{d}$.
For each $z\in\bbZ^2$,  let $\caA_{\{z\}}$ be an isomorphic copy of $\Mat_{d}$, and for any finite subset $\Lambda\subset\bbZ^2$, we set $\caA_{\Lambda} = \bigotimes_{z\in\Lambda}\caA_{\{z\}}$.
For finite $\Lambda$, the algebra $\caA_{\Lambda} $ can be regarded as the set of all bounded operators acting on
the Hilbert space $\bigotimes_{z\in\Lambda}{\bbC}^{d}$.
We use this identification freely.
If $\Lambda_1\subset\Lambda_2$, the algebra $\caA_{\Lambda_1}$ is naturally embedded in $\caA_{\Lambda_2}$ by tensoring its elements with the identity. 
For an infinite subset $\Gamma\subset \bbZ^{2}$,
$\caA_{\Gamma}$
is given as the inductive limit of the algebras $\caA_{\Lambda}$ with $\Lambda$, finite subsets of $\Gamma$.
We call $\caA_{\Gamma}$ the quantum spin system on $\Gamma$.
For a subset $\Gamma_1$ of $\Gamma\subset\bbZ^{2}$,
the algebra $\caA_{\Gamma_1}$ can be regarded as a subalgebra of $\caA_{\Gamma}$. 
For $\Gamma\subset \bbR^2$, with a bit abuse of notation, we write $\caA_{\Gamma}$
to denote $\caA_{\Gamma\cap \bbZ^2}$.
Also, $\Gamma^c$ denotes the complement of $\Gamma$ in $\bbR^2$.
We regularly use the notations and  facts about cones
collected in Appendix \ref{apcone}.

\subsection{Approximate Haag duality}

A representation $(\caH,\pi_0)$ is said to satisfy
the Haag duality if $\pi_0(\caA_{\Lambda^c})'=\pi_0(\caA_{\Lambda})''$.
The corresponding (original) condition is broadly assumed in AQFT.
The problem for us about introducing this condition in quantum spin systems
is that it does not look to be plausible that this condition is stable under
quasi-local automorphisms, 
the fundamental operation in the analysis of gappd ground state phases.
For this reason, we introduce a weaker version of Haag duality.
\begin{defn}\label{assum7}[Approximate Haag duality]
Let $(\caH,\pi_0)$ be an irreducible representation of $\at$.
We say that $(\caH,\pi_0)$ satisfies the approximate Haag duality if
the following condition holds.:
For any $\varphi\in (0,2\pi)$ and 
 $\varepsilon>0$ with
$\varphi+4\varepsilon<2\pi$,
there is some $R_{\varphi,\varepsilon}>0$ and decreasing
functions $f_{\varphi,\varepsilon,\delta}(t)$, $\delta>0$
on $\bbR_{\ge 0}$
with $\lim_{t\to\infty}f_{\varphi,\varepsilon,\delta}(t)=0$
such that
\begin{description}
\item[(i)]
for any cone $\Lambda$ with $|\arg\Lambda|=\varphi$, there is a unitary 
$U_{\Lambda,\varepsilon}\in \caU(\caH)$
satisfying
\begin{align}\label{lem7p}
\pi_0\lmk\caA_{\Lambda^c}\rmk'\subset 
\Ad\lmk U_{\Lambda,\varepsilon}\rmk\lmk 
\pi_0\lmk \caA_{\lmk \Lambda-R_{\varphi,\varepsilon}\bm e_\Lambda\rmk_\varepsilon}\rmk''
\rmk,
\end{align}
and 
\item[(ii)]
 for any $\delta>0$ and $t\ge 0$, there is a unitary 
 $\tilde U_{\Lambda,\varepsilon,\delta,t}\in \pi_0\lmk \caA_{\Lambda_{\varepsilon+\delta}-t\bm e_{\Lambda}}\rmk''$
 satisfying
\begin{align}\label{uappro}
\lV
U_{\Lambda,\varepsilon}-\tilde U_{\Lambda,\varepsilon,\delta,t}
\rV\le f_{\varphi,\varepsilon,\delta}(t).
\end{align}
\end{description}
\end{defn}
The good point about this weaker version is that we know it is stable under quasi-local automorphisms.
Quasi-local automorphisms are automorphisms given by time-dependent interactions
satisfying suitable locality conditions \cite{NSY}.
We will not repeat the definition here. The only property of
quasi-local automorphisms we need is the following factorization property.
\begin{defn}\label{qfdef}
Let $\alpha$ be an automorphism of $\at$.
We say that $\alpha$ is {approximately-factorizable}
if the following condition holds.
\begin{description}
\item[(i)]
For any cone $\Lambda$ and $\delta>0$,
there are automorphisms $\beta_{\Lambda}, \tilde \beta_{\Lambda}\in\Aut\lmk\caA_{\Lambda}\rmk$,
$\beta_{\Lambda^c}, \tilde \beta_{\Lambda^c}\in \Aut\lmk\caA_{\Lambda^c}\rmk$
and $\Xi_{\Lambda, \delta}, \tilde \Xi_{\Lambda, \delta}\in \Aut\lmk\caA_{\Lambda_{\delta}\cap (\Lambda^c)_\delta}\rmk$
and unitaries $v_{\Lambda\delta},\tilde v_{\Lambda\delta}\in\at$
such that 
\begin{align}
\begin{split}
&\alpha=\Ad\lmk v_{\Lambda\delta}\rmk\circ
\Xi_{\Lambda, \delta}\circ \lmk\beta_{\Lambda}\otimes \beta_{\Lambda^c}\rmk,\\
&\alpha^{-1}=\Ad\lmk \tilde v_{\Lambda\delta}\rmk\circ
\tilde \Xi_{\Lambda, \delta}\circ \lmk\tilde \beta_{\Lambda}\otimes \tilde \beta_{\Lambda^c}\rmk.
\end{split}
\end{align}
\item[(ii)]
For each $\delta,\delta'>0,\varphi\in (0,2\pi)$, there exists a decreasing
function $g_{\varphi,\delta,\delta'}(t)$ on $\bbR_{\ge 0}$
with $\lim_{t\to\infty}g_{\varphi,\delta,\delta'}(t)=0$.
For any cone $\Lambda$ with $\varphi=|\arg\Lambda|$, there are unitaries
 $v_{\Lambda,\delta,\delta',t}', \tilde v_{\Lambda,\delta,\delta',t}'\in \caA_{\Lambda_{\delta+\delta'}-t\bm e_{\Lambda}}$
 satisfying
\begin{align}\label{uappro}
\lV
v_{\Lambda,\delta}-v_{\Lambda,\delta,\delta',t}'
\rV,\lV
\tilde v_{\Lambda,\delta}-\tilde v_{\Lambda,\delta,\delta',t}'
\rV
\le g_{\varphi,\delta,\delta'}(t),
\end{align}
for unitaries $v_{\Lambda\delta},\tilde v_{\Lambda\delta}$ in (i).
\end{description}

\end{defn} 
The approximate Haag duality is stable under the approximately-factorizable automorphisms.
\begin{prop}\label{staah}
Let $(\caH,\pi_0)$ be an irreducible representation of $\at$
satisfying the approximate Haag duality.
Then for any approximately-factorizable automorphism $\alpha$ on $\at$,
 $(\caH,\pi_0\circ\alpha)$ also satisfies the approximate Haag duality.
\end{prop}
\begin{proof}
We use the notation in Definition \ref{assum7} and Definition \ref{qfdef}.
For each 
 $\varphi\in (0,2\pi)$ and 
 $\varepsilon>0$ with
$\varphi+4\varepsilon<2\pi$, we fix $\delta>0$ 
 so that 
 \begin{align}
 16\delta<\varepsilon,\quad \varphi+16\delta<2\pi.
 \end{align}
and set
\begin{align}
\begin{split}
&R_{\varphi,\varepsilon}^{(1)}:=
R_{\varphi+2\delta,\delta},\\
&f_{\varphi,\varepsilon,\delta_0}^{(1)}(t):=\left\{
\begin{gathered}
g_{\varphi+2\delta,\delta,\delta}\lmk \frac {t-R_{\varphi+2\delta, \delta}}2\rmk
+2g_{\varphi+6\delta,\delta,\delta}\lmk \frac {t-R_{\varphi+2\delta, \delta}}2\rmk
+f_{\varphi+2\delta,\delta,\delta}\lmk \frac {t-R_{\varphi+2\delta, \delta}}2\rmk\\
+g_{\varphi+4\delta,\delta,\delta}\lmk \frac {t-R_{\varphi+2\delta, \delta}}2\rmk
+2g_{\varphi+8\delta,\delta,\delta}\lmk \frac {t-R_{\varphi+2\delta, \delta}}2\rmk
,\quad\text{if} \quad t\ge R_{\varphi+2\delta,\delta},\\
2+6\lV g\rV_{\infty}+\lV f\rV_{\infty},\quad \text{if} \quad 0\le t \le R_{\varphi+2\delta,\delta}.
\end{gathered}
\right.
\end{split}
\end{align}
for $\delta_0>0$.
%

We show (i), (ii) of Definition \ref{assum7} for this 
$f_{\varphi,\varepsilon,\delta_0}^{(1)}(t)$, $R_{\varphi,\varepsilon}^{(1)}$.
Let  $\Lambda$ be any cone with $|\arg\Lambda|=\varphi$. We have
\begin{align}\label{api}
\begin{split}
\pi_0\lmk \caA_{\lmk \Lambda_\delta\rmk ^c}\rmk
=\pi_0\circ\alpha\circ\alpha^{-1}\lmk \caA_{\lmk \Lambda_\delta\rmk^c}\rmk
=\pi_0\circ\alpha\circ\Ad\lmk \tilde v_{\Lambda_\delta,\delta}\rmk\circ
\tilde \Xi_{\Lambda_\delta,\delta}\circ \lmk\tilde \beta_{\Lambda_\delta}\otimes \tilde
\beta_{(\Lambda_\delta)^c}\rmk
\lmk \caA_{\lmk \Lambda_\delta\rmk^c}\rmk\\
\subset 
\Ad\lmk \pi_0\lmk \alpha\lmk  \tilde v_{\Lambda_\delta,\delta}
\rmk\rmk\rmk
\lmk
\pi_0\circ\alpha(\caA_{\Lambda^c})
\rmk.
\end{split}
\end{align}
We use the notation
\begin{align}
\tilde \Lambda:=
 \Lambda-R_{|\arg \Lambda|+2\delta,\delta}\bm e_\Lambda.
\end{align}

Taking the commutant of (\ref{api}), we obtain
\begin{align}
\begin{split}
&\Ad\lmk \pi_0\circ\alpha\lmk  \tilde v_{\Lambda_\delta,\delta}\rmk \rmk
\lmk
\pi_0\circ\alpha(\caA_{\Lambda^c})'
\rmk
\subset \pi_0\lmk \caA_{\lmk \Lambda_\delta\rmk ^c}\rmk'\\
&\subset \Ad\lmk U_{\Lambda_\delta,\delta}\rmk\lmk 
\pi_0\lmk \caA_{\lmk \tilde \Lambda\rmk_{2\delta}}\rmk''
\rmk
= \Ad\lmk U_{\Lambda_\delta,\delta}\rmk\lmk 
\pi_0\alpha\alpha^{-1}
\lmk \caA_{\lmk \tilde \Lambda\rmk_{2\delta}}\rmk''
\rmk\\
&=
\Ad\lmk
U_{\Lambda_\delta,\delta}
\pi_0\alpha\lmk \tilde v_{\tilde\Lambda_{2\delta}, \delta}
\rmk
\rmk\lmk
\lmk \pi_0\alpha
\circ\tilde \Xi_{\tilde\Lambda_{2\delta}, \delta}
\lmk
 \caA_{\tilde\Lambda_{2\delta}}
\rmk\rmk''\rmk\\
&\subset
\Ad\lmk
U_{\Lambda_\delta,\delta}
\pi_0\alpha\lmk \tilde v_{\tilde\Lambda_{2\delta}, \delta}
\rmk
\rmk
\lmk 
\pi_0\alpha
\lmk
 \caA_{\tilde\Lambda_{3\delta}}
\rmk''\rmk
\end{split}
\end{align}
Hence we obtain
\begin{align}
&\pi_0\circ\alpha(\caA_{\Lambda^c})'\subset
\Ad\lmk \pi_0\circ\alpha\lmk  \tilde v_{\Lambda_\delta,\delta}\rmk ^*
 U_{\Lambda_\delta,\delta}
 \pi_0\alpha\lmk \tilde v_{\tilde \Lambda_{2\delta}, \delta}
\rmk
 \rmk\lmk
 \pi_0\alpha
\lmk
 \caA_{\tilde  \Lambda_{3\delta}}
\rmk''\rmk.
\end{align}
Note that
\begin{align}
\tilde \Lambda_{3\delta}=\lmk \Lambda-R_{|\arg \Lambda|+2\delta,\delta}\bm e_\Lambda\rmk_{3\delta}
\subset \lmk \Lambda-R_{\varphi,\varepsilon}^{(1)}\bm e_\Lambda\rmk_{\varepsilon}.
\end{align}
Hence setting
\begin{align}
U_{\Lambda\epsilon}^{(1)}:=\pi_0\circ\alpha\lmk  \tilde v_{\Lambda_\delta,\delta}\rmk ^*
 U_{\Lambda_\delta,\delta}
 \pi_0\alpha\lmk \tilde v_{\tilde \Lambda_{2\delta}, \delta}
\rmk,
\end{align}
we obtain
\begin{align}
\pi_0\circ\alpha(\caA_{\Lambda^c})'\subset
\Ad\lmk
U_{\Lambda\epsilon}^{(1)}
 \rmk\lmk
 \pi_0\alpha
\lmk
 \caA_{\lmk \Lambda-R_{\varphi,\varepsilon}^{(1)}\bm e_\Lambda\rmk_{\varepsilon}}
\rmk''\rmk,
\end{align}
proving (i).
For any $t\ge 0$,
\begin{align}
\begin{split}
&\lV
\alpha\lmk  \tilde v_{\Lambda_\delta,\delta}\rmk-
\Ad\lmk v'_{\Lambda_{3\delta}-\frac t2 \bm e_{\Lambda},\delta, \delta,\frac t2 }\rmk\circ
\Xi_{\Lambda_{3\delta}-\frac t2 \bm e_{\Lambda}, \delta}
\circ\beta_{\Lambda_{3\delta}-\frac t2 \bm e_{\Lambda}}\lmk  \tilde v'_{\Lambda_\delta,\delta,\delta,
\frac t2}\rmk
\rV\\
&\le 
\lV
\alpha\lmk  \tilde v_{\Lambda_{\delta}, \delta}\rmk-
\alpha\lmk  \tilde v'_{\Lambda_\delta,\delta,\delta,
\frac t2} \rmk
\rV\\
&+
2 \lV  v'_{\Lambda_{3\delta}-\frac t2\bm e_{\Lambda},\delta, \delta,\frac t2 }- 
v_{\Lambda_{3\delta}-\frac t2\bm e_{\Lambda}, \delta}\rV\\
&\le
g_{|\arg\Lambda|+2\delta,\delta,\delta}\lmk \frac t2\rmk
+2g_{|\arg \Lambda|+6\delta,\delta,\delta}\lmk \frac t2\rmk
=g_{\varphi+2\delta,\delta,\delta}\lmk \frac t2\rmk
+2g_{\varphi+6\delta,\delta,\delta}\lmk \frac t2\rmk.
\end{split}
\end{align}
Similarly, we have
\begin{align}
\begin{split}
&\lV
\alpha\lmk  \tilde v_{{\tilde \Lambda}_{2\delta},\delta}\rmk-
\Ad\lmk v'_{{\tilde \Lambda}_{4\delta}-\frac t2 \bm e_{{\tilde \Lambda}},\delta, \delta,\frac t2 }\rmk\circ
\Xi_{{\tilde \Lambda}_{4\delta}-\frac t2 \bm e_{{\tilde \Lambda}}, \delta}
\circ\beta_{{\tilde \Lambda}_{4\delta}-\frac t2 \bm e_{{\tilde \Lambda}}}\lmk  \tilde v'_{{\tilde \Lambda}_{2\delta},\delta,\delta,
\frac t2}\rmk
\rV\\
&\le
g_{\varphi+4\delta,\delta,\delta}\lmk \frac t2\rmk
+2g_{\varphi+8\delta,\delta,\delta}\lmk \frac t2\rmk.
\end{split}
\end{align}

For each $t\ge 0$, set
\begin{align}
\begin{split}
 { U_{t}}:&=\pi_0\lmk \Ad\lmk v'_{\Lambda_{3\delta}-\frac t2\bm e_{\Lambda},\delta, \delta,\frac t2 }\rmk\circ
\Xi_{\Lambda_{3\delta}-\frac t2 \bm e_{\Lambda}, \delta}
\circ\beta_{\Lambda_{3\delta}-\frac t2 \bm e_{\Lambda}}\lmk  \tilde {v'}_{\Lambda_\delta,\delta,\delta,\frac t2}\rmk\rmk^*
\cdot
\tilde U_{\Lambda_\delta,\delta,\delta,\frac t2}\\
&\pi_0\lmk\Ad
\lmk v'_{{\tilde \Lambda}_{4\delta}-\frac t2 \bm e_{{\tilde \Lambda}},\delta, \delta,\frac t2 }\rmk\circ
\Xi_{{\tilde \Lambda}_{4\delta}-\frac t2 \bm e_{{\tilde \Lambda}}, \delta}
\circ\beta_{{\tilde \Lambda}_{4\delta}-\frac t2 \bm e_{{\tilde \Lambda}}}\lmk  \tilde v'_{{\tilde \Lambda}_{2\delta},\delta,\delta,
\frac t2}\rmk\rmk
\\
&\in 
\pi_0\lmk \caA_{\Lambda_{5\delta}-t\bm e_{\Lambda}}\rmk''
\cdot
\pi_0\lmk\caA_{\Lambda_{3\delta}-\frac t2\bm e_{\Lambda}}\rmk''
\pi_0\lmk \caA_{\tilde \Lambda_{6\delta}-t\bm e_{\Lambda}}\rmk''
\subset \pi_0\lmk
\caA_{\Lambda_{6\delta}-\lmk R_{\varphi+2\delta,\delta}+t\rmk}\bm e_{\Lambda}
\rmk''.
\end{split}
\end{align}
Because 
\begin{align}
{\Lambda_{6\delta}-\lmk R_{\varphi+2\delta,\delta}+t\rmk}\bm e_{\Lambda}
\subset \Lambda_{\varepsilon+\delta_{0}}-\lmk t+R_{\varphi+2\delta,\delta}\rmk\bm e_{\Lambda}
=\Lambda_{\varepsilon+\delta_{0}}-\lmk t+R_{\varphi,\varepsilon}^{(1)}\rmk\bm e_{\Lambda},
\end{align}
we have
\begin{align}
{ U_{t}}\in  \pi_0\lmk
\caA_{\Lambda_{\varepsilon+\delta_{0}}-\lmk t+R_{\varphi,\varepsilon}^{(1)}\rmk\bm e_{\Lambda}}
\rmk''
\end{align}
for any $t\ge 0$ and $\delta_{0}>0$.
For this $U_{t}$, we have
\begin{align}
\begin{split}
&\lV
  U_{t}
-U_{\Lambda\epsilon}^{(1)}
\rV\\
&\le
g_{\varphi+2\delta,\delta,\delta}\lmk \frac t2\rmk
+2g_{\varphi+6\delta,\delta,\delta}\lmk \frac t2\rmk
+f_{\varphi+2\delta,\delta,\delta}\lmk \frac t2\rmk
+g_{\varphi+4\delta,\delta,\delta}\lmk \frac t2\rmk
+2g_{\varphi+8\delta,\delta,\delta}\lmk \frac t2\rmk
 =f_{\varphi,\varepsilon,\delta_0}^{(1)}\lmk t+R_{\varphi+2\delta,\delta}\rmk.
 \end{split}
\end{align}
Hence setting
\begin{align}
{\tilde U_{\Lambda,\varepsilon, \delta_{0}, t}}^{(1)}
:=\left\{
\begin{gathered}
U_{t-R_{\varphi+2\delta,\delta}},\quad t\ge R_{\varphi+2\delta,\delta},\\
\unit,\quad t< R_{\varphi+2\delta,\delta}
\end{gathered}
\right.,
\end{align}
we obtain
\begin{align}
{\tilde U_{\Lambda,\varepsilon, \delta_{0}, t}}^{(1)}\in  \pi_0\lmk
\caA_{\Lambda_{\varepsilon+\delta_{0}}-t\bm e_{\Lambda}}
\rmk'',
\end{align}
and
\begin{align}
\lV
{\tilde U_{\Lambda,\varepsilon, \delta_{0}, t}}^{(1)}-U_{\Lambda\epsilon}^{(1)}
\rV
\le f_{\varphi,\varepsilon,\delta_0}^{(1)}\lmk t\rmk.
\end{align}
This proves (ii).
\end{proof}
From the proof of the proposition, we obtain the following Lemma which we will use later.
\begin{lem}\label{intr}
Let $(\caH,\pi_0)$ be a representation of $\at$
satisfying the approximate Haag duality and $\alpha$
an approximately-factorizable automorphism $\alpha$ on $\at$.
With the notations in Definition \ref{assum7} and Definition \ref{qfdef},
for any cone $\Lambda$, $\delta>0$, $t\ge R_{|\arg \Lambda|+2\delta,\delta}$ with $|\arg\Lambda|+16\delta<2\pi$,
 there is a unitary $\tilde w_{\Lambda\delta t}$ on $\caH$
 such that
 \begin{align}
 \Ad\lmk \tilde w_{\Lambda\delta t}\rmk\lmk \pi_0\circ\alpha(\caA_{\Lambda^c})'\rmk
 \subset \pi_0\lmk \caA_{\Lambda_{5\delta}-t\bm e_{\Lambda}}\rmk''
 \end{align} and
\begin{align}
\lV
\tilde w_{\Lambda,\delta,t}-\unit
\rV
\le
g_{\varphi+2\delta,\delta,\delta}\lmk \frac t2\rmk
+2g_{\varphi+6\delta,\delta,\delta}\lmk \frac t2\rmk
+f_{\varphi+2\delta,\delta,\delta}\lmk \frac t2\rmk,
\end{align}
with $\varphi=|\Lambda|$.
\end{lem}
\begin{proof}
With the notation in the proof of Proposition \ref{staah},
set
\begin{align}
\begin{split}
\tilde w_{\Lambda,\delta,t}:=&
\pi_{0}\lmk
\Ad\lmk v'_{\Lambda_{3\delta}-\frac t2 \bm e_{\Lambda},\delta, \delta,\frac t2 }\rmk\circ
\Xi_{\Lambda_{3\delta}-\frac t2 \bm e_{\Lambda}, \delta}
\circ\beta_{\Lambda_{3\delta}-\frac t2 \bm e_{\Lambda}}\lmk  \tilde v'_{\Lambda_\delta,\delta,\delta,
\frac t2}\rmk
\rmk^{*}
\tilde U_{\Lambda_\delta,\delta,\delta,\frac t2}
U_{\Lambda_\delta,\delta}^{*}
\pi_0\circ\alpha\lmk  \tilde v_{\Lambda_\delta,\delta}\rmk.
 \end{split}
 \end{align}
\end{proof}

\subsection{Main result}
Let us recall the definition of the superselection criterion.
For representations $\pi_1,\pi_2$ of a $C^*$-algebra, we write $\pi_1\simeq_{u.e.}\pi_2$
if they are unitarily equivalent.
\begin{defn}
Let $(\caH,\pi_0)$ be an irreducible representation of $\caA_{\bbZ^2}$.
We say a representation $\pi$ of $\caA_{\bbZ^2}$ on $\caH$
satisfies the superselection criterion for $\pi_0$ if 
\[
\pi\vert_{\caA_{\Lambda^c}}\simeq_{u.e.}
 \pi_0\vert_{{\caA_{\Lambda^c}}},
\]
for any cone $\Lambda$ in $\bbR^2$.
We say that $\pi_0$ has a trivial sector theory if 
any representation satisfying the superselection criterion for $\pi_0$
is quasi-equivalent to $\pi_0$.
Otherwise, we say $\pi_0$ has a non-trivial sector theory.

\end{defn}

In this paper, we show the following theorem.
For the more precise statement, see Theorem \ref{bracatthm} and Theorem \ref{monoidalthm}.
\begin{ithm}
Let  $\omega_{\Phi}$ be a pure ground state of a uniformly bounded finite range interaction $\Phi$
on $\caA_{\bbZ^2}$
satisfying the gap condition. Suppose that its
GNS representation satisfies
 the approximate Haag duality and has a non-trivial super selection sector.
 Then its superselection sectors 
 form a braided $C^*$-tensor category.
 If two of such states $\omega_{\Phi_1},\omega_{\Phi_2}$
 are connected via an approximately factorizable automorphism $\alpha$
 as $\omega_{\Phi_2}=\omega_{\Phi_1}\circ\alpha$,
 then  corresponding  braided $C^*$-tensor categories
 are monoidal equivalent.
\end{ithm}
In section \ref{supersec}, following \cite{BF},
we extend our superselection sectors to a larger $C^*$-algebra $\btv$.
As in \cite{BF}, this allows us to define the  composition
of superselection sectors (section \ref{compsec}), which gives the tensor
of the object in our braided tensor category.
In section \ref{intsec}, we introduce tensor for the morphisms
(i.e., intertwiners) and further obtain braiding operators, following the argument in \cite{DHRI}
.
Everything from section \ref{supersec} to section \ref{intsec}
is parallel to that of \cite{BF} and \cite{DHRI}, but the fact that our Haag duality
is only an approximate one requires some additional estimate.
In section \ref{braidsec}, we prove that
our superselection sectors form a braided $C^*$-tensor category.
Here, we need the gap condition to show the existence of
direct sums and subobjects.
This part requires arguments that are different from those in AQFT.
We then show the monoidal equivalence between the braided $C^*$-categories
obtained from states related by an approximate-factorizable automorphism
in section \ref{stasec}.

\section{Superselection sectors and their extensions}\label{supersec}
Throughout this section we consider the following setting.
\begin{setting}\label{setni}
Let $(\caH,\pi_0)$ be an irreducible representation of $\caA_{\bbZ^2}$.
We denote by $\caO_0$ the set of all representations of $\at$ on $\caH$
satisfying the superselection criterion for $\pi_0$.
For $\rho,\sigma\in \caO_0$ and $R\in \caB(\caH)$,
we say that $R$ is an intertwiner from $\rho$ to $\sigma$
if 
\begin{align}
R\rho(A)=\sigma(A) R,\quad A\in \at.
\end{align}
We denote by $(\rho,\sigma)$ the set of all intertwiners from $\rho$ to $\sigma$.
For each $\rho\in\oz$ and a cone 
$\Lambda$, we denote by $\caV_{\rho,\Lambda}$
the set of all unitaries $V_{\rho,\Lambda}$
satisfying
\begin{align}
\Ad\lmk V_{\rho,\Lambda}\rmk\circ\rho\vert_{\caA_{\Lambda^c}}=
 \pi_0\vert_{{\caA_{\Lambda^c}}}.
\end{align}
(By $\rho\in\oz$, $\caV_{\rho,\Lambda}$ is not empty.)
\end{setting}
\begin{lem}\label{lem3}
Consider Setting \ref{setni}.
Let $\Lambda_1,\Lambda_2$ be cones, $\rho,\sigma\in \caO_0$, and
$V_{\rho,\Lambda_1}\in \caV_{\rho,\Lambda_1}$, $V_{\sigma,\Lambda_2}\in \caV_{\sigma,\Lambda_2}$.
Let $R\in\caB(\caH)$ be an intertwiner from $\rho$ to $\sigma$, i.e.,
$\sigma(A) R=R\rho(A)$ for $A\in \at$.
Then
the operator 
$V_{\sigma,\Lambda_2} RV_{\rho,\Lambda_1}^*$
belongs to $\pi_0\lmk \caA_{\lmk \Lambda_1\cup\Lambda_2\rmk^c}\rmk'$.
\end{lem}
\begin{proof}
For any $A\in \caA_{(\Lambda_1\cup\Lambda_2)^c}$,
we have
\begin{align}
V_{\sigma,\Lambda_2} R V_{\rho,\Lambda_1}^*\pi_0(A)
=V_{\sigma,\Lambda_2} R\rho(A)  V_{\rho,\Lambda_1}^*
=V_{\sigma,\Lambda_2} \sigma(A) R  V_{\rho,\Lambda_1}^*
=\pi_0(A) V_{\rho,\Lambda_2} R V_{\rho,\Lambda_1}^*,
\end{align}
by the definition of $\caV_{\sigma,\Lambda_1}, \caV_{\sigma,\Lambda_2}$ and the fact $A\in \caA_{(\Lambda_1\cup\Lambda_2)^c}$.
\end{proof}
We introduce a $C^*$-subalgebra of $\caB(\caH)$.
\begin{defn}Consider Setting \ref{setni}. For each cone $\Lambda$, we set
\begin{align}
\mathfrak A(\Lambda):=\pi_0\lmk\caA_{\Lambda^c}\rmk'.
\end{align}
For each $\theta\in\bbR$ and $\varphi\in (0,\pi)$,
we set
\begin{align}\label{btvdef}
\caB_{(\theta,\varphi)}:=
\overline{\cup_{\Lambda\in\caC_{(\theta,\varphi)} } \mathfrak A(\Lambda) }.
\end{align}
Here $\overline{\cdot}$ denotes the norm closure.
(See Appendix \ref{apcone} (\ref{ctvdef}) for the definition of $\ctv$.)
\end{defn}
Note that $\amf(\Lambda_1)\subset \amf(\Lambda_2)$
if $\Lambda_1\subset\Lambda_2$.
Because $\caC_{(\theta,\varphi)}$ is an upward filtering set,
$\caB_{(\theta,\varphi)}$ is a $C^*$-algebra.
For each $\Lambda\in\ctv$, we have
\begin{align}
\pi_0(\caA_{\Lambda})\subset\pi_0(\caA_{\Lambda})''\subset 
\pi_0\lmk\caA_{\Lambda^c}\rmk'=\mathfrak A(\Lambda)\subset \btv.
\end{align}
By Lemma \ref{lem4}, this implies that
$\pi_0(\at)$ is a $C^*$-subalgebra of $\btv$.

\begin{lem}\label{lemhoshi}
Consider Setting \ref{setni} and assume the approximate Haag duality.
The unitary $U_{\Lambda,\varepsilon}$ in Definition \ref{assum7}
belongs to $\btv$.
The $*$-algebra $\caB_0:=\cup_{\Lambda\in \ctv}\pi_0(\caA_\Lambda)''$
is norm-dense in $\btv$.
\end{lem}
\begin{proof}
By the definition (\ref{btvdef}), it suffices to show that
for any $\Lambda\in \ctv$, $\amf(\Lambda)$ is in the norm colosure $\overline{\caB_0}$
of $\caB_0$.
For any cone $\Lambda\in\ctv$, choose 
 $\varepsilon>0$ so that
$\Lambda_{4\varepsilon}\in \ctv$.
For each $t\ge 0$ and $\delta>0$, 
 $\tilde U_{\Lambda,\varepsilon,\delta,t}$ in Definition \ref{assum7} belongs to
 $\pi_0\lmk \caA_{\Lambda_{|\arg\Lambda|+\delta}-t\bm e_{\Lambda}}\rmk''\subset \btv$,
 and we have $\lim_{t\to\infty}\lV
U_{\Lambda,\varepsilon}-\tilde U_{\Lambda,\varepsilon,\delta,t}
\rV=0$.
Hence we have
 $U_{\Lambda,\varepsilon}\in \overline{\caB_0}$.
Furthermore
 we have $\pi_0\lmk \caA_{\lmk \Lambda-R_{|\arg\Lambda|,\varepsilon}\bm e_\Lambda\rmk_\varepsilon}\rmk''\subset \caB_0$
 because $\lmk \Lambda-R_{|\arg\Lambda|,\varepsilon}\bm e_\Lambda\rmk_\varepsilon\in\ctv$.
Hence
\begin{align}
\amf(\Lambda)\subset 
\Ad(U_{\Lambda,\varepsilon})\lmk 
\pi_0\lmk \caA_{\lmk \Lambda-R_{|\arg\Lambda|,\varepsilon}\bm e_\Lambda\rmk_\varepsilon}\rmk''
\rmk
\subset  \overline{\caB_0}.
\end{align}
\end{proof}
We also note the following from the approximate Haag duality.
\begin{lem}\label{lem2t6}
Consider Setting \ref{setni} and assume the approximate Haag duality.
For any cone $\Lambda$, a unitary $u_\Lambda\in \amf(\Lambda)$,
$\varepsilon,\delta>0$ with $|\arg\Lambda|+4\varepsilon<2\pi$
and $t\ge R_{|\arg\Lambda|,\varepsilon}$,
there is some unitary $\tilde u_{\Lambda}\in \pi_0\lmk
\caA_{\Lambda_{\varepsilon+\delta}-t\bm e_{\Lambda}}
\rmk''$
such that $\lV u_\Lambda-\tilde u_{\tilde \Lambda}\rV\le 2f_{|\arg\Lambda|, \varepsilon,\delta}(t)$.
\end{lem}
\begin{proof}
Set
\begin{align}
\tilde u_{\tilde \Lambda}
:=\Ad\lmk \tilde U_{\Lambda\varepsilon\delta t} U_{\Lambda\varepsilon}^*\rmk
\lmk
u_{\Lambda}\rmk.
\end{align}
\end{proof}

\begin{defn}\label{def36}
Consider Setting \ref{setni} and assume the approximate Haag duality.
We say that a set of cones $\caS_1$ 
is distal from a set of cones $\caS_2$
if there are cones $\tilde\Lambda_1,\tilde\Lambda_2$ and $\varepsilon>0$
 such that 
\begin{align}\label{iep}
\cup_{\Lambda_1\in \caS_1} \Lambda_1\subset\tilde \Lambda_1,\quad
\cup_{\Lambda_2\in \caS_2} \Lambda_2\subset\tilde \Lambda_2
\end{align}
and 
\begin{align}\label{tep}
\lmk \tilde \Lambda_1-R_{|\arg\tilde \Lambda_1|, \varepsilon}\bm e_{\tilde \Lambda_1}\rmk_{\varepsilon}
\subset \tilde \Lambda_2^c.
\end{align}
We also say $\Lambda_1$ is distal from $\Lambda_2$ if
$\{\Lambda_1\}$ is distal from $\{\Lambda_2\}$.
Note that if $\Lambda_1$ is distal from $\Lambda_2$ then $\Lambda_1\cap\Lambda_2=\emptyset$.

Let $\theta\in \bbR$, $\varphi\in (0,\pi)$.
We say that a set of cones $\caS_1\subset \caC_{(\theta,\varphi)}$ 
is distal from a set of cones $\caS_2\subset \caC_{(\theta,\varphi)}$
with the forbidden direction ${(\theta,\varphi)}$
if $\caS_1$ 
is distal from $\caS_2$ and the above 
cones $\tilde\Lambda_1,\tilde\Lambda_2$ and $\varepsilon>0$ can be taken so that
\begin{align}\label{lep}
\lmk \tilde\Lambda_1\rmk_{\varepsilon},
\lmk \tilde\Lambda_2\rmk_{\varepsilon}\in \caC_{(\theta,\varphi)},\quad
\arg\lmk \tilde \Lambda_1\rmk_{\varepsilon}\cap \arg\lmk \tilde \Lambda_2\rmk_{\varepsilon}=\emptyset.
\end{align}

We write
$\caS_1\perp_{(\theta,\varphi)}\caS_2$ if
$\caS_1\subset \caC_{(\theta,\varphi)}$ 
is distal from a set of cones $\caS_2\subset \caC_{(\theta,\varphi)}$
with the forbidden direction ${(\theta,\varphi)}$, and
$\caS_2$ 
is distal from a set of cones $\caS_1$
with the forbidden direction ${(\theta,\varphi)}$.
Furthermore, for such $\caS_1\perp_{(\theta,\varphi)}\caS_2$,
we denote by $\caL(\caS_1,\caS_2)$ the set of
$(\varepsilon,\tilde\Lambda_1,\tilde \Lambda_2)$
satisfying (\ref{iep}), (\ref{tep}) (\ref{lep}),
and
\begin{align}
\lmk \tilde \Lambda_2-R_{|\arg\tilde \Lambda_2|, \varepsilon}\bm e_{\tilde \Lambda_2}\rmk_{\varepsilon}
\subset \tilde \Lambda_1^c.
\end{align}

\end{defn}
\begin{lem}Consider Setting \ref{setni} and assume the approximate Haag duality.
Let $\caS_1, \caS_2\subset\ctv$ be finite sets of
cones such that $\caS_1\perp_{(\theta,\varphi)}\caS_2$.
Let $(\varepsilon, \tilde \Lambda_1$ $\tilde \Lambda_2)\in \caL(\caS_1,\caS_2)$.
Then for any $t\ge 0$,
there are $t_{\Lambda_1}\ge 0$,
 $t_{\Lambda_2}\ge 0$ for each $\Lambda_1\in \caS_1$, $\Lambda_2\in \caS_2$,
such that
\begin{align}
\cup_{\Lambda_1\in\caS_1}\lmk  \Lambda_1+t_{\Lambda_1}\bm e_{\Lambda_1}\rmk
\subset \tilde \Lambda_1+t\bm e_{\tilde \Lambda_1},\quad
\cup_{\Lambda_2\in\caS_2} \lmk \Lambda_1+t_{\Lambda_2}\bm e_{\Lambda_2}\rmk
\subset \tilde \Lambda_2+t\bm e_{\tilde \Lambda_2}.
\end{align}
\end{lem}
\begin{proof}
This is immediate from Lemma \ref{lem1p}.
\end{proof}

\begin{lem}\label{lem201}
Consider Setting \ref{setni} and assume the approximate Haag duality.
Let $\Lambda_1$, $\Lambda_2$ be cones such that
$\Lambda_1$ is distal from $\Lambda_2$.
Let  $X^{t_1}\in \amf(\Lambda_1+t_1\bm e_{\Lambda_1})$
and $X^{t_2}\in \amf(\Lambda_2+t_2\bm e_{\Lambda_2})$ be
operators
given for each $t_1,t_2\ge 0$.
Suppose that $\sup_{t_i\ge 0}\lV X^{t_i}\rV<\infty$, $i=1,2$.
Then we have 
\begin{align}
\lim_{t_1,t_2\to\infty} \lV X^{t_1} X^{t_2}-X^{t_2}X^{t_1}\rV=0.
\end{align}
\end{lem}
\begin{proof}
By definition, there are cones 
$\tilde\Lambda_1,\tilde\Lambda_2$ and $\varepsilon>0$
 such that 
\begin{align}
\Lambda_1\subset \tilde \Lambda_1\subset \lmk \tilde \Lambda_1-R_{|\arg\tilde \Lambda_1|, \varepsilon}\bm e_{\tilde \Lambda_1}\rmk_{\varepsilon}
\subset \tilde \Lambda_2^c,\quad  \Lambda_2\subset\tilde \Lambda_2.
\end{align}
Note that 
\begin{align}\label{inclu}
\Lambda_1+t_1\bm e_{\Lambda_1}\subset\Lambda_1+\frac{t_1}2\bm e_{\Lambda_1}
\subset
\Lambda_1\subset \tilde \Lambda_1\subset \lmk \tilde \Lambda_1-R_{|\arg\tilde \Lambda_1|, \varepsilon}\bm e_{\tilde \Lambda_1}\rmk_{\varepsilon}
\subset \tilde \Lambda_2^c
\subset \Lambda_2^c\subset \lmk \Lambda_2+t_{2}\bm e_{\Lambda_2}\rmk^c,
\end{align}
for $t_1,t_2\ge 0$.
From (\ref{lem7p}), we have
\begin{align}
\begin{split}
&\amf(\Lambda_1+t_1\bm e_{\Lambda_1})\subset
\Ad\lmk U_{\Lambda_1+t_1\bm e_{\Lambda_1},\frac{\varepsilon}2}\rmk\lmk
 \pi_0\lmk \caA_{\lmk  \Lambda_1+\lmk t_1-R_{|\arg \Lambda_1|, \frac{\varepsilon}2}\rmk\bm e_{\Lambda_1}\rmk_{\frac{\varepsilon}2}}\rmk''\rmk\\
&\subset 
\Ad\lmk U_{\Lambda_1+t_1\bm e_{\Lambda_1},\frac{\varepsilon}2}\rmk\lmk 
\pi_0\lmk \caA_{ \lmk \Lambda_2+t_{2}\bm e_{\Lambda_2}\rmk^c}\rmk''\rmk,
\end{split}
\end{align}
for $t_1\ge R_{|\arg\Lambda_1|,\frac\varepsilon 2}$.
Therefore, $\Ad\lmk   U_{\Lambda_1+t_1\bm e_{\Lambda_1},\frac{\varepsilon}2}^*\rmk
\lmk X^{t_1}\rmk\in \pi_0\lmk \caA_{ \lmk \Lambda_2+t_{2}\bm e_{\Lambda_2}\rmk^c}\rmk''$
and $X^{t_2}\in \amf(\Lambda_2+t_2\bm e_{\Lambda_2})$
commute.
Furthermore, 
we have
\begin{align}
\begin{split}
 &\lV
  U_{\Lambda_1+t_1\bm e_{\Lambda_1},\frac{\varepsilon}2} X^{t_2}-
X^{t_2}  U_{\Lambda_1+t_1\bm e_{\Lambda_1},\frac{\varepsilon}2}
\rV\\
&\le2\lV
  U_{\Lambda_1+t_1\bm e_{\Lambda_1},\frac{\varepsilon}2}
  -
\tilde U_{\Lambda_1+t_1\bm e_{\Lambda_1},\frac{\varepsilon}2,\frac\varepsilon 2, \frac {t_1}2}
\rV\lV X^{t_2}\rV
\le 2 f_{ |\arg\Lambda_1|, \frac{\varepsilon}2, \frac{\varepsilon}2}\lmk \frac{t_1}2\rmk
\lV X^{t_2}\rV
\end{split}
\end{align}
because of (\ref{inclu})
and
\begin{align}
U_{\Lambda_1+t_1\bm e_{\Lambda_1}, \frac \varepsilon 2,\frac \varepsilon 2, \frac {t_1} 2}
\subset 
\pi_0\lmk \caA_{\lmk \Lambda_1+t_1\bm e_{\Lambda_1}-\frac {t_1}2 \bm e_{\Lambda_1}\rmk_{\varepsilon}}\rmk''.
\end{align}
An analogous estimate holds for $\lmk X^{t_2}\rmk^*$.
Hence we get 
\begin{align}
\begin{split}
 &\lim_{t_1,t_2\to\infty} \lV X^{t_1} X^{t_2}-X^{t_2}X^{t_1}\rV
 = \lim_{t_1,t_2\to\infty}  \lV
  U_{\Lambda_1+t_1\bm e_{\Lambda_1},\frac{\varepsilon}2}^*\lmk
 X^{t_1} X^{t_2}-X^{t_2}X^{t_1}
\rmk
U_{\Lambda_1+t_1\bm e_{\Lambda_1},\frac{\varepsilon}2}\rV\\
& = \lim_{t_1,t_2\to\infty}  \lV
\left[
\Ad\lmk   U_{\Lambda_1+t_1\bm e_{\Lambda_1},\frac{\varepsilon}2}^*\rmk
\lmk X^{t_1}\rmk, X^{t_2}
 \right ]
 \rV=0.
\end{split}
\end{align}
\end{proof}

\begin{lem}\label{lem23}
Consider Setting \ref{setni} and assume the approximate Haag duality.
Let $\theta\in \bbR$, $\varphi\in (0,\pi)$.
Let $\bbA_1,\bbA_2\subset \bbT$ be closed intervals such that 
$\lmk \bbA_1\cup\bbA_2\rmk\cap \bbA_{[\theta-\varphi,\theta+\varphi]}=\emptyset$
and $\bbA_1\cap \bbA_2=\emptyset$.
Then for any $\Lambda_1,\Lambda_1', \Lambda_2,\Lambda_2'\in \caC_{(\theta,\varphi)}$
with 
\[
\arg\Lambda_i, \arg\Lambda_i'\subset \bbA_i,\quad i=1,2,
\]
there are $L_i,L_i'\ge 0$, $i=1,2$
such that
\begin{align}
\{\Lambda_1+L_1\bm e_{\Lambda_1}, \Lambda_1'+L_1'\bm e_{\Lambda_1}\}\perp_{(\theta,\varphi)}
\{\Lambda_2+L_2\bm e_{\Lambda_2}, \Lambda_2'+L_2'\bm e_{\Lambda_2}\}.
\end{align}
\end{lem}
\begin{proof}
We may write $\bbA_1=\bbA_{I_1}$, $\bbA_2=\bbA_{I_2}$ with closed 
intervals $I_i=\theta_i+[-\varphi_i,\varphi_i]$, $i=1,2$ in $\bbR$. Fix some $\varepsilon>0$ such that
$\bbA_{\lmk \theta_1+[-\varphi_1 +\varepsilon,\varphi_1+\varepsilon]\rmk}\cap 
\bbA_{\lmk \theta_2+[-\varphi_2-\varepsilon,\varphi_2+\varepsilon]\rmk}=\emptyset$ and 
$\bbA_{\lmk \theta_i+[-\varphi_i-\varepsilon,\varphi_i+\varepsilon]\rmk}
\cap \bbA_{{[\theta-\varphi,\theta+\varphi]}}=\emptyset$, $i=1,2$.
Set
$\tilde\Lambda_i:=\Lambda_{\bm 0, \theta_i,\varphi_i}+(R_{2\varphi_i,\varepsilon}+1)\bm e_{\theta_i}$,
$i=1,2$.
Then we have
\begin{align}
&\lmk \tilde \Lambda_1-R_{|\arg\tilde \Lambda_1|, \varepsilon}\bm e_{\tilde \Lambda_1}\rmk_{\varepsilon}
=\Lambda_{\bm 0, \theta_1,\varphi_1+\varepsilon}+\bm e_{\theta_1}
\subset \lmk \Lambda_{\bm 0,\theta_2, \varphi_2}\rmk^c
\subset \tilde \Lambda_2^c,\\
&\lmk \tilde \Lambda_2-R_{|\arg\tilde \Lambda_2|, \varepsilon}\bm e_{\tilde \Lambda_2}\rmk_{\varepsilon}
=\Lambda_{\bm 0, \theta_2,\varphi_2+\varepsilon}+\bm e_{\theta_2}
\subset \lmk \Lambda_{\bm 0,\theta_1, \varphi_1}\rmk^c
\subset \tilde \Lambda_1^c,\\
 &\lmk \tilde \Lambda_i\rmk_\varepsilon\in\ctv,\quad i=1,2 ,\\
 &\arg\lmk \tilde \Lambda_1\rmk_{\varepsilon}\cap \arg\lmk \tilde \Lambda_2\rmk_{\varepsilon}=\emptyset.
\end{align}
By Lemma \ref{lem1p}, there are $L_i,L_i'\ge 0$, $i=1,2$
such that
\begin{align}
\{\Lambda_i+L_i\bm e_{\Lambda_i}, \Lambda_i'+L_i'\bm e_{\Lambda_i}\}
\subset \tilde\Lambda_i.
\end{align}
This proves the claim.
\end{proof}

In order to introduce the extension of $\rho\in\caO_0$
to $\btv$, we prepare the following sets.
\begin{defn}
Let $\theta\in \bbR$, $\varphi\in (0,\pi)$, and $\Lambda\in\ctv$.
Then we have $\Lambda=\Lambda_{\bm p,  \bar\theta,\bar\varphi}$
with some $\bm p\in \bbR^2$,  $\bar\theta\in \bbR$, $\bar\varphi\in (0,\pi)$
satisfying
\begin{align}\label{bard}
\theta+\varphi<\bar\theta-\bar\varphi<\bar\theta+\bar\varphi<\theta-\varphi+2\pi.
\end{align}
We denote by $\kappa_{\Lambda,\theta,\varphi}$
the set of all cones $K_\Lambda$
distal from $\Lambda$ with the forbidden direction ${(\theta,\varphi)}$,
satisfying $\arg K_\Lambda\subset \bbA_{(\theta+\varphi,\bar \theta-\bar\varphi)}$.
\end{defn}
Clearly, we have $\kappa_{\Lambda,\theta,\varphi}\subset\ctv$.
\begin{lem}\label{lem8lem8}
Consider Setting \ref{setni} and assume the approximate Haag duality.
Let $\theta\in \bbR$, $\varphi\in (0,\pi)$, and $\Lambda\in\ctv$.
Let $K_\Lambda,\tilde K_\Lambda\in \kappa_{\Lambda,\theta,\varphi}$.
Then there are $L_1,\tilde L_1\ge 0$
such that
$
\{K_\Lambda+L_1\bm e_{K_\Lambda}, \tilde K_\Lambda+\tilde L_1\bm e_{\tilde K_\Lambda}\}
$ is distal from $\{\Lambda\}$ with the forbidden direction
$(\theta,\varphi)$.
\end{lem}
\begin{proof}
Writing $\Lambda=\Lambda_{\bm p, \bar\theta,\bar\varphi}$ 
with $\bm p\in\bbR^2$ and $\bar\theta,\bar\varphi$ satisfying (\ref{bard}), 
there are $\beta,\gamma\in\bbR$ and $\varepsilon>0$ such that 
\[
\theta+\varphi<\beta-\gamma-\varepsilon<\beta+\gamma+\varepsilon<
\bar \theta-\bar\varphi
\]
and
\[
\arg K_\Lambda,\arg\tilde K_{\Lambda}\subset \bbA_{[\beta-\gamma,\beta+\gamma]},\quad
\Lambda_\varepsilon\in\ctv.
\]
Set
$\tilde\Lambda_1:=\Lambda_{\bm p, \beta,\gamma}+(R_{2\gamma,\varepsilon}+1)\bm e_{\beta}$,
and $\tilde\Lambda_2:=\Lambda$.
Then we have
\begin{align}
&\lmk \tilde \Lambda_1-R_{|\arg\tilde \Lambda_1|, \varepsilon}\bm e_{\tilde \Lambda_1}\rmk_{\varepsilon}
=\Lambda_{\bm p, \beta,\gamma+\varepsilon}+\bm e_{\beta}
\subset \lmk \Lambda_{\bm p,\bar\theta, \bar\varphi}\rmk^c
=  \Lambda^c=\tilde \Lambda_2^c,\\
 &\lmk \tilde \Lambda_i\rmk_\varepsilon\in\ctv,\quad i=1,2 \\
 &\arg\lmk \tilde \Lambda_1\rmk_{\varepsilon}\cap \arg\lmk \tilde \Lambda_2\rmk_{\varepsilon}=\emptyset.
\end{align}
By Lemma \ref{lem1p}, there are $L_1,\tilde L_1\ge 0$,
such that
\begin{align}
\{K_\Lambda+L_1\bm e_{K_\Lambda}, \tilde K_{\Lambda}+\tilde L_1\bm e_{\tilde K_\Lambda}\}
\subset \tilde\Lambda_1.
\end{align}
This proves the claim.
\end{proof}

\begin{lem}\label{lem8}
Consider Setting \ref{setni} and assume the approximate Haag duality.
Let $\theta\in\bbR$, $\varphi\in (0,\pi)$, $\Lambda_0\in \ctv$,
$\rho\in \caO_0$ and $\bar V_{\rho,\Lambda_0}\in \caV_{\rho,\Lambda_0}$.
For any $\Lambda\in \ctv$, the followings hold.
\begin{description}
\item[(i)] For any $K_\Lambda\in \kappa_{\Lambda,\theta,\varphi}$ and $V_{\rho,K_\Lambda}\in\caV_{\rho,K_\Lambda}$,
we have 
\begin{align}\label{ieq}
\Ad\lmk \bar V_{\rho, \Lambda_0} V_{\rho,K_\Lambda}^*\rmk
\lmk \amf(\Lambda)\rmk\subset\caB_{(\theta,\varphi)},
\end{align}
and 
\begin{align}\label{ieq2}
\Ad\lmk \bar V_{\rho, \Lambda_0} V_{\rho,K_\Lambda}^*\rmk\circ\pi_0(A)
=\Ad\lmk \bar V_{\rho, \Lambda_0} \rmk\circ\rho(A),\quad A\in \caA_{\Lambda}.
\end{align}
In particular, there is some $C_\Lambda\in\ctv$
such that $\Ad\lmk \bar V_{\rho, \Lambda_0} V_{\rho,K_\Lambda}^*\rmk \lmk \amf(\Lambda)\rmk
\subset\amf(C_\Lambda)$ for each $\Lambda\in\ctv$.
\item[(ii)]
For any $K_\Lambda,\tilde K_\Lambda\in \kappa_{\Lambda,\theta,\varphi}$ 
 and 
$V_{\rho,K_\Lambda}\in\caV_{\rho,K_\Lambda}$, 
$V_{\rho,\tilde K_\Lambda}\in\caV_{\rho,\tilde K_\Lambda}$, 
we have
\begin{align}
\left.\Ad\lmk \bar V_{\rho, \Lambda_0} V_{\rho,K_\Lambda}^*\rmk\right\vert_{\pi_0(\al_\Lambda)''}
=\left.\Ad\lmk \bar V_{\rho, \Lambda_0} V_{\rho,\tilde K_\Lambda}^*\rmk\right\vert_{\pi_0(\al_\Lambda)''}.
\end{align}
\end{description}
\end{lem}
\begin{proof}
{\it (i)} Because $\ctv$ is an upward filtering set, for $\Lambda, \Lambda_0, K_\Lambda\in \ctv$ 
there is a $C_\Lambda\in \ctv$
such that $\Lambda,\Lambda_0, K_\Lambda\subset C_\Lambda$.
From Lemma \ref{lem3}, we have
\begin{align}
\bar V_{\rho, \Lambda_0} V_{\rho,K_\Lambda}^*
\in \pi_0\lmk \caA_{\lmk \Lambda_0\cup K_\Lambda\rmk^c}\rmk'
\subset \pi_0\lmk \caA_{C_\Lambda^c}\rmk'=\amf(C_\Lambda)\subset \caB_{(\theta,\varphi)}.
\end{align}
Hence we conclude (\ref{ieq}) and $\Ad\lmk \bar V_{\rho, \Lambda_0} V_{\rho,K_\Lambda}^*\rmk \lmk \amf(\Lambda)\rmk
\subset\amf(C_\Lambda)$.
By the definition, we have $\Lambda\subset K_\Lambda^c$ for
$K_\Lambda\in \kappa_{\Lambda,\theta,\varphi}$.
Therefore, by the definition of $\caV_{\rho,K_\Lambda}$ (Setting \ref{setni}),
 for any $A\in \caA_{\Lambda}\subset \caA_{K_\Lambda^c}$, we have
\begin{align}
\Ad\lmk \bar V_{\rho, \Lambda_0} V_{\rho,K_\Lambda}^*\rmk\circ\pi_0(A)
=\Ad\lmk \bar V_{\rho, \Lambda_0} \rmk\circ\rho(A),
\end{align}
proving (\ref{ieq2}).

\noindent{\it (ii)}
For any $K_\Lambda,\tilde K_\Lambda\in \kappa_{\Lambda,\theta,\varphi}$ , there are 
$L_1,\tilde L_1\ge 0$
such that
$
\{K_\Lambda+L_1\bm e_{K_\Lambda}, \tilde K_\Lambda+\tilde L_1\bm e_{\tilde K_\Lambda}\}
$ is distal from $\{\Lambda\}$ with the forbidden direction
$(\theta,\varphi)$, by Lemma \ref{lem8lem8}.
Namely, there is a cone $\tilde \Lambda_1\in\ctv$
including $K_\Lambda+L_1\bm e_{K_\Lambda}, \tilde K_\Lambda+\tilde L_1\bm e_{\tilde K_\Lambda}$,
such that $\tilde\Lambda_1\cap\Lambda=\emptyset$.
Note that $\Lambda\subset K_\Lambda^c\cap\lmk\tilde\Lambda_1\rmk^c\cap \lmk\tilde K_\Lambda\rmk^c$.
By Lemma \ref{lem3}, all of
\begin{align}
\begin{split}
&  V_{\rho,K_\Lambda+L_1\bm e_{K_\Lambda}}V_{\rho,K_\Lambda}^*
  \in \pi_0\lmk \caA_{K_\Lambda^c}\rmk'=\amf(K_\Lambda),\\
  &V_{\rho,\tilde K_\Lambda+L_1\bm e_{\tilde K_\Lambda}}
   V_{\rho,K_\Lambda+L_1\bm e_{K_\Lambda}}^*
   \in \pi_0\lmk
   \caA_{\lmk \lmk \tilde K_\Lambda+L_1\bm e_{\tilde K_\Lambda}\rmk \cup \lmk K_\Lambda+L_1\bm e_{K_\Lambda}\rmk\rmk^c}
   \rmk'
  \subset  \pi_0\lmk
   \caA_{\lmk \tilde\Lambda_1\rmk^c}
   \rmk'
   =\amf( \tilde\Lambda_1),\\
 & V_{\rho,\tilde K_\Lambda} V_{\rho,\tilde K_\Lambda+L_1\bm e_{\tilde K_\Lambda}}^*
  \in \pi_0\lmk \caA_{\tilde K_\Lambda^c}\rmk'=\amf (\tilde K_\Lambda)
\end{split}
\end{align}
commute with $\pi_0(\al_\Lambda)''$.
As a result,
\begin{align} 
V_{\rho,\tilde K_\Lambda}V_{\rho,K_\Lambda}^*
=
V_{\rho,\tilde K_\Lambda} V_{\rho,\tilde K_\Lambda+L_1\bm e_{\tilde K_\Lambda}}^*
V_{\rho,\tilde K_\Lambda+L_1\bm e_{\tilde K_\Lambda}}
   V_{\rho,K_\Lambda+L_1\bm e_{K_\Lambda}}^*
 V_{\rho,K_\Lambda+L_1\bm e_{K_\Lambda}}V_{\rho,K_\Lambda}^*
\end{align}
commutes with $\pi_0\lmk\caA_\Lambda\rmk''$, proving the claim.
%
\end{proof}

\begin{defn}\label{lem8p}
Consider Setting \ref{setni} and assume the approximate Haag duality.
Let $\theta\in\bbR$, $\varphi\in (0,\pi)$, $\Lambda_0\in \ctv$,
$\rho\in \caO_0$, $\bar V_{\rho,\Lambda_0}\in \caV_{\rho,\Lambda_0}$, and
$\Lambda\in \ctv$.
From Lemma \ref{lem8}
we may define a $\sigma$w-continuous $*$-homomorphism
\begin{align}
T_{\rho,\Lambda}^{(\theta,\varphi),\Lambda_0, \bar V_{\rho,\Lambda_0}}:
\pi_0\lmk \caA_\Lambda\rmk''\to \caB_{(\theta,\varphi)}
\end{align}
by
\begin{align}
T_{\rho,\Lambda}^{(\theta,\varphi),\Lambda_0, \bar V_{\rho,\Lambda_0}}
:=\left. \Ad\lmk \bar V_{\rho, \Lambda_0} V_{\rho,K_\Lambda}^*\rmk\right\vert_{\pi_0\lmk \caA_\Lambda\rmk''},
\end{align}
independent of choice of $K_\Lambda\in \kappa_{\Lambda,\theta,\varphi}$ and $V_{\rho,K_\Lambda}\in\caV_{\rho,K_\Lambda}$.
\end{defn}

\begin{lem}\label{lem9}
Consider Setting \ref{setni} and assume the approximate Haag duality.
Let $\theta\in\bbR$, $\varphi\in (0,\pi)$, $\Lambda_0\in \ctv$,
$\rho\in \caO_0$, $\bar V_{\rho,\Lambda_0}\in \caV_{\rho,\Lambda_0}$, and
$\Lambda\in \ctv$.
Then there exists a $*$-homomorphism $T_\rho^{(\theta,\varphi), \Lambda_0,\bar V_{\rho,\Lambda_0}}
: \caB_{(\theta,\varphi)}\to \caB_{(\theta,\varphi)}$ such that 
\begin{description}
\item[(i)]
$\trl{\rho}{0}$ is $\sigma$w-continuous on $\pi_0\lmk \caA_\Lambda\rmk''$
for all $\Lambda\in\ctv$,
\item[(ii)]
$\trl{\rho}{0}\circ\pi_0(A)=\Ad\lmk \bar V_{\rho,\Lambda_0}\rmk\circ\rho(A)$,
for all $A\in \caA_{\bbZ^2}$.
\end{description}
Furthermore, it has the following properties.
\begin{description}
\item[(a)]It is unique in the sense that if $X_\rho : \btv\to \caB(\caH)$ is a $*$-homomorphism which is 
$\sigma$w-continuous on $\pi_0\lmk \caA_\Lambda\rmk''$
for all $\Lambda\in\ctv$ and $X_\rho\circ\pi_0(A)=\Ad\lmk \bar V_{\rho,\Lambda_0}\rmk\circ\rho(A)$,
for all $A\in \caA_{\bbZ^2}$, then $X_\rho=\trl{\rho}{0}$.
\item[(b)]For any cone $\Gamma$ and $V_{\rho\Gamma}\in \caV_{\rho\Gamma}$we have
\begin{align}
\trl{\rho}{0}\circ\pi_0\vert_{\caA_{\Gamma^c}}=\Ad\lmk \bar{V}_{\rho, \Lambda_0}V_{\rho,\Gamma}^*\rmk\circ\pi_0\vert_{\caA_{\Gamma^c}}.
\end{align}
In particular, the representation $\trl{\rho}{0}\circ \pi_0 : \caA_{\bbZ^2}\to \caB(\caH)$ belongs to $\caO_0$,
\item[(c)]
For any $\Lambda\in \ctv$, we have
$\trl{\rho}{0}\vert_{\pi_0\lmk \caA_\Lambda\rmk''}=\trll{\rho}{0}{\Lambda}$.
\item[(d)]
We have $\trl{\rho}{0}\circ\pi_0\vert_{\caA_{\Lambda_0^c}}=\pi_0\vert_{\caA_{\Lambda_0^c}}$.
\item[(e)]
 For any $\Lambda\in\ctv$, there is some $C_\Lambda\in\ctv$
such that $\trl\rho 0\lmk \pi_0\lmk \caA_\Lambda\rmk''\rmk
\subset\amf(C_\Lambda)$.
\end{description}
\end{lem}
\begin{proof}
Recall $\caB_0$ from Lemma \ref{lemhoshi}.
We claim that 
\begin{align}
T_0(x):=\trll{\rho}{0}{\Lambda}(x),\quad \text{if}\quad x\in \pi_0\lmk \caA_\Lambda\rmk'',\quad \Lambda\in \ctv
\end{align}
defines a well-defined $*$-homomorphism 
$T_0: \caB_0\to \btv$.
To see this, suppose that $x\in \pi_0\lmk \caA_{\Lambda_1}\rmk''\cap \pi_0\lmk \caA_{\Lambda_2}\rmk''$, for some
$\Lambda_1,\Lambda_2\in \ctv$.
Because $\ctv$ is an upward filtering set (Lemma \ref{lem1})
there is a $\Lambda\in\ctv$ such that $\Lambda_1,\Lambda_2\subset \Lambda$.
Note that $\kappa_{\Lambda,\theta,\varphi}\subset \kappa_{\Lambda_1,\theta,\varphi}, \kappa_{\Lambda_2,\theta,\varphi}$.
Therefore, by Definition \ref{lem8p}, we have
$\trll{\rho}{0}{\Lambda_1}(x)=\trll{\rho}{0}{\Lambda}(x)=\trll{\rho}{0}{\Lambda_2}(x)$.
Hence the map $T_0 : \caB_0\to \caB(\caH)$ is well-defined.
To see that it is a $*$-homomorphism, let $x\in \pza{1}$ and $y\in \pza{2}$.
As above, there is a $\Lambda\in\ctv$ such that $\Lambda_1,\Lambda_2\subset \Lambda$, and we have
$x,y\in \pza{}$. On $\pi_0(\al_\Lambda)''$, $T_0$ is a $*$-homomorphism we see that
$T_0$ is $*$-homomorphism on $\caB_0$. 
By the definition, $T_0$ is a bounded map.
As $\caB_0$ is norm-dense in $\btv$ by Lemma \ref{lemhoshi}, $T_0$ extends to a $*$-homomorphism $\trl{\rho}{0}$ 
on $\btv$ uniquely.
From Lemma \ref{lem8}, for each $\Lambda$, there is some
$C_\Lambda\in\ctv$ such that
$\trl{\rho}{0}\lmk  \pi_0(\al_\Lambda)''\rmk=T_0\lmk  \pi_0(\al_\Lambda)''\rmk\subset\amf(C_\Lambda)\subset\btv$.
This proves (e) and $\trl{\rho}{0}\lmk \btv\rmk\subset\btv$, 
By the definition, $\trl{\rho}{0}$ is clearly 
$\sigma$w-continuous on $\pza{}$
for any $\Lambda\in \ctv$ (i). (c) holds from the above definition.
To see that $\trl{\rho}{0}$ satisfies (ii), 
note that any $A\in\caA_{\rm loc}$, belongs to $\caA_\Lambda$ for some $\Lambda\in\ctv$.
Then we have
\begin{align}
\begin{split}
\trl{\rho}{0}\circ\pi_0(A)
=\trll{\rho}{0}{\Lambda} \circ\pi_0(A)
=\Ad\lmk \bar{V}_{\rho\Lambda_0}V_{\rho K_{\Lambda}}^*\rmk\circ\pi_0(A)
=\Ad\lmk \bar{V}_{\rho\Lambda_0}\rmk\circ\rho(A).
\end{split}
\end{align}
In the last equality, we used $A\in \caA_\Lambda\subset \caA_{K_\Lambda^c}$.
As this holds for any $A\in\caA_{\rm loc}$, we obtain (ii).

Next let us prove (a). 
For $X_\rho$ in (a), 
$X_\rho$ and $\trl{\rho}0$ coincides on $\pi_0(\caA_{\Lambda})$, for all $\Lambda\in\ctv$.
From the $\sigma$w-continuity of $X_\rho$ and $\trl{\rho}0$ on $\pza{}$,
$X_\rho$ and $\trl{\rho}0$ coincides on $\pi_0(\caA_{\Lambda})''$, for all $\Lambda\in\ctv$.
Hence they coincides on $\caB_0$.
Because $\caB_0$ is norm-dense in $\btv$, we obtain (a).

To prove (b),
let $\Gamma$ be a cone.
From (ii), 
\begin{align}
\trl{\rho}{0}\circ\pi_0(A)
=\Ad\lmk \bar{V}_{\rho, \Lambda_0}\rmk \circ\rho(A)
=\Ad\lmk \bar{V}_{\rho, \Lambda_0}V_{\rho,\Gamma}^*\rmk\circ\pi_0(A)
\end{align}
for any $A\in \caA_{\Gamma^c}$.
This proves (b).
(d) follows from (b) setting $\Gamma=\Lambda_0$ and $ \bar{V}_{\rho, \Lambda_0}= {V}_{\rho, \Lambda_0}$.
\end{proof}
These extensions depending on $\Lambda_0,\theta,\varphi$
are related to each other as follows.
\begin{lem}\label{lem10}
Consider Setting \ref{setni} and assume the approximate Haag duality.
We have the following.
\begin{description}
\item[(i)]
Let $\theta\in\bbR$, $\varphi\in (0,\pi)$, $\Lambda_0\in \ctv$,
$\rho_1,\rho_2\in \caO_0$ such that $\rho_1\simeq_{u.e.} \rho_2$.
Then for any 
$\bar V_{\rho_i,\Lambda_0}\in \caV_{\rho_i,\Lambda_0}$ $i=1,2$, we have
$\trl{\rho_1}{0}\simeq_{u.e.} \trl{\rho_2}{0}$.
\item[(ii)]
For  $\theta_i\in\bbR$, $\varphi_i\in (0,\pi)$ with
$(\theta_1-\varphi_1,\theta_1+\varphi_1)\subset (\theta_2-\varphi_2,\theta_2+\varphi_2)$,
 we have $\caB_{(\theta_2, \varphi_2)} \subset \caB_{(\theta_1, \varphi_1)}$ and
 $\caC_{(\theta_2, \varphi_2)} \subset \caC_{(\theta_1, \varphi_1)}$.
 For any $\Lambda_0\in \caC_{(\theta_2, \varphi_2)} $,
$\rho\in \caO_0$, $\bar V_{\rho,\Lambda_0}\in \caV_{\rho,\Lambda_0}$,
we have 
\[
T_\rho^{(\theta_1,\varphi_1)\Lambda_0\bar V_{\rho,\Lambda_0}}\vert_{\caB_{(\theta_2, \varphi_2)} }
=T_\rho^{(\theta_2,\varphi_2)\Lambda_0\bar V_{\rho,\Lambda_0}}.
\]
\item[(iii)]
Let $\theta\in\bbR$, $\varphi\in (0,\pi)$, $\Lambda_0,\Lambda_1\in \ctv$,
$\rho\in \caO_0$ and $\bar V_{\rho,\Lambda_0}\in \caV_{\rho,\Lambda_0}$,$\bar V_{\rho,\Lambda_1}\in \caV_{\rho,\Lambda_1}$.
Then 
\[
\trl{\rho}0=\Ad \lmk \bar V_{\rho,\Lambda_0}\bar V_{\rho,\Lambda_1}^*\rmk
\trl{\rho}1.
\]
\end{description}
\end{lem}
\begin{proof}
\noindent{\it (i)}
Because $\rho_1\simeq_{u.e.} \rho_2$, there is a unitary $W$ on $\caH$ such that
$\rho_2=\Ad(W)\rho_1$.
Because 
\begin{align}
\Ad\lmk V_{\rho_2,{\tilde \Lambda}} W\rmk\circ\rho_1\vert_{\caA_{{\tilde \Lambda}^c}}
=\Ad\lmk V_{\rho_2,{\tilde \Lambda}} \rmk\circ\rho_2\vert_{\caA_{{\tilde \Lambda}^c}}
=\pi_0\vert_{\caA_{{\tilde \Lambda}^c}},
\end{align}
we have $  V_{\rho_2,{\tilde \Lambda}} W\in \caV_{\rho_1,{\tilde \Lambda}}$, for any ${\tilde \Lambda}\in \ctv$ and $V_{\rho_2\tilde\Lambda}\in \caV_{\rho_2\tilde\Lambda}$.
Using this for $\tilde \Lambda:= K_\Lambda$ with $\Lambda\in\ctv$ and  (c) of Lemma \ref{lem8}, 
we have
\begin{align}
\begin{split}
&\trl{\rho_1}0\vert_{\pza{}}
=\Ad\lmk \bar V_{\rho_1,\Lambda_0}  \lmk V_{\rho_2,K_\Lambda} W\rmk^* \rmk\vert_{\pza{}}\\
&=\Ad\lmk  \bar V_{\rho_1,\Lambda_0} W^* \bar V_{\rho_2,\Lambda_0}^*   
\bar V_{\rho_2,\Lambda_0} V_{\rho_2,K_\Lambda}^* \rmk\vert_{\pza{}}
=\Ad\lmk  \bar V_{\rho_1,\Lambda_0} W^* \bar V_{\rho_2,\Lambda_0}^*   \rmk 
\trl{\rho_2}0\vert_{\pza{}},
\end{split}
\end{align}
for any $\Lambda\in\ctv$ and $K_\Lambda\in\kappa(\Lambda,\theta,\varphi)$.
Hence by Lemma \ref{lem9} (a), we obtain
\begin{align}\label{51}
\trl{\rho_1}0=\Ad\lmk  \bar V_{\rho_1,\Lambda_0} W^* \bar V_{\rho_2,\Lambda_0}^*   \rmk 
\trl{\rho_2}0.
\end{align}
\noindent{\it (ii)} The first claim is trivial.
Note that $X_\rho:=T_\rho^{(\theta_1,\varphi_1)\Lambda_0\bar V_{\rho,\Lambda_0}}\vert_{\caB_{(\theta_2, \varphi_2)} } : \caB_{(\theta_2,\varphi_2)}\to \caB(\caH)$ is a $*$-homomorphism
$\sigma$w-continous on any $\pza{}$ with $\Lambda\in \caC_{(\theta_2,\varphi_2)}
\subset \caC_{(\theta_1,\varphi_1)}$,
satisfying $X_\rho\circ\pi_0=\Ad\lmk \bar V_{\rho,\Lambda_0}  \rmk\circ\rho$.
From Lemma \ref{lem9} (a), we get $X_\rho=T_\rho^{(\theta_2,\varphi_2)\Lambda_0\bar V_{\rho,\Lambda_0}}$.

\noindent{\it (iii)}
Note that 
$X_\rho:=\Ad\lmk  \bar V_{\rho,\Lambda_0}  \bar V_{\rho,\Lambda_1}^*  \rmk\trl{\rho}1 :
\btv\to \caB(\caH)$ is a  $*$-homomorphism
$\sigma$w-continous on any $\pza{}$ with $\Lambda\in\ctv$
such that
\begin{align}
X_\rho\circ\pi_0
=\Ad\lmk \bar V_{\rho,\Lambda_0}  \bar V_{\rho,\Lambda_1}^*   \rmk\trl{\rho}1\circ\pi_0
=\Ad\lmk \bar V_{\rho,\Lambda_0}  \bar V_{\rho,\Lambda_1}^*   \rmk
\Ad \lmk \bar V_{\rho,\Lambda_1}\rmk\rho
=\Ad\lmk \bar V_{\rho,\Lambda_0} \rmk\rho.
\end{align}
From Lemma \ref{lem9} (a), we get $X_\rho=\trl{\rho}0$.
\end{proof}
\begin{lem}\label{lem33}
Consider Setting \ref{setni} and assume the approximate Haag duality.
Let $\theta\in \bbR$, $\varphi\in (0,\pi)$, and $\Lambda_0\in \ctv$.
For any $\varepsilon>0$, $\Lambda\in\ctv$ with $\Lambda_0\subset \Lambda$,
$\rho\in\caO_0$ and $\bar V_{\rho\Lambda_0}\in\caV_{\rho\Lambda_0}$,
we have
\begin{align}
\trl{\rho}{0}(\pza{})\subset \amf(\Lambda_\varepsilon).
\end{align}
\end{lem}
\begin{proof}
We may write $\Lambda=\Lambda_{\bm p,\bar\theta,\bar\varphi}$
with $\theta+\varphi<\bar\theta-\bar\varphi<\bar\theta+\bar\varphi<\theta-\varphi+2\pi$.
Set $\varepsilon':=\frac 14\min\{\varepsilon,\bar\theta-\bar\varphi-\theta-\varphi \}>0$.
Then
\begin{align}
K_\Lambda:=\Lambda_{\bm p,\bar\theta-\bar\varphi-2\varepsilon',\varepsilon'}
+\lmk
R_{2\varepsilon', \frac14\min
\{\varepsilon', \bar\theta-\bar\varphi-3\varepsilon'-\theta-\varphi\}}+1\rmk\bm 
e_{\bar\theta-\bar\varphi-2\varepsilon'}
\end{align}
belongs to $\kappa_{\Lambda,\theta,\varphi}$.
Hence, by definition, with $V_{\rho K_{\Lambda}}\in \caV_{\rho K_{\Lambda}}$, we have 
\begin{align}
\trl{\rho}{0}\vert_{\pza{}}
=\Ad\lmk \bar V_{\rho\Lambda_0} V_{\rho K_{\Lambda}}^*\rmk\vert_{\pza{}}.
\end{align}
Note that
\begin{align}
\Lambda_0\cup K_\Lambda\subset \Lambda\cup \Lambda_{\bm p, \bar\theta-\bar\varphi-2\varepsilon',\varepsilon'}\subset \Lambda_\varepsilon.
\end{align}
By Lemma \ref{lem3}, from this, we have
$V_{\rho\Lambda} V_{\rho K_{\Lambda}}^*\in \amf(\Lambda_\varepsilon)$.
Hence we get
\begin{align}
\trl{\rho}{0}\lmk {\pi_0(\al_\Lambda)''}\rmk
=\Ad\lmk \bar V_{\rho\Lambda} V_{\rho K_{\Lambda}}^*\rmk\lmk{\pza{}}\rmk
\subset  \amf(\Lambda_\varepsilon).
\end{align}
\end{proof}
\begin{lem}\label{sclem}
Consider Setting \ref{setni} and assume the approximate Haag duality.
Let $\theta\in \bbR$, $\varphi\in (0,\pi)$ and $\Lambda_0\in \ctv$.
For any $\Lambda\in\ctv$, $\rho\in\caO_0$ and $\bar V_{\rho\Lambda_0}\in\caV_{\rho\Lambda_0}$,
$\trl{\rho}{0}$ is $\sigma$w-continuous on $\amf(\Lambda)$.
\end{lem}
\begin{proof}
Let $\{x_\alpha\}_\alpha$ be a net in $\amf(\Lambda)$
converging to $x\in \amf(\Lambda)$ in the $\sigma$w-topology.
By the approcimate Haag duality, we have
$\Ad(U_{\Lambda,\varepsilon}^*)(x_\alpha), \Ad(U_{\Lambda,\varepsilon}^*)(x)\in\pi_0\lmk \caA_{\lmk \Lambda-R_{|\arg\Lambda|,\varepsilon}\bm e_\Lambda\rmk_\varepsilon}\rmk''$,
for $\varepsilon>0$ small enough.
From
$\sigma w-\lim_{\alpha}\Ad(U_{\Lambda,\varepsilon}^*)(x_\alpha)=\Ad(U_{\Lambda,\varepsilon}^*)(x)$
and the $\sigma$w-continuity of $\trl{\rho}{0}$ on 
$\pi_0\lmk \caA_{\lmk \Lambda-R_{|\arg\Lambda|,\varepsilon}\bm e_\Lambda\rmk_\varepsilon}\rmk''
$,
we have
\begin{align}
\sigma w-\lim_{\alpha}\Ad\lmk \trl{\rho}{0}\lmk  U_{\Lambda,\varepsilon}^*\rmk\rmk
\circ
\trl{\rho}{0}(x_\alpha)=\Ad\lmk \trl{\rho}{0}\lmk U_{\Lambda,\varepsilon}^*\rmk\rmk)\circ\trl{\rho}{0}(x).
\end{align}
This proves the claim.
\end{proof}

\begin{lem}\label{niku}
Consider Setting \ref{setni} and assume the approximate Haag duality.
Let $\theta\in \bbR$, $\varphi\in (0,\pi)$.
Let $\rho\in\caO_0$,
$\Lambda_1,\Lambda_2\in\ctv$, $t\ge 0$, $\varepsilon,\delta>0$
with $\lmk \Lambda_1\rmk_{\varepsilon+\delta}, \lmk \Lambda_2\rmk_{\varepsilon+\delta}\in\ctv$
and $|\arg\Lambda_2|+4\varepsilon<2\pi$.
Let $\bar V_{\rho\Lambda_1}\in \caV_{\rho\Lambda_1}$.
Let $U_{\Lambda_2\varepsilon}$, $f_{|\arg\Lambda_2|,\varepsilon,\delta}(t)$
be a unitary and a function
given in Definition \ref{assum7}.
Suppose that $\lmk\Lambda_2-t\bm e_{\Lambda_2}\rmk_{\varepsilon+\delta}\subset \Lambda_1^c$.
Then we have
\begin{align}
\lV
\trl{\rho}{1}\lmk U_{\Lambda_2,\varepsilon}\rmk  U_{\Lambda_2,\varepsilon}^*-
\unit
\rV\le 2 f_{|\arg \Lambda_2|,\varepsilon,\delta} (t).
\end{align}
\end{lem}
\begin{proof}
\begin{align}
\begin{split}
&\lV
\trl{\rho}{1}\lmk U_{\Lambda_2,\varepsilon}\rmk U_{\Lambda_2,\varepsilon}^*
-\unit\rV=
\lV
\trl{\rho}{1}\lmk U_{\Lambda_2,\varepsilon}\rmk-U_{\Lambda_2,\varepsilon}
\rV\\
&\le
\lV
 \trl{\rho}{1}\lmk 
\tilde U_{\Lambda_2,\varepsilon,\delta,t}
\rmk
-\tilde U_{\Lambda_2,\varepsilon,\delta,t}
\rV
+2\lV
\tilde U_{\Lambda_2,\varepsilon,\delta,t}
-U_{\Lambda_2,\varepsilon}
\rV\\
&=
2\lV
\tilde U_{\Lambda_2,\varepsilon,\delta,t}
-U_{\Lambda_2,\varepsilon}
\rV
\le 2 f_{|\arg \Lambda_2|,\varepsilon,\delta} (t),
\end{split}
\end{align}
because $\lmk\Lambda_2-t\bm e_{\Lambda_2}\rmk_{\varepsilon+\delta}\subset \Lambda_1^c$,
Lemma \ref{lem9} (d) and $\sigma$w-continuity of $ \trl{\rho}{1}$
on $\pi_0\lmk\caA_{\lmk \Lambda_2-t\bm e_{\Lambda_2}\rmk_{\varepsilon+\delta}}\rmk''$.
\end{proof}
\begin{lem}\label{lem319}
Consider Setting \ref{setni} and assume the approximate Haag duality.
Let $\theta\in \bbR$, $\varphi\in (0,\pi)$, $\Lambda_1,\Lambda_2\in \ctv$
with $\arg (\Lambda_1)\cap \arg (\Lambda_2)=\emptyset$.
Then for any $\rho\in\caO_0$ and $\bar V_{\rho_1\Lambda_1+t\bm e_{\Lambda_1}}
\in \caV_{\rho_1\Lambda_1+t\bm e_{\Lambda_1}}$,
$t\ge 0$
we have
\begin{align}
\lim_{t\to\infty} 
\lV
\left. \trltb{\rho}{1}\right\vert_{\amf(\Lambda_2)}-\id_{\amf(\Lambda_2)}
\rV=0.
\end{align}
\end{lem}
\begin{proof}
Choose and fix $\varepsilon,\delta>0$ such that
$(\Lambda_1)_{\varepsilon+\delta}, (\Lambda_2)_{\varepsilon+\delta}\in \ctv$,
$|\arg\Lambda_2|+4\varepsilon<2\pi$
and $\arg (\Lambda_1)_{\varepsilon+\delta}\cap \arg (\Lambda_2)_{\varepsilon+\delta}=\emptyset$.
Set
\begin{align}
\tilde\Lambda_2
:=\lmk \Lambda_2\rmk_{\varepsilon}-R_{|\arg\Lambda_2|,\varepsilon}\bm e_{\Lambda_2}.
\end{align}
By the approximate Haag duality, we have
\begin{align}
\amf(\Lambda_2)\subset \Ad\lmk U_{\Lambda_2,\varepsilon}\rmk
\lmk\pi_0\lmk \caA_{\tilde\Lambda_2}\rmk''\rmk.
\end{align}

Fix any $\epsilon'>0$.
Fix $s\ge 0$ such that 
\begin{align}
2f_{|\arg\Lambda_2|,\varepsilon,\delta}(s)<\epsilon'.
\end{align}
There is some $t_1\ge 0$ such that 
\begin{align}\label{deta}
\tilde\Lambda_2, \lmk
\Lambda_2-s\bm e_{\Lambda_2}
\rmk_{\varepsilon+\delta}
\subset 
\lmk \Lambda_1+t\bm e_{\Lambda_1}\rmk^c,
\end{align}
for all $t\ge t_1$.
From Lemma \ref{niku},
we have
\begin{align}
\lV \trltb{\rho}1\lmk U_{\Lambda_2\varepsilon}\rmk-U_{\Lambda_2\varepsilon} \rV<\epsilon',
\end{align}
for any $t\ge t_1$.
Then for any $t\ge t_1$,
$x\in \amf(\Lambda_2)$,
we have
\begin{align}
\begin{split}
&\lV
\trltb{\rho}1(x)-x
\rV
=
\lV
\trltb{\rho}1\circ \Ad\lmk U_{\Lambda_2\varepsilon}\rmk 
\circ \Ad\lmk U_{\Lambda_2\varepsilon}^*\rmk(x)-x
\rV\\
&=
\lV
\Ad\lmk\trltb{\rho}1\lmk U_{\Lambda_2\varepsilon}\rmk \rmk
\circ 
\trltb{\rho}1\lmk \Ad\lmk U_{\Lambda_2\varepsilon}^*\rmk(x)\rmk-x
\rV\\
&=\lV
\Ad\lmk\trltb{\rho}1\lmk U_{\Lambda_2\varepsilon}\rmk \rmk
\circ 
 \Ad\lmk U_{\Lambda_2\varepsilon}^*\rmk(x)-x
\rV\le 2\epsilon' \lV x\rV.
\end{split}
\end{align}
In the third equality,
we used $ \Ad\lmk U_{\Lambda_2\varepsilon}^*\rmk(x)\in \pi_0\lmk \caA_{\tilde\Lambda_2}\rmk''$
and
\begin{align}\label{kore}
\left.\trltb{\rho}1\right\vert_{\pi_0\lmk \caA_{\tilde\Lambda_2}\rmk''}
=\id_{\pi_0\lmk \caA_{\tilde\Lambda_2}\rmk''}.
\end{align}
The last property (\ref{kore}) follows from (\ref{deta})
and Lemma \ref{lem9} (d) (i).
Hence we get
\begin{align}
\lV
\left. \trltb{\rho}1\right\vert_{\amf(\Lambda_2)}-\id_{\amf(\Lambda_2)}
\rV\le 2\epsilon'
\end{align}
for any $t\ge t_1$, completing the proof.
\end{proof}
\begin{lem}\label{lem320}
Consider Setting \ref{setni} and assume the approximate Haag duality.
Let $\theta\in \bbR$, $\varphi\in (0,\pi)$, $\Lambda_1,\Lambda_2\in \ctv$
with $\arg (\Lambda_1)\cap \arg (\Lambda_2)=\emptyset$.
Then we have
\begin{align}
\lim_{t\to \infty}\sup_{\substack{x\in \amf(\Lambda_2), \\y\in \amf(\Lambda_1+t\bm e_{\Lambda_1}),\\ \lV x\rV,\lV y\rV\le 1}}
\lV
xy-yx
\rV=0.
\end{align}
\end{lem}
\begin{proof}
Choose and fix $\varepsilon,\delta>0$ such that
$(\Lambda_1)_{\varepsilon+\delta}, (\Lambda_2)_{\varepsilon+\delta}\in \ctv$
and $\arg (\Lambda_1)_{\varepsilon+\delta}\cap \arg (\Lambda_2)_{\varepsilon+\delta}=\emptyset$,
$|\arg\Lambda_2|+4\varepsilon<2\pi$
.
Recall the notations in Definition \ref{assum7}.

Fix any $\epsilon'>0$.
There is some $s\ge 0$ such that
\begin{align}
f_{|\arg\Lambda_2|,\varepsilon,\delta}(s)<\epsilon'.
\end{align}
Choose and fix such $s$. 
There is some $ t_1\ge 0$ such that
\begin{align}\label{tata}
\lmk \Lambda_2\rmk_{\varepsilon+\delta}-s\bm e_{\Lambda_2},
\quad \tilde\Lambda_2:=\lmk \Lambda_2\rmk_{\varepsilon}-R_{|\arg\Lambda_2|,\varepsilon}\bm e_{\Lambda_2}
\subset \lmk \Lambda_1+t\bm e_{\Lambda_1}\rmk^c,
\end{align}
for all $t\ge t_1$.
From the approximate Haag duality, there are some unitaries
 $U_{\Lambda_2,\varepsilon}$
 and
 $\tilde U_{\Lambda_2,\varepsilon,\delta,s}\in 
 \pi_0\lmk \caA_{\lmk \Lambda_2\rmk_{\varepsilon+\delta}-s\bm e_{\Lambda}}\rmk''$
such that
\begin{align}\label{haako}
\begin{split}
 \amf(\Lambda_2)=\pi_0\lmk\caA_{\Lambda_2^c}\rmk'\subset 
\Ad\lmk U_{\Lambda_2,\varepsilon}\rmk
\lmk 
\pi_0\lmk 
\caA_{\lmk \Lambda_2-R_{|\arg\Lambda_2|,\varepsilon}\bm e_{\Lambda_2}\rmk_\varepsilon}
\rmk''\rmk,\\
\lV
U_{\Lambda_2,\varepsilon}-\tilde U_{\Lambda_2,\varepsilon,\delta,s}
\rV\le
f_{|\arg\Lambda_2|,\varepsilon,\delta}(s)<\epsilon'.
\end{split}
\end{align}
Let $t\ge t_1$,
$x\in \amf(\Lambda_2)$, $y\in \amf(\Lambda_1+t\bm e_{\Lambda_1})$,
with $\lV x\rV,\lV y\rV\le 1$.
By (\ref{haako}) 
$
\Ad\lmk U_{\Lambda_2,\varepsilon}^*\rmk(x)
\in \pi_0\lmk \caA_{\tilde\Lambda_2}\rmk'',
$
commutes with $y$ because of (\ref{tata}).
By (\ref{tata}), $y$ also commutes with $\tilde U_{\Lambda_2,\varepsilon,\delta,s}$.
Therefore, we have
\begin{align}
\begin{split}
&\lV
xy-yx
\rV
=\lV
\Ad\lmk U_{\Lambda_2,\varepsilon}\rmk
\lmk
\Ad\lmk U_{\Lambda_2,\varepsilon}^*\rmk(x)
\rmk
y-
y \Ad\lmk U_{\Lambda_2,\varepsilon}\rmk
\lmk
\Ad\lmk U_{\Lambda_2,\varepsilon}^*\rmk(x)
\rmk
\rV\\
&\le
4\lV
U_{\Lambda_2,\varepsilon}-\tilde U_{\Lambda_2,\varepsilon,\delta,s}
\rV<4\epsilon'.
\end{split}
\end{align}
Hence for any $t\ge t_1$, we have
\begin{align}
\sup_{\substack{x\in \amf(\Lambda_2), \\y\in \amf(\Lambda_1+t\bm e_{\Lambda_1}),\\ \lV x\rV,\lV y\rV\le 1}}
\lV
xy-yx
\rV\le 4\epsilon'.
\end{align}
This completes the proof.

\end{proof}

\section{The composition}\label{compsec}
From the extension of superselection sectors in the previous section, we can define the composition
of them now.
\begin{defn}\label{compdef}
Consider Setting \ref{setni} and assume the approximate Haag duality.
Let $\theta\in\bbR$, $\varphi\in (0,\pi)$, $\Lambda_0\in \ctv$  and 
$\{\bar V_{\eta,\Lambda_0}\in \caV_{\eta,\Lambda_0}\mid \eta\in\caO_0\}$.
For any 
$\rho,\sigma\in \caO_0$, we define 
\begin{align}
\rho\;\;\circ_{\comp{}{0}}\;\;\sigma
:=\trl\rho{0}\circ\trl{\sigma}{0}\circ\pi_0 : \at\to \btv.
\end{align}
It is a well-defined $*$-homomorphism because $\pi_0(\at)\subset\btv$.
\end{defn}
\begin{lem}\label{lem11}
Consider Setting \ref{setni} and assume the approximate Haag duality.
For any $\theta_1,\theta_2\in \bbR$, $\varphi_1,\varphi_2\in (0,\pi)$,
$\Lambda_i\in\caC_{(\theta_i,\varphi_i)}$, $i=1,2$
$\rho,\sigma\in \caO_0$, and $\{\bar V_{\eta,\Lambda_i}\in \caV_{\eta,\Lambda_i}\mid \eta\in\caO_0\}$,
$i=1,2$,
we have
\begin{align}
\rho\;\;\circ_{\comp{1}{1}}\;\;\sigma\simeq_{u.e.}
\rho\;\;\circ_{\comp{2}{2}}\;\;\sigma
\end{align}
\end{lem}
\begin{proof}
{\it 1.}
For any $\theta \in \bbR$, $\varphi\in (0,\pi)$,
$\Lambda_i\in\caC_{(\theta,\varphi)}$, $i=1,2$,
$\rho,\sigma\in \caO_0$ and $\{\bar V_{\eta,\Lambda_i}\in \caV_{\eta,\Lambda_i}\mid \eta\in\caO_0\}$,
$i=1,2$, we
 show 
\begin{align}\label{63}
\rho\;\;\circ_{\comp{}{1}}\;\;\sigma\simeq_{u.e.}
\rho\;\;\circ_{\comp{}{2}}\;\;\sigma.
\end{align}
Let $\Lambda\in\ctv$ be a cone including $\Lambda_1,\Lambda_2$ (Lemma \ref{lem4}).
Fix some  $\{\bar V_{\eta,\Lambda}\in \caV_{\eta,\Lambda}\mid \eta\in\caO_0\}$
for this $\Lambda$.
Then, by Lemma \ref{lem10} (iii), we have
\begin{align}
\begin{split}
&\rho\;\;\circ_{\comp{}{i}}\;\;\sigma=
\trl{\rho}i\circ\trl{\sigma}i\circ\pi_0\\
&=\Ad\lmk \bar V_{\rho,\Lambda_i} \bar V_{\rho,\Lambda}^*\rmk\trl{\rho}{}
\circ \Ad\lmk \bar V_{\sigma,\Lambda_i} \bar V_{\sigma,\Lambda}^*\rmk\trl{\sigma}{}\circ\pi_0\\
&=\Ad\lmk \bar V_{\rho,\Lambda_i} \bar V_{\rho,\Lambda}^*
\trl{\rho}{}\lmk \bar V_{\sigma,\Lambda_i} \bar V_{\sigma,\Lambda}^*\rmk
\rmk\circ \trl{\rho}{}\trl{\sigma}{}\circ\pi_0\\
&=\Ad\lmk \bar V_{\rho,\Lambda_i} \bar V_{\rho,\Lambda}^*
\trl{\rho}{}\lmk \bar V_{\sigma,\Lambda_i} \bar V_{\sigma,\Lambda}^*\rmk
\rmk\circ
\rho\;\;\circ_{\comp{}{}}\;\;\sigma,
\end{split}
\end{align}
for $i=1,2$.
In the third equality, we used the fact $\bar V_{\sigma,\Lambda_i} \bar V_{\sigma,\Lambda}^*\in\amf(\Lambda)\subset \btv$ (Lemma \ref{lem3}).
Hence we get (\ref{63}).

\noindent{\it 2.}
For  $\theta_i\in\bbR$, $\varphi_i\in (0,\pi)$ with
$(\theta_1-\varphi_1,\theta_1+\varphi_1)\subset (\theta_2-\varphi_2,\theta_2+\varphi_2)$,
  $\Lambda_0\in \caC_{(\theta_2, \varphi_2) }\subset \caC_{(\theta_1, \varphi_1) }$,
$\rho,\sigma\in \caO_0$, and $\{\bar V_{\eta,\Lambda_0}\in \caV_{\eta,\Lambda_0}\mid \eta\in\caO_0\}$,
we have 
\begin{align}
\begin{split}
&\rho\;\;\circ_{\comp{2}{0}}\;\;\sigma
=T_\rho^{(\theta_2,\varphi_2)\Lambda_0\bar V_{\rho,\Lambda_0}}
\circ T_\sigma^{(\theta_2,\varphi_2)\Lambda_0\bar V_{\sigma,\Lambda_0}}\circ\pi_0
=T_\rho^{(\theta_1,\varphi_1)\Lambda_0\bar V_{\rho,\Lambda_0}}
\circ T_\sigma^{(\theta_2,\varphi_2)\Lambda_0\bar V_{\sigma,\Lambda_0}}\circ\pi_0\\
&=T_\rho^{(\theta_1,\varphi_1)\Lambda_0\bar V_{\rho,\Lambda_0}}
\circ T_\sigma^{(\theta_1,\varphi_1)\Lambda_0\bar V_{\sigma,\Lambda_0}}\circ\pi_0
=
\rho\;\;\circ_{\comp{1}{0}}\;\;\sigma.
\end{split}
\end{align}
by Lemma \ref{lem10} (ii).

\noindent{\it 3.}
Let  $\theta_1,\theta_2\in \bbR$, $\varphi_1,\varphi_2\in (0,\pi)$,
$\Lambda_i\in\caC_{(\theta_i,\varphi_i)}$, $i=1,2$,
$\rho,\sigma\in \caO_0$ and $\{\bar V_{\eta,\Lambda_i}\in \caV_{\eta,\Lambda_i}\mid \eta\in\caO_0\}$,
$i=1,2$.
Because $(\theta_i-\frac{\varphi_i}4, \theta_i+\frac{\varphi_i}4)\subset (\theta_i-\varphi_i, \theta_i+\varphi_i)$, and $\Lambda_i\in \caC_{(\theta_i,\varphi_i)}\subset \caC_{(\theta_i,\frac{\varphi_i}4)}$,
we get
\begin{align}\label{e111}
\rho\;\;\circ_{(\theta_i,\frac{\varphi_i}{4})\Lambda_i\{\bar V_{\eta\Lambda_i}\}}\;\;\sigma
=\rho\;\;\circ_{(\theta_i,{\varphi_i} )\Lambda_i\{\bar V_{\eta\Lambda_i}\}}\;\;\sigma
\end{align}
by {\it 2.} above.
Choose some $\theta\in\bbR$, $\varphi\in (0,\pi)$ 
such that $(\theta_i-\frac{\varphi_i}4, \theta_i+\frac{\varphi_i}4)\subset (\theta-\varphi,\theta+\varphi)$,
$i=1,2$.(These exist because $\frac{\varphi_i}4<\frac\pi 4$.)
Fix some $\Lambda_0\in \caC_{(\theta,\varphi)}\subset \caC_{(\theta_i,\frac{\varphi_i}{4})}$
and $\{\bar V_{\eta,\Lambda_0}\in \caV_{\eta,\Lambda_0}\mid \eta\in\caO_0\}$.
Then from {\it 2.}, we have
\begin{align}\label{e112}
\rho\;\;\circ_{(\theta_i,\frac{\varphi_i}{4})\Lambda_0\{\bar V_{\eta\Lambda_0}\}}\;\;\sigma
=\rho\;\;\circ_{(\theta,{\varphi})\Lambda_0\{\bar V_{\eta\Lambda_0}\}}\;\;\sigma.
\end{align}
Replacing
$\theta,\varphi, \Lambda_1,\Lambda_2,\bar V_{\eta\Lambda_1}, \bar V_{\eta\Lambda_2}$
with
$\theta_i,\frac{\varphi_i}{4}, \Lambda_0,\Lambda_i, \bar V_{\eta\Lambda_0}, \bar V_{\eta\Lambda_i}$,
we apply {\it 1.} and obtain
\begin{align}\label{e113}
\rho\;\;\circ_{(\theta_i,\frac{\varphi_i}{4})\Lambda_0\{\bar V_{\eta\Lambda_0}\}}\;\;\sigma
\simeq_{u.e.}
\rho\;\;\circ_{(\theta_i,\frac{\varphi_i}{4})\Lambda_i\{\bar V_{\eta\Lambda_i}\}}\;\;\sigma.
\end{align}
Combining (\ref{e111}), (\ref{e112}), (\ref{e113}), we complete the proof of the Lemma.
\end{proof}
\begin{lem}\label{lem12}
Consider Setting \ref{setni} and assume the approximate Haag duality.
For any $\theta\in \bbR$, $\varphi\in (0,\pi)$,
$\Lambda_0\in\caC_{(\theta,\varphi)}$, 
$\rho,\sigma\in \caO_0$ and 
$\{\bar V_{\eta,\Lambda_0}\in \caV_{\eta,\Lambda_0}\mid \eta\in\caO_0\}$,
$\rho\;\;\circ_{\comp{}{0}}\;\;\sigma$ defined in Definition \ref{compdef}
belongs to $\caO_0$.
In particular, we have
\begin{align}
\rho\;\;\circ_{(\bar\theta,\bar\varphi), \Lambda_0,\{\bar V_{\eta\Lambda_0}\}}\;\; \sigma\vert_{\caA_{\Lambda_0^c}}
=\pi_0\vert_{\caA_{\Lambda_0^c}}.
\end{align}
\end{lem}
\begin{proof}
Let  $\Lambda$ be a cone in $\bbR^2$.
For this cone, there exist $\bar\theta\in\bbR$,
$\bar\varphi\in (0,\pi)$, such that $\Lambda\in \caC_{(\bar\theta,\bar\varphi)}$. 
Fix some 
$\{\bar V_{\eta,\Lambda}\in \caV_{\eta,\Lambda}\mid \eta\in\caO_0\}$.
By the definition, we have
\begin{align}\label{121eq}
\rho\;\;\circ_{(\bar\theta,\bar\varphi), \Lambda,\{\bar V_{\eta\Lambda}\}}\;\; \sigma\vert_{\caA_{\Lambda^c}}
=T_\rho^{(\bar\theta,\bar\varphi)\Lambda\bar V_{\rho\Lambda}} 
\circ T_\sigma^{(\bar\theta,\bar\varphi)\Lambda\bar V_{\sigma\Lambda}} \circ\pi_0\vert_{\caA_{\Lambda^c}}
=\pi_0\vert_{\caA_{\Lambda^c}}
\end{align}
using Lemma \ref{lem9} (d).
From Lemma \ref{lem11}, we have
\begin{align}
\rho\;\;\circ_{(\bar\theta,\bar\varphi), \Lambda,\{\bar V_{\eta\Lambda}\}}\;\; \sigma
\simeq_{u.e.}
\rho\;\;\circ_{(\theta,\varphi), \Lambda_0,\{\bar V_{\eta\Lambda_0}\}}\;\; \sigma.
\end{align}
Hence we have
\begin{align}
\rho\;\;\circ_{(\theta,\varphi), \Lambda_0,\{\bar V_{\eta\Lambda_0}\}}\;\; \sigma\vert_{\caA_{\Lambda^c}}
\simeq_{u.e.}
\rho\;\;\circ_{(\bar\theta,\bar\varphi), \Lambda,\{\bar V_{\eta\Lambda}\}}\;\; \sigma\vert_{\caA_{\Lambda^c}}
=\pi_0\vert_{\caA_{\Lambda^c}}.
\end{align}
As this holds for any cone $\Lambda$, we obtain $\rho\;\;\circ_{\comp{}{0}}\;\;\sigma\in \caO_0$.
\end{proof}
The product preserves the unitary equivalence.
\begin{lem}\label{lem13}
Consider Setting \ref{setni} and assume the approximate Haag duality.
Let $\theta\in \bbR$, $\varphi\in (0,\pi)$,
$\Lambda_0\in\caC_{(\theta,\varphi)}$, 
$\rho_1,\rho_2,\sigma_1,\sigma_2\in \caO_0$ and 
$\{\bar V_{\eta,\Lambda_0}\in \caV_{\eta,\Lambda_0}\mid \eta\in\caO_0\}$.
Suppose that $\rho_1\simeq_{u.e.}\rho_2$ and $\sigma_1\simeq_{u.e.}\sigma_2$.
Then we have
 \begin{align}
 \rho_1\;\;\circ_{(\theta,\varphi), \Lambda_0,\{\bar V_{\eta\Lambda_0}\}}\;\; \sigma_1
 \simeq_{u.e.} \rho_2\;\;\circ_{(\theta,\varphi), \Lambda_0,\{\bar V_{\eta\Lambda_0}\}}\;\; \sigma_2.
 \end{align}
\end{lem}
\begin{proof}
Let $W_\rho,W_\sigma$ be unitaries such that 
$\rho_2=\Ad(W_\rho)\circ\rho_1$
and $\sigma_2=\Ad(W_\sigma)\circ\sigma_1$.
From the proof of Lemma \ref{lem10} (i) (\ref{51}),
we have
\begin{align}
\begin{split}
&\trl{\rho_1}{0}=\Ad\lmk\bar V_{\rho_1,\Lambda_0}W_\rho^*\bar V_{\rho_2,\Lambda_0}^*\rmk
\circ \trl{\rho_2}{0},\\
&\trl{\sigma_1}{0}=\Ad\lmk\bar V_{\sigma_1,\Lambda_0}W_\sigma^*\bar V_{\sigma_2,\Lambda_0}^*\rmk
\circ \trl{\sigma_2}{0}.
\end{split}
\end{align}
Because $W_\sigma^*$ is an intertwiner from $\sigma_2$ to $\sigma_1$,
from Lemma \ref{lem3}, 
$\bar V_{\sigma_1,\Lambda_0}W_\sigma^*\bar V_{\sigma_2,\Lambda_0}^*$
belongs to $\pi_0(\caA_{\Lambda_0^c})'=\amf(\Lambda_0)\subset\btv$.
Therefore, we have
\begin{align}
 \rho_1\;\;\circ_{(\theta,\varphi), \Lambda_0,\{\bar V_{\eta\Lambda_0}\}}\;\; \sigma_1
=\Ad\lmk\bar V_{\rho_1,\Lambda_0}W_\rho^*\bar V_{\rho_2,\Lambda_0}^*
\trl{\rho_2}{0}
\lmk\bar V_{\sigma_1,\Lambda_0}W_\sigma^*\bar V_{\sigma_2,\Lambda_0}^*\rmk
\rmk
\circ\rho_2\;\;\circ_{(\theta,\varphi), \Lambda_0,\{\bar V_{\eta\Lambda_0}\}}\;\; \sigma_2.
\end{align}
This proves the Lemma.
\end{proof}
\begin{lem}\label{ccon}
Consider Setting \ref{setni} and assume the approximate Haag duality.
Let $\theta\in \bbR$, $\varphi\in (0,\pi)$.
Let $\Lambda_1,\Lambda_2\in \ctv$,
 $\rho,\sigma\in \caO_0$,
and $V_{\rho\Lambda_i}\in \caV_{\rho\Lambda_i}$, $V_{\sigma\Lambda_i}\in \caV_{\sigma\Lambda_i}$,
$i=1,2$.
Then for any $\Lambda\in\ctv$,
$\trl{\rho}1\circ\trl{\sigma}2$
is $\sigma$w-continuous on $\amf(\Lambda)$.
\end{lem}
\begin{proof}
Let $\Lambda\in\ctv$ and $0<\varepsilon<\frac{2\pi-|\arg\Lambda|}{4}$.
We have 
\begin{align}
\begin{split}
&\trl{\sigma}2\lmk\amf(\Lambda)\rmk\subset 
\Ad\lmk \trl{\sigma}2 \lmk  U_{\Lambda,\varepsilon}\rmk\rmk
\lmk\trl{\sigma}2\lmk
\pi_0\lmk \caA_{\lmk \Lambda-R_{|\Lambda|,\varepsilon}\bm e_\Lambda\rmk_\varepsilon}\rmk''
\rmk\rmk\\
&
\subset 
\Ad\lmk \trl{\sigma}2 \lmk  U_{\Lambda,\varepsilon}\rmk\rmk
\lmk
\amf\lmk C_{\lmk \Lambda-R_{|\Lambda|,\varepsilon}\bm e_\Lambda\rmk_\varepsilon}\rmk
\rmk
\end{split}
\end{align}
for some cone $ C_{\lmk \Lambda-R_{|\Lambda|,\varepsilon}\bm e_\Lambda\rmk_\varepsilon}\in\ctv$.
(Lemma \ref{lem9} (e)).
By the $\sigma$w-continuity of $\trl{\rho}1$ on 
$\amf\lmk C_{\lmk \Lambda-R_{|\Lambda|,\varepsilon}\bm e_\Lambda\rmk_\varepsilon}\rmk
$
and the $\sigma$w-continuity of 
$\trl{\sigma}2$
on $\amf(\Lambda)$ (Lemma \ref{sclem})
$\trl{\rho}1\circ\trl{\sigma}2$ is
 $\sigma$w-continuous on  $\amf(\Lambda)$.
\end{proof}

\begin{lem}\label{lem14}
Consider Setting \ref{setni} and assume the approximate Haag duality.
Let $\theta\in \bbR$, $\varphi\in (0,\pi)$,
$\Lambda_0\in\caC_{(\theta,\varphi)}$, 
$\rho, \sigma\in \caO_0$ and 
$\{\bar V_{\eta,\Lambda_0}\in \caV_{\eta,\Lambda_0}\mid \eta\in\caO_0\}$.
Set
\[
\gamma:=\rho\;\;\circ_{\comp{}{0}}\;\;\sigma\in\caO_0.
\]
Then we have
\begin{align}
\trl{\gamma}{0}=\Ad\lmk \bar V_{\gamma\Lambda_0}\rmk
\circ \trl{\rho}0\circ\trl{\sigma}0,
\end{align}
$1\in \caV_{\gamma,\Lambda_0}$,
and $\bar V_{\gamma\Lambda_0}\in\amf(\Lambda_0)\subset\btv$.
\end{lem}
\begin{proof}
The map $X:=\Ad\lmk \bar V_{\gamma\Lambda_0}\rmk
\circ \trl{\rho}0\circ\trl{\sigma}0: \btv\to \caB(\caH)$ is a
$*$-homomorphism.
By Lemma \ref{ccon},
$X$ is 
 $\sigma$w-continuous on  $\pza{}$ for any $\Lambda\in \ctv$.
 Furthermore, we have $X\circ\pi_0=\Ad\lmk \bar V_{\gamma\Lambda_0}\rmk\circ\gamma$.
Hence by the uniqueness in  Lemma \ref{lem9} (a), we obtain $X=\trl{\gamma}{0}$.
As in (\ref{121eq}), we have $\gamma\vert_{\caA_{\Lambda_0^c}}=\pi_0\vert_{\caA_{\Lambda_0^c}}$.
This means $1\in \caV_{\gamma,\Lambda_0}$.
Therefore, we have
\[
\bar V_{\gamma\Lambda_0}=\bar V_{\gamma\Lambda_0}\cdot 1^*\in \amf(\Lambda_0),
\]
by Lemma \ref{lem3}.
\end{proof}
Finally we obtain the associativity in the following sense.
\begin{lem}\label{lem15}
Consider Setting \ref{setni} and assume the approximate Haag duality.
Let $\theta\in \bbR$, $\varphi\in (0,\pi)$,
$\Lambda_0\in\caC_{(\theta,\varphi)}$, 
$\rho, \sigma,\gamma\in \caO_0$ and 
$\{\bar V_{\eta,\Lambda_0}\in \caV_{\eta,\Lambda_0}\mid \eta\in\caO_0\}$.
Then we have
\begin{align}
\lmk \rho\;\;\circ_{\comp{}{0}}\;\;\sigma\rmk\circ\;\;_{\comp{}{0}}\;\;\gamma
\simeq_{u.e.}
\rho\;\;\circ_{\comp{}{0}}\;\;\lmk \sigma\circ\;\;_{\comp{}{0}}\;\;\gamma\rmk.
\end{align}
\end{lem}
\begin{proof}
Set $\alpha:= \rho\;\;\circ_{\comp{}{0}}\;\;\sigma$
and $\beta:= \sigma\circ\;\;_{\comp{}{0}}\;\;\gamma$.
Applying Lemma \ref{lem14}, we have
\begin{align}
\begin{split}
&\lmk \rho\;\;\circ_{\comp{}{0}}\;\;\sigma\rmk\circ\;\;_{\comp{}{0}}\;\;\gamma
=\trl{\alpha}0 \circ\trl{\gamma}0\circ\pi_0\\
&=\Ad\lmk\bar V_{\alpha\Lambda_0}\rmk \trl{\rho}0\trl{\sigma}0\trl{\gamma}0\pi_0
=\Ad\lmk\bar V_{\alpha\Lambda_0}\rmk \trl{\rho}0
\Ad\lmk \bar V_{\beta\Lambda_0}^*\rmk\trl{\beta}0\pi_0\\
&=\Ad\lmk\bar V_{\alpha\Lambda_0}  \trl{\rho}0\lmk  \bar V_{\beta\Lambda_0}^*\rmk \rmk
\trl{\rho}0\trl{\beta}0\pi_0\\
&=\Ad\lmk\bar V_{\alpha\Lambda_0}  \trl{\rho}0\lmk  \bar V_{\beta\Lambda_0}^*\rmk \rmk
\rho\;\;\circ_{\comp{}{0}}\;\;\lmk \sigma\circ\;\;_{\comp{}{0}}\;\;\gamma\rmk.
\end{split}
\end{align}
In the fourth equality, we used $\bar V_{\beta\Lambda_0}\in\btv$, Lemma \ref{lem14}.
\end{proof}

\section{The intertwiners}\label{intsec}
For endomorphisms $T_1, T_2$ of $\btv$, we denote by
$(T_1,T_2)$ the set of all intertwiners from $T_1$ to $T_2$, i.e.,
the set of all bounded operators $R$ on $\caH$ such that $RT_1(x)=T_2(x) R$,
for all $x\in \btv$.

\begin{lem}\label{lem16}
Consider Setting \ref{setni} and assume the approximate Haag duality.
Let $\theta\in \bbR$, $\varphi\in (0,\pi)$,
$\Lambda_1,\Lambda_2\in\caC_{(\theta,\varphi)}$, 
$\rho\in \caO_0$ and 
$ \bar V_{\rho,\Lambda_i}\in \caV_{\rho,\Lambda_i}$, $i=1,2$.
Then we have
\begin{align}
\bar V_{\rho,\Lambda_2}\bar V_{\rho,\Lambda_1}^*\in \lmk \trl{\rho}1,\trl{\rho}2\rmk.
\end{align}
\end{lem}
\begin{proof}
By Lemma \ref{lem9} (ii), we have
\begin{align}
\bar V_{\rho,\Lambda_2}\bar V_{\rho,\Lambda_1}^* \trl{\rho}1(\pi_0(A))
=\trl{\rho}2(\pi_0(A)) \bar V_{\rho,\Lambda_2}\bar V_{\rho,\Lambda_1}^*,\quad A\in\at.
\end{align}
From the $\sigma$w-continuity of  $\trl{\rho}i$, $i=1,2$ on $\pza{}$,
$\Lambda\in\ctv$ (Lemma \ref{lem9} (i)), this equality extends to any $\pi_0(\caA_{\Lambda})''$
and hence to $\btv$ by Lemma \ref{lemhoshi}.
\end{proof}
\begin{lem}\label{lem17}
Consider Setting \ref{setni} and assume the approximate Haag duality.
Let $\theta\in \bbR$, $\varphi\in (0,\pi)$,
$\Lambda_1,\Lambda_2\in\caC_{(\theta,\varphi)}$, 
$\rho,\sigma\in \caO_0$ and 
$ \bar V_{\rho,\Lambda_1}\in \caV_{\rho,\Lambda_1}$, 
$ \bar V_{\sigma,\Lambda_2}\in \caV_{\rho,\Lambda_2}$.
Then we have
\begin{align}
\lmk \trl{\rho}1,\trl{\sigma}2\rmk\subset \pi_0(\caA_{(\Lambda_1\cup\Lambda_2)^c})'
\subset \btv.
\end{align}
\end{lem}
\begin{proof}
For $R\in \lmk \trl{\rho}1,\trl{\sigma}2\rmk$ we have
\begin{align}
R \pi_0(A)=R \trl{\rho}1(\pi_0(A))
= \trl{\sigma}2(\pi_0(A)) R=\pi_0(A)R,\quad A\in \al_{\lmk\Lambda_1\cup\Lambda_2\rmk^c}
\end{align}
because of Lemma \ref{lem9} (d).
\end{proof}
\begin{lem}\label{lem18}
Consider Setting \ref{setni} and assume the approximate Haag duality.
Let $\theta\in \bbR$, $\varphi\in (0,\pi)$,
$\Lambda_1,\Lambda_1',\Lambda_2,\Lambda_2'\in\caC_{(\theta,\varphi)}$, 
$\rho,\rho',\sigma,\sigma'\in \caO_0$ and 
$ V_{\rho,\Lambda_1}\in \caV_{\rho,\Lambda_1}$, $ V_{\rho',\Lambda_1'}\in \caV_{\rho',\Lambda_1'}$, 
$ V_{\sigma,\Lambda_2}\in \caV_{\sigma,\Lambda_2}$,
$ V_{\sigma',\Lambda_2'}\in \caV_{\sigma',\Lambda_2'}$.
Then for any  $R_1\in \lmk \trl{\rho}1,\trlp{\rho'}1\rmk$ and
$R_2\in \lmk \trl{\sigma}2,\trlp{\sigma'}2\rmk$, we
have 
\begin{align}\label{mordef}
R_1\otimes R_2:=R_1\trl{\rho}1(R_2)
\in \lmk\trl{\rho}1\circ  \trl{\sigma}2, \trlp{\rho'}1\circ \trlp{\sigma'}2\rmk,
\end{align}
and 
\begin{align}
\lmk R_1\otimes R_2\rmk^*=R_1^*\otimes R_2^*.
\end{align}
\end{lem}
Note from Lemma \ref{lem17},
$\trl{\rho}1(R_2)$ is well-defined.
\begin{proof}
This follows directly from the definition.
\end{proof}
\begin{defn}
In the above setting, we define $R_1\otimes R_2$ by (\ref{mordef}).
\end{defn}

\begin{lem}\label{lem19}
Consider Setting \ref{setni} and assume the approximate Haag duality.
Let $\theta\in \bbR$, $\varphi\in (0,\pi)$,
$\Lambda_i,\Lambda_i',\Lambda_i''\in\caC_{(\theta,\varphi)}$, 
$\rho_i,\rho_i',\rho_i''\in \caO_0$ and 
$ V_{\rho_i,\Lambda_i}\in \caV_{\rho_i,\Lambda_i}$, 
$ V_{\rho_i',\Lambda_i'}\in \caV_{\rho_i',\Lambda_i'}$
$ V_{\rho_i'',\Lambda_i''}\in \caV_{\rho_i'',\Lambda_i''}$, $i=1,2$.
Then for any  $R_i\in \lmk \trl{\rho}i,\trlp{\rho'}i\rmk$ and
$R_i'\in \lmk \trl{\rho'}i,\trlp{\rho''}i\rmk$ $i=1,2$,
we have
\begin{align}
R_1'R_1\otimes R_2' R_2
=(R_1'\otimes R_2')(R_1\otimes R_2).
\end{align}
\end{lem}
\begin{proof}
This follows directly from the definition.
\end{proof}
\begin{lem}\label{lem22}
Consider Setting \ref{setni} and assume the approximate Haag duality.
Let $\theta\in \bbR$, $\varphi\in (0,\pi)$,
$\Lambda_i,\Lambda_i'\in\caC_{(\theta,\varphi)}$, $\rho_i,\rho_i'\in \caO_0$ and 
$ \bar V_{\rho_i,\ltj{i}{t_i}}\in \caV_{\rho_i,\ltj{i}{t_i}}$, 
$ \bar V_{\rho_i',\ltjp{i}{t_i'}}\in \caV_{\rho_i',\ltjp{i}{t_i'}}$, 
$t_i,t_i'\ge 0$, $i=1,2$.
Suppose that $\{\Lambda_1,\Lambda_1'\}\perp_{(\theta,\varphi)}
\{\Lambda_2,\Lambda_2'\}$.
Suppse  
\[
S_i^{t_i,t_i'}\in \lmk \trlt{\rho_i}i,\trlpt{\rho'_i}i\rmk, \quad i=1,2
\]
with $\lV S_i^{t_it_i'}\rV\le 1$
 is given
for each  $t_i, t_i'\ge 0$.
Then we have
\begin{align}
\lim_{t_1,t_2\to\infty} \lV S_1^{t_1,t_1'} \otimes
S_2^{t_2, t_2'}-S_2^{t_2t_2'}\otimes S_1^{t_1t_1'}\rV=0.
\end{align}
\end{lem}
\begin{proof}
By the definition (Definition \ref{def36}), there are $\tilde\Lambda_1,\tilde \Lambda_2\in\ctv$ and $\varepsilon>0$
such that $\Lambda_i,\Lambda_i'\subset \tilde \Lambda_i$, $i=1,2$,
and
\begin{align}\label{fat}
\begin{split}
\lmk \tilde \Lambda_2-R_{\lv \arg \tilde \Lambda_2\rv,\varepsilon}
\bm e_{\tilde \Lambda_2}
\rmk_\varepsilon \subset \lmk \tilde \Lambda_1\rmk^c,\quad
\lmk \tilde \Lambda_1-R_{\lv \arg \tilde \Lambda_1\rv, \varepsilon}
\bm e_{\tilde \Lambda_1}
\rmk_\varepsilon \subset \lmk \tilde \Lambda_2\rmk^c,\\
\lmk \tilde\Lambda_1\rmk_{\varepsilon},
\lmk \tilde\Lambda_2\rmk_{\varepsilon}\in \caC_{(\theta,\varphi)},\quad
\arg\lmk \tilde \Lambda_1\rmk_{\varepsilon}\cap \arg\lmk \tilde \Lambda_2\rmk_{\varepsilon}=\emptyset.
\end{split}
\end{align}
Choose $\delta>0$ small enough so that  $(\tilde\Lambda_1)_{\varepsilon+\delta}, (\tilde\Lambda_2)_{\varepsilon+\delta}\in \ctv$
and $\arg\lmk \tilde \Lambda_1\rmk_{\varepsilon+\delta}
\cap \arg\lmk \tilde \Lambda_2\rmk_{\varepsilon+\delta}=\emptyset.$
For any $t\ge 0$, we have
\begin{align}\label{tphos}
\begin{split}
&\Lambda_i+t_i\bm e_{\Lambda_i},\Lambda_i'+t_i'\bm e_{\Lambda_i'}\subset \tilde \Lambda_i
+t\bm e_{\tilde \Lambda_i},\quad i=1,2,\\
&\lmk \tilde \Lambda_2+\frac t 2\bm e_{\tilde \Lambda_2}\rmk_{\varepsilon+\delta}\subset 
\lmk \Lambda_1+t_1\bm e_{\Lambda_1}\rmk^c,\quad 
\lmk \tilde \Lambda_1+\frac t 2\bm e_{\tilde \Lambda_1}\rmk_{\varepsilon+\delta}\subset 
\lmk \Lambda_2+t_2\bm e_{\Lambda_2}\rmk^c
\end{split}
\end{align}
 for $t_1,t_2, t_1',t_2'$ large enough by Lemma \ref{lem1p}.
For such $t_1,t_2, t_1',t_2'$ given for $t\ge 0$, we have
\begin{align}\label{nago}
\begin{split}
&\lmk \tilde \Lambda_2+\lmk t-R_{|\arg\tilde\Lambda_2|,\varepsilon}\rmk
\bm e_{\tilde \Lambda_2}
\rmk_\varepsilon \subset \lmk \tilde \Lambda_2-R_{|\arg\tilde\Lambda_2|,\varepsilon}
\bm e_{\tilde \Lambda_2}
\rmk_\varepsilon\subset
 \lmk \tilde \Lambda_1\rmk^c\subset 
 \lmk \tilde \Lambda_1+t\bm e_{\tilde \Lambda_1}\rmk^c\subset
 \lmk \Lambda_1+t_1\bm e_{\Lambda_1}
\rmk^c\\
&
\lmk \tilde \Lambda_1+\lmk t-R_{|\arg\tilde\Lambda_1|,\varepsilon}\rmk
\bm e_{\tilde \Lambda_1}
\rmk_\varepsilon \subset \lmk \tilde \Lambda_1-R_{|\arg\tilde\Lambda_1|,\varepsilon}
\bm e_{\tilde \Lambda_1}
\rmk_\varepsilon\subset
 \lmk \tilde \Lambda_2\rmk^c\subset 
 \lmk \tilde \Lambda_2+t\bm e_{\tilde \Lambda_2}\rmk^c\subset
 \lmk \Lambda_2+t_2\bm e_{\Lambda_1}
\rmk^c.
\end{split}
\end{align}

Furthermore, for any $t\ge 0$, and $t_1,t_2, t_1',t_2'$ satisfying (\ref{tphos}) for this $t$,
from  Lemma \ref{lem17}, we have
\begin{align}\label{92}
S_2^{t_2, t_2'}\in \pi_0\lmk \caA_{\lmk \lmk \Lambda_2+t_2 \bm e_{ \Lambda_2}\rmk\cup
\lmk \Lambda_2'+t_2' \bm e_{ \Lambda_2'}\rmk\rmk^c}\rmk'
\subset \pi_0(\caA_{(\tilde \Lambda_2+t\bm e_{\tilde \Lambda_2})^c})'
=\amf(\tilde \Lambda_2+t\bm e_{\tilde \Lambda_2}).
\end{align}
Similarly, we have $S_1^{t_1, t_1'}\in\amf(\tilde \Lambda_1+t\bm e_{\tilde \Lambda_1}).
$
From this and the approximate Haag duality, we have 
\begin{align}\label{nana}
\begin{split}
\Ad\lmk  U_{\tilde \Lambda_2+t\bm e_{\tilde \Lambda_2},\varepsilon }^*\rmk\lmk S_2^{t_2, t_2'}\rmk
\in\pi_0\lmk
\al_{\lmk \tilde \Lambda_2+\lmk t-R_{|\arg \tilde\Lambda_2|,\varepsilon}\rmk
\bm e_{\tilde \Lambda_2}
\rmk_\varepsilon}
\rmk''\subset 
\pi_0\lmk\caA_{\lmk \Lambda_1+t_1\bm e_{\Lambda_1}\rmk^c}\rmk'',\\
\Ad\lmk
 U_{\tilde \Lambda_1+t\bm e_{\tilde \Lambda_1},\varepsilon}^*\rmk\lmk S_1^{t_1, t_1'}\rmk
\in\pi_0\lmk
\caA_{\lmk \tilde \Lambda_1+\lmk t-R_{|\arg \tilde\Lambda_1|,\varepsilon}\rmk
\bm e_{\tilde \Lambda_1}
\rmk_\varepsilon}
\rmk''\subset \pza{\lmk \Lambda_2+t_2\bm e_{\Lambda_2}\rmk^c}.
\end{split}
\end{align}
Therefore, by Lemma \ref{lemhoshi} 
($U_{\tilde \Lambda_2+t\bm e_{\tilde \Lambda_2},\varepsilon }\in \btv$),
 Lemma \ref{lem9} (i) ($\trlt{\rho_1}1$ is $\sigma$w-continuous on $\pi_0\lmk
\al_{\lmk \tilde \Lambda_2+\lmk t-R_{|\arg \tilde\Lambda_2|,\varepsilon}\rmk
\bm e_{\tilde \Lambda_2}
\rmk_\varepsilon}
\rmk''$
 ) (d) and (\ref{nana}), we have 
\begin{align}
\begin{split}
&\trlt{\rho_1}1\lmk S_2^{t_2, t_2'}\rmk\\
&=
\Ad\lmk \trlt{\rho_1}1 \lmk  U_{\tilde \Lambda_2+t\bm e_{\tilde \Lambda_2},\varepsilon}
\rmk\rmk\circ
\trlt{\rho_1}1 \lmk\Ad\lmk  U_{\tilde \Lambda_2+t\bm e_{\tilde \Lambda_2},\varepsilon}^*\rmk\lmk S_2^{t_2, t_2'}\rmk\rmk\\
&=\Ad\lmk \trlt{\rho_1}1 \lmk U_{\tilde \Lambda_2+t\bm e_{\tilde \Lambda_2},\varepsilon}\rmk  U^*_{\tilde \Lambda_2+t\bm e_{\tilde \Lambda_2},\varepsilon}\rmk
\lmk S_2^{t_2, t_2'}\rmk.
\end{split}
\end{align}
Similarly, we have
\begin{align}
\begin{split}
&\trlt{\rho_2}2\lmk S_1^{t_1, t_1'}\rmk\\
&=\Ad\lmk \trlt{\rho_2}2 \lmk U_{\tilde \Lambda_1+t\bm e_{\tilde \Lambda_1},\varepsilon}\rmk  U_{\tilde \Lambda_1+t\bm e_{\tilde \Lambda_1},\varepsilon}^*\rmk
\lmk S_1^{t_1, t_1'}\rmk.
\end{split}
\end{align}

From Lemma \ref{niku} applied to $\Lambda_1$, $\Lambda_2$ $t$ replaced by $\Lambda_1+t_1\bm e_{\Lambda_1}$,
$\tilde \Lambda_2+t_2\bm e_{\tilde \Lambda_2}$, $\frac t2$,
  we have
\begin{align}
\lV
\trlt{\rho_1}{1}\lmk U_{\tilde \Lambda_2+t\bm e_{\tilde \Lambda_2},\varepsilon}\rmk  U_{\tilde \Lambda_2+t\bm e_{\tilde \Lambda_2},\varepsilon}^*-
\unit
\rV\le 2 f_{|\arg \tilde \Lambda_2|,\varepsilon,\delta} \lmk \frac{t} 2\rmk,
\end{align}
for any $t\ge 0$ and 
$t_1,t_1',t_2,t_2'\le 0$
satisfying (\ref{tphos})
for this $t$.

Now, for any $\epsilon'>0$, choose $t>0$ such that
$2 f_{|\arg \tilde \Lambda_2|,\varepsilon,\delta} (\frac{t} 2)<\epsilon'$.
Then for any $t_1,t_2,t_1',t_2'$ satisfying (\ref{tphos}) for this $t$
we have 
\begin{align}
\begin{split}
&\lV \trlt{\rho_1}1\lmk S_2^{t_2, t_2'}\rmk-S_2^{t_2, t_2'}\rV\\
&\le 2\lV
\trlt{\rho_1}{1}\lmk U_{\tilde \Lambda_2+t\bm e_{\tilde \Lambda_2},\varepsilon}\rmk  U_{\tilde \Lambda_2+t\bm e_{\tilde \Lambda_2},\varepsilon}^*-
\unit
\rV\le 4 f_{|\arg \tilde \Lambda_2|,\varepsilon,\delta} \lmk \frac{t} 2\rmk<4\epsilon'.
\end{split}
\end{align}

Hence we have
\begin{align}
\begin{split}
&\lim_{t_1,t_2,t_1',t_2'\to\infty}\lV \trlt{\rho_1}1\lmk S_2^{t_2, t_2'}\rmk-S_2^{t_2, t_2'}\rV=0.
\end{split}
\end{align}
Similarly, we have
\begin{align}
\begin{split}
&\lim_{t_1,t_2,t_1',t_2'\to\infty}\lV \trlt{\rho_2}2\lmk S_1^{t_1, t_1'}\rmk-S_1^{t_1, t_1'}\rV=0.
\end{split}
\end{align}
Combinig this with Lemma \ref{lem201} and (\ref{92}), we obtain
\begin{align}
\begin{split}
&\lim_{t_1,t_2,t_1',t_2'\to\infty}\lV
S_1^{t_1t_1'}\otimes S_2^{t_2t_2'}
-S_2^{t_2t_2'}\otimes S_1^{t_1t_1'}\rV\\
&=\lim_{t_1,t_2,t_1',t_2'\to\infty}\lV
S_1^{t_1t_1'}\trlt{\rho_1}1(S_2^{t_2t_2'})-S_2^{t_2t_2'}\trlt{\rho_2}2(S_1^{t_1t_1'})\rV
=0.
\end{split}
\end{align}
\end{proof}
\begin{defn}
Let $\theta\in \bbR$, $\varphi\in (0,\pi)$.
We write $\Lambda_2\leftarrow_{(\theta,\varphi)}\Lambda_1$ if
$\Lambda_1,\Lambda_2\in\ctv$
can be written $\Lambda_i=\Lambda_{\bm p_i, \theta_i,\varphi_i}$, $i=1,2$
with $\theta_i\in\bbR$, $\varphi_i\in (0,\pi)$ satisfying
\begin{align}\label{lao}
\theta+\varphi<\theta_1-\varphi_1<\theta_1+\varphi_1
<\theta_2-\varphi_2<\theta_2+\varphi_2
<\theta-\varphi+ 2\pi.
\end{align}
\end{defn}
\begin{lem}\label{lem25}
Let $\theta\in \bbR$, $\varphi\in (0,\pi)$.
For any $\Lambda_1,\Lambda_2,\Lambda_1',\Lambda_2'\in \ctv$
with $\Lambda_1\perp_{(\theta,\varphi)}\Lambda_2$, 
$\Lambda_1'\perp_{(\theta,\varphi)}\Lambda_2'$,
$\Lambda_2\leftarrow_{(\theta,\varphi)}\Lambda_1$, and $\Lambda_2'\leftarrow_{(\theta,\varphi)}\Lambda_1'$,
there are $\{\Lambda_i^{(j)}\}_{j=0}^4, \{{\Lambda'}_i^{(j)}\}_{j=0}^4\subset \ctv$,
$i=1,2$
such that
\begin{align}
&\{\Lambda_1^{(j)}, \Lambda_1^{(j+1)}\}\perp_{(\theta,\varphi)}\{\Lambda_2^{(j)}, \Lambda_2^{(j+1)}\},\quad
\{{\Lambda'}_1^{(j)}, {\Lambda'}_1^{(j+1)}\}\perp_{(\theta,\varphi)}
\{{\Lambda'}_2^{(j)}, {\Lambda'}_2^{(j+1)}\},\quad
j=0,1,2,3\label{seqseq1}\\
&\{\Lambda_1^{(4)}, {\Lambda'}_1^{(4)}\} \perp_{(\theta,\varphi)}\{\Lambda_2^{(4)}, {\Lambda'}_2^{(4)}\}\label{seqseq2}\\
&\Lambda_i^{(0)}=\Lambda_i,\quad {\Lambda'}_i^{(0)}=\Lambda_i',\quad i=1,2.\label{seqseq3}
\end{align}
\end{lem}
\begin{proof}
Set $\Lambda_i^{(0)}=\Lambda_i$, $ i=1,2$.
Because $\Lambda_1\perp_{(\theta,\varphi)}\Lambda_2$, there are 
$\tilde \Lambda_1,\tilde \Lambda_2\in\ctv$ with 
$\tilde \Lambda_1\perp_{(\theta,\varphi)}\tilde \Lambda_2$
such that $\Lambda_i\subset \tilde\Lambda_i$, $i=1,2$.
Because $\Lambda_2\leftarrow_{(\theta,\varphi)}\Lambda_1$,
we may write $\Lambda_i=\Lambda_{\bm p_i,\theta_i,\varphi_i}$, $i=1,2$ with
$\theta_i,\varphi_i$
satisfying (\ref{lao}).
Set $\Lambda_i^{(1)}:=\Lambda_{\bm p_i,\theta_i,\frac{\varphi_i}4}$, $i=1,2$.
Then because of
$\Lambda_i^{(1)}\subset \Lambda_i\subset \tilde\Lambda_i$, $i=1,2$,
we have $\{\Lambda_1^{(0)}, \Lambda_1^{(1)}\}\perp_{(\theta,\varphi)}\{\Lambda_2^{(0)}, \Lambda_2^{(1)}\}$.

Next, note 
$\bbA_{[\theta_2-\varphi_2,\theta-\varphi+2\pi-\frac{\varphi_2}4]}$,
$\bbA_{[\theta+\varphi+\frac{\varphi_1}{4},\theta_1+\varphi_1]}$,
$\Lambda_2^{(1)}$, $\Lambda_{\bm p_2,\theta-\varphi+2\pi-\frac{\varphi_2}2, \frac{\varphi_2}8}$,
$\Lambda_1^{(1)}$, $\Lambda_{\bm p_1,\theta+\varphi+\frac{\varphi_1}2, \frac{\varphi_1}8}$,
satisfy the conditions of 
$\bbA_1$, $\bbA_2$, $\Lambda_1$, $\Lambda_1'$, $\Lambda_2$, $\Lambda_2'$
in Lemma \ref{lem23}.
Applying the latter lemma, we obtain
$L_1, L_1', L_2, L_2'\ge 0$
such that $\{\Lambda_1^{(2)}, \Lambda_1^{(3)}\}\perp_{(\theta,\varphi)}\{\Lambda_2^{(2)}, \Lambda_2^{(3)}\}$
for
\begin{align}
&\Lambda_i^{(2)}:=\Lambda^{(1)}_i+L_i\bm e_{\Lambda^{(1)}_i},\quad i=1,2\\
&\Lambda_2^{(3)}:=\Lambda_{\bm p_2,\theta-\varphi+2\pi-\frac{\varphi_2}2, \frac{\varphi_2}8}
+L_1'\bm e_{\theta-\varphi+2\pi-\frac{\varphi_2}2},\\
&\Lambda_1^{(3)}:=\Lambda_{\bm p_1,\theta+\varphi+\frac{\varphi_1}2, \frac{\varphi_1}8}
+L_2' \bm e_{\theta+\varphi+\frac{\varphi_1}2}.
\end{align}
Note also that from $\Lambda_i^{(2)}\subset \Lambda_i^{(1)}\subset \tilde \Lambda_i$,
we have $\{\Lambda_1^{(1)}, \Lambda_1^{(2)}\}\perp_{(\theta,\varphi)}
\{\Lambda_2^{(1)}, \Lambda_2^{(2)}\}$.

Hence we obtain 
 $\{\Lambda_i^{(j)}\}_{j=0}^3\subset \ctv$,
$i=1,2$
such that $\Lambda_i^{(0)}=\Lambda_i$, $i=1,2$,
$\{\Lambda_1^{(j)}, \Lambda_1^{(j+1)}\}\perp_{(\theta,\varphi)}\{\Lambda_2^{(j)}, \Lambda_2^{(j+1)}\}
$$j=0,1,2$
and
\begin{align}
\arg\Lambda_2^{(3)}\subset \bbA_{(\theta+\frac{11}8\pi-\frac 3 8\varphi, \theta+2\pi-\varphi)},\quad
\arg\Lambda_1^{(3)}\subset \bbA_{(\theta+\varphi,\theta+\frac 58\pi+\frac 38 \varphi)}.
\end{align}
By the same argument  we obtain 
 $\{{\Lambda'}_i^{(j)}\}_{j=0}^3\subset \ctv$,
$i=1,2$
such that ${\Lambda'}_i^{(0)}=\Lambda_i$, $i=1,2$,
$\{{\Lambda'}_1^{(j)}, {\Lambda'}_1^{(j+1)}\}\perp_{(\theta,\varphi)}\{{\Lambda'}_2^{(j)}, {\Lambda'}_2^{(j+1)}\}
$, $j=0,1,2$,
and
\begin{align}
\arg{\Lambda'}_2^{(3)}\subset \bbA_{(\theta+\frac{11}8\pi-\frac 3 8\varphi, \theta+2\pi-\varphi)},\quad
\arg{\Lambda'}_1^{(3)}\subset \bbA_{(\theta+\varphi,\theta+\frac 58\pi+\frac 38 \varphi)}.
\end{align}

Choose $\delta>0$ small enough so that
\begin{align}
\arg\Lambda_2^{(3)},\arg{\Lambda'}_2^{(3)}\subset \bbA_{[\theta+\frac{11}8\pi-\frac 3 8\varphi, \theta+2\pi-\varphi-\delta]},\quad
\arg\Lambda_1^{(3)}, \arg{\Lambda'}_1^{(3)}\subset \bbA_{[\theta+\varphi+\delta,\theta+\frac 58\pi+\frac 38 \varphi]},
\end{align}
and $\theta+\frac{11}8\pi-\frac 3 8\varphi<\theta+2\pi-\varphi-\delta$,
$\theta+\varphi+\delta<\theta+\frac 58\pi+\frac 38 \varphi$.
Note 
$ \bbA_{[\theta+\frac{11}8\pi-\frac 3 8\varphi, \theta+2\pi-\varphi-\delta]}$,
$\bbA_{[\theta+\varphi+\delta,\theta+\frac 58\pi+\frac 38 \varphi]}$,
$\Lambda_2^{(3)}$, ${\Lambda'}_2^{(3)}$
$\Lambda_1^{(3)}$, ${\Lambda'}_1^{(3)}$, 
satisfy the conditions of 
$\bbA_1$, $\bbA_2$, $\Lambda_1$, $\Lambda_1'$, $\Lambda_2$, $\Lambda_2'$
in Lemma \ref{lem23}.
Applying the latter lemma, we obtain
$s_1, s_1', s_2, s_2'\ge 0$
such that 
 $\{\Lambda_1^{(4)}, {\Lambda'}_1^{(4)}\}\perp_{(\theta,\varphi)}\{\Lambda_2^{(4)}, {\Lambda'}_2^{(4)}\}$
for
\begin{align}
\Lambda_i^{(4)}:=\Lambda^{(3)}_i+s_i\bm e_{\Lambda^{(3)}_i},\quad
{\Lambda'}_i^{(4)}:={\Lambda'}^{(3)}_i+s'_i\bm e_{{\Lambda'}^{(3)}_i}\quad i=1,2.
\end{align}
Note that $\Lambda_i^{(4)}, {\Lambda'}_i^{(4)}\in \ctv$.
From $\Lambda_i^{(4)}\subset \Lambda_i^{(3)}$ while $\{\Lambda_1^{(3)}\}\perp_{(\theta,\varphi))}\{\Lambda_2^{(3)}\}$,
we have $\{\Lambda_1^{(3)}, \Lambda_1^{(4)}\}\perp_{(\theta,\varphi)}
\{\Lambda_2^{(3)}, \Lambda_2^{(4)}\}$.
Similarly, we have $\{{\Lambda'}_1^{(3)}, {\Lambda'}_1^{(4)}\}\perp_{(\theta,\varphi)}
\{{\Lambda'}_2^{(3)}, {\Lambda'}_2^{(4)}\}$.
This completes the proof.
\end{proof}
\begin{lem}\label{lem26}
Consider Setting \ref{setni} and assume the approximate Haag duality.
Let $\theta\in \bbR$, $\varphi\in (0,\pi)$, and 
$\Lambda_1,\Lambda_2,\Lambda_1',\Lambda_2'\in\ctv$
satisfying
$\Lambda_1\perp_{(\theta,\varphi)} \Lambda_2$, 
$\Lambda'_1\perp_{(\theta,\varphi)} \Lambda'_2$
and $\Lambda_2\leftarrow_{(\theta,\varphi)}\Lambda_1$,
$\Lambda_2'\leftarrow_{(\theta,\varphi)}\Lambda_1'$.
Let $\rho_i,\rho_i'\in\caO_0$,
$V_{\rho_i\ltj{i}{t_i}}\in\caV_{\rho_i\ltj{i}{t_i}}$,
$V_{\rho_i'\ltjp{i}{t_i}}\in\caV_{\rho_i'\ltjp{i}{t_i}}$,
and 
\[R_i^{t_i,t_i'}\in \lmk \trlt{\rho_i}i, \trlpt{\rho_i'}i\rmk,\quad
\lV R_i^{t_i,t_i'}\rV\le 1
\]
given for each $t_i,t_i'\ge 0$
and $i=1,2$.
Then we have
\begin{align}
\lim_{t_1,t_2,t_1't_2'\to\infty}
\lV R_1^{t_1t_1'}\otimes R_2^{t_2t_2'}-R_2^{t_2t_2'}\otimes R_1^{t_1t_1'}\rV=0.
\end{align}
\end{lem}
\begin{proof}
By Lemma \ref{lem25}, there are 
$\{\Lambda_i^{(j)}\}_{j=0}^4, \{{\Lambda'}_i^{(j)}\}_{j=0}^4\subset \ctv$,
$i=1,2$ satisfying (\ref{seqseq1}), (\ref{seqseq2}), (\ref{seqseq3}).
For $t_{ij}, t'_{ij}\ge 0$, $i=1,2$, $j=0,\ldots 4$, we denote the set $\bm t:=\lmk (t_{ij}), (t'_{ij})\rmk$
and $\bm t\to \infty$ means $t_{i,j}, t_{ij}'\to\infty$ for all $i,j$.
 Set 
 \begin{align}
 W_i^{(j), \bm t}:=\bar V_{\rho_i\Lambda_i^{(j+1)}+t_{i,j+1}\bm e_{\Lambda_i^{(j+1)}}} {\lmk \bar V_{\rho_i\Lambda_i^{(j)}+t_{ij}\bm e_{\Lambda_i^{(j)}}}\rmk^*},\quad\text{
 and}\quad
 {W'}_i^{(j),\bm t}:=\bar V_{\rho_i'{\Lambda'}_i^{(j+1)}+t'_{i,j+1}\bm e_{{\Lambda'}_i^{(j+1)}}} {\lmk \bar V_{\rho'_i{\Lambda'}_i^{(j)}+t'_{ij}\bm e_{{\Lambda'}_i^{(j)}}}\rmk^*}.
 \end{align} 
 We have 
 \begin{align}
 \begin{split}
 W_i^{(j), \bm t}\in\lmk  \trltj{\rho_i}{i}{j}{ij}, \trltj{\rho_i}{i}{j+1}{i,j+1}\rmk,\\
 {W'}_i^{(j), \bm t}\in\lmk  \trltjp{\rho_i'}{i}{j}{ij}, \trltjp{\rho_i'}{i}{j+1}{i,j+1}\rmk
\end{split} 
  \end{align}
  for $i=1,2$ and $j=0,1,2,3$ from Lemma \ref{lem16}.
By Lemma \ref{lem22}, we have
\begin{align}\label{wcom}
\lim_{\bm t\to\infty}\lV W_1^{(j),\bm t}\otimes W_2^{(j),\bm t}-
W_2^{(j)\bm t}\otimes W_1^{(j)\bm t}\rV=0,\quad
\lim_{\bm t\to\infty }\lV {W'}_1^{(j),\bm t}\otimes {W'}_2^{(j),\bm t}-
{W'}_2^{(j)\bm t}\otimes {W'}_1^{(j)\bm t}\rV=0
 \quad j=0,1,2,3.
\end{align}
Note that 
\begin{align}
\begin{split}
S_i^{\bm t}:=
{W'}_i^{(3)\bm t} {W'}_i^{(2)\bm t} {W'}_i^{(1)\bm t}{W'}_i^{(0)\bm t}
R_i^{(t_{i_0}t_{i_0}')}
\lmk W_i^{(3)\bm t} W_i^{(2)\bm t} W_i^{(1)\bm t}W_i^{(0)\bm t}\rmk^*\\
\in \lmk
\trltj {\rho_i}i4{i4},\trltjp {\rho_i'}i4{i4}
\rmk.
\end{split}
\end{align}
Because $\{\Lambda_1^{(4)}, {\Lambda'}_1^{(4)}\} \perp_{(\theta,\varphi)}\{\Lambda_2^{(4)}, {\Lambda'}_2^{(4)}\}$, we have
\begin{align}\label{scom}
\lim_{\bm t\to\infty}
\lV
S_1^{\bm t}\otimes S_2^{\bm t}-S_2^{\bm t}\otimes S_1^{\bm t}
\rV=0
\end{align}
by Lemma \ref{lem22}.

By Lemma \ref{lem19} and (\ref{wcom}), (\ref{scom}) we complete the proof.
\end{proof}
\begin{lem}\label{lem28}
Consider Setting \ref{setni} and assume the approximate Haag duality.
Let $\theta\in \bbR$, $\varphi\in (0,\pi)$, $\rho,\sigma\in \caO_0$,
$\Lambda_0,\Lambda_1,\Lambda_2, {\Lambda'}_1,{\Lambda'}_2\in \ctv$,
and $\bar V_{\rho\Lambda_i+t_i \bm e_{\Lambda_i}}\in \caV_{\rho\Lambda_i+t_i \bm e_{\Lambda_i}}$, $\bar V_{\sigma\Lambda_i+t_i \bm e_{\Lambda_i}}\in \caV_{\sigma\Lambda_i+t_i \bm e_{\Lambda_i}}$,
$\bar V_{\rho{\Lambda'}_i+t'_i \bm e_{\Lambda'_i}}\in \caV_{\rho{\Lambda'}_i}+t'_i \bm e_{\Lambda'_i}$, $\bar V_{\sigma{\Lambda'}_i+t'_i \bm e_{\Lambda'_i}}\in \caV_{\sigma{\Lambda'}_i+t'_i \bm e_{\Lambda'_i}}$,
for each $i=1,2$, $t_i, t_i'\ge 0$ and
$\bar V_{\rho\Lambda_0}\in \caV_{\rho\Lambda_0}$,
$\bar V_{\sigma\Lambda_0}\in \caV_{\sigma\Lambda_0}$.
Set $\bm t:=(t_i)_{i=1,2}, \bm t':=(t_i')_{i=1,2}$ and 
 $W_{\rho\Lambda_0\Lambda_1}^{\bm t}:=\bar V_{\rho\Lambda_1+t_1 \bm e_{\Lambda_1}}\bar V_{\rho\Lambda_0}^*$,
$W_{\sigma\Lambda_0\Lambda_2}^{\bm t}:=\bar V_{\sigma\Lambda_2+t_2\bm e_{\Lambda_2}}\bar V_{\sigma\Lambda_0}^*$,
 $W_{\rho{\Lambda}_0{\Lambda'}_1}^{\bm t'}:=\bar V_{\rho{\Lambda'}_1+t_1' \bm e_{\Lambda'_1}}\bar V_{\rho\Lambda_0}^*$,
$W_{\sigma\Lambda_0\Lambda_2'}^{\bm t'}:=\bar V_{\sigma\Lambda'_2+t_2'\bm e_{\Lambda'_2}}\bar V_{\sigma\Lambda_0}^*$.
If $\Lambda_1\perp_{(\theta,\varphi)} \Lambda_2$, $\Lambda_2\leftarrow_{(\theta,\varphi)}\Lambda_1$ and
${\Lambda'}_1\perp_{(\theta,\varphi)} {\Lambda'}_2$, ${\Lambda'}_2\leftarrow_{(\theta,\varphi)}{\Lambda'}_1$,
then we have 
\begin{align}
\lim_{\bm t,\bm t'\to\infty}\lV
\lmk W_{\sigma\Lambda_0\Lambda_2}^{\bm t}\otimes W_{\rho\Lambda_0\Lambda_1}^{\bm t}\rmk^*
\lmk
 W_{\rho\Lambda_0\Lambda_1}^{\bm t}\otimes  W_{\sigma\Lambda_0\Lambda_2}^{\bm t}
\rmk
-\lmk W_{\sigma{\Lambda}_0{\Lambda'}_2}^{\bm t'}\otimes W_{\rho{\Lambda}_0{\Lambda'}_1}^{\bm t'}\rmk^*
\lmk
 W_{\rho{\Lambda}_0{\Lambda'}_1}^{\bm t'}\otimes  W_{\sigma{\Lambda}_0{\Lambda'}_2}^{\bm t'}
\rmk\rV=0.
\end{align}
Here, $\bm t,\bm t'\to\infty$ means 
$t_i, t_i'\to\infty$ for each $i$.
\end{lem}
\begin{defn}\label{epdef}
From Lemma \ref{lem28}, for any $\theta\in \bbR$, $\varphi\in (0,\pi)$,
$\rho,\sigma\in \caO_0$ and $\Lambda_0\in\ctv$
we may define
\begin{align}
\epsilon_+^{(\Lambda_0)}(\rho,\sigma):=
\lim_{\bm t\to\infty}
\lmk W_{\sigma\Lambda_0\Lambda_2}^{\bm t}\otimes W_{\rho\Lambda_0\Lambda_1}^{\bm t}\rmk^*
\lmk
W_{\rho\Lambda_0\Lambda_1}^{\bm t}\otimes  W_{\sigma\Lambda_0\Lambda_2}^{\bm t}
\rmk\in \btv,
\end{align}
independent of the choice of $\Lambda_1,\Lambda_2\in\ctv$,
satisfying $\Lambda_1\perp_{(\theta,\varphi)} \Lambda_2$, $\Lambda_2\leftarrow_{(\theta,\varphi)}\Lambda_1$,
and $V_{\rho\ltj{i}{t_i}}\in \caV_{\rho\ltj{i}{t_i}}$, $V_{\sigma\ltj{i}{t_i}}\in \caV_{\sigma\ltj{i}{t_i}}$.
\end{defn}
\begin{proof}
Set  $W_{\rho\Lambda_1\Lambda_1'}^{\bm t, \bm t'}:=\bar V_{\rho\Lambda_1'+t_1' \bm e_{\Lambda'_1}}\bar V_{\rho\Lambda_1+t_1 \bm e_{\Lambda_1}}^*$,
$W_{\sigma\Lambda_2\Lambda_2'}^{\bm t, \bm t'}:=\bar V_{\sigma\Lambda_2'+t_2' \bm e_{\Lambda'_2}}\bar V_{\sigma\Lambda_2+t_2 \bm e_{\Lambda_2}}^*$.
%
From Lemma \ref{lem16},
\begin{align}
\begin{split}
&W^{\bm t}_{\rho\Lambda_0\Lambda_1}\in \lmk \trl{\rho}0,\trlt{\rho}1\rmk,\quad
W^{\bm t}_{\sigma\Lambda_0\Lambda_2}\in \lmk \trl{\sigma}0,\trlt{\sigma}2\rmk,\\
&W^{\bm t'}_{\rho{\Lambda}_0{\Lambda'}_1}\in \lmk \trl{\rho}0,\trlpt{\rho}1\rmk,\quad
W^{\bm t'}_{\sigma{\Lambda}_0{\Lambda'}_2}\in \lmk \trl{\sigma}0,\trlpt{\sigma}2\rmk,\\
&W^{\bm t,\bm t'}_{\rho\Lambda_1\Lambda_1'}\in \lmk \trlt{\rho}1,\trlpt{\rho}1\rmk,\\
&W^{\bm t,\bm t'}_{\sigma\Lambda_2\Lambda_2'}\in \lmk \trlt{\sigma}2,\trlpt{\sigma}2\rmk.
\end{split}
\end{align}
By Lemma \ref{lem19},
\begin{align}\label{wtachi}
\begin{split}
&\lmk W^{\bm t'}_{\sigma{\Lambda}_0{\Lambda'}_2}\otimes W^{\bm t'}_{\rho{\Lambda}_0{\Lambda'}_1}\rmk^*
\lmk
 W^{\bm t'}_{\rho{\Lambda}_0{\Lambda'}_1}\otimes  W^{\bm t'}_{\sigma{\Lambda}_0{\Lambda'}_2}
\rmk\\
&=\lmk W^{\bm t,\bm t'}_{\sigma{\Lambda}_2{\Lambda'}_2}
W^{\bm t}_{\sigma{\Lambda}_0{\Lambda}_2}\otimes W^{\bm t,\bm t'}_{\rho{\Lambda}_1{\Lambda'}_1}W^{\bm t}_{\rho{\Lambda}_0{\Lambda}_1}\rmk^*
\lmk
 W^{\bm t,\bm t'}_{\rho{\Lambda}_1{\Lambda'}_1}W^{\bm t}_{\rho{\Lambda}_0{\Lambda}_1}
 \otimes  W^{\bm t,\bm t'}_{\sigma{\Lambda}_2{\Lambda'}_2} W^{\bm t}_{\sigma{\Lambda}_0{\Lambda}_2}
\rmk\\
&=
\lmk W^{\bm t}_{\sigma{\Lambda}_0{\Lambda}_2} \otimes W^{\bm t}_{\rho{\Lambda}_0{\Lambda}_1}\rmk^*
\lmk W^{\bm t,\bm t'}_{\sigma{\Lambda}_2{\Lambda'}_2}
\otimes W^{\bm t,\bm t'}_{\rho{\Lambda}_1{\Lambda'}_1}\rmk^*
\lmk
 W^{\bm t,\bm t'}_{\rho{\Lambda}_1{\Lambda'}_1}
 \otimes  W^{\bm t, \bm t'}_{\sigma{\Lambda}_2{\Lambda'}_2}
\rmk
\lmk
W^{\bm t}_{\rho{\Lambda}_0{\Lambda}_1}
 \otimes  {W^{\bm t}}_{\sigma{\Lambda}_0{\Lambda}_2}
\rmk.
\end{split}
\end{align}
Note that 
setting $\rho_1:=\rho$, $\rho_1':=\rho$, $\rho_2:=\sigma$, $\rho_2':=\sigma$
and $R_1:=W^{\bm t, \bm t'}_{\rho\Lambda_1\Lambda_1'}$
$R_2:=W^{\bm t,\bm t'}_{\sigma\Lambda_2\Lambda_2'}$
in Lemma \ref{lem26}, we obtain
\begin{align}
\lim_{\bm t,\bm t'\to\infty}
\lV
\lmk W^{\bm t, \bm t'}_{\sigma{\Lambda}_2{\Lambda'}_2}
\otimes W^{\bm t, \bm t'}_{\rho{\Lambda}_1{\Lambda'}_1}\rmk^*
\lmk
 W^{\bm t,\bm t'}_{\rho{\Lambda}_1{\Lambda'}_1}
 \otimes  W^{\bm t,\bm t'}_{\sigma{\Lambda}_2{\Lambda'}_2}
\rmk
-\unit\rV=0.
\end{align}
Substituting this to (\ref{wtachi}), we prove the claim.
\end{proof}
\begin{lem}\label{lem30}
Consider Setting \ref{setni} and assume the approximate Haag duality.
Let $\theta\in \bbR$, $\varphi\in (0,\pi)$.
For any $\rho,\sigma\in \caO_0$ and $\Lambda_0\in\ctv$
we have
\begin{align}
\epsilon_+^{(\Lambda_0)}(\rho,\sigma)\in 
\lmk\trl{\rho}0\trl{\sigma}0,\trl{\sigma}0\trl{\rho}0\rmk.
\end{align}
\end{lem}
\begin{proof}
Let 
 $\Lambda_1,\Lambda_2\in\ctv$,
be cones satisfying $\Lambda_1\perp_{(\theta,\varphi)} \Lambda_2$, $\Lambda_2\leftarrow_{(\theta,\varphi)}\Lambda_1$.
We use the notation in Lemma \ref{lem28}.
Note that 
\begin{align}
 \begin{split}
 &{W^{\bm t}}_{\rho\Lambda_0\Lambda_1}\otimes  {W^{\bm t}}_{\sigma\Lambda_0\Lambda_2}\in 
 \lmk\trl{\rho}0\trl{\sigma}0,\trlt{\rho}1\trlt{\sigma}2\rmk,\\
 & \lmk {W^{\bm t}}_{\sigma\Lambda_0\Lambda_2}\otimes {W^{\bm t}}_{\rho\Lambda_0\Lambda_1}\rmk^*
 \in  \lmk\trlt{\sigma}2\trlt{\rho}1,\trl{\sigma}0\trl{\rho}0\rmk.
 \end{split}
\end{align}
Note also that by Lemma \ref{lem9} (d), we have
 for any $A\in\caA_{loc}$,
 we have
 \begin{align}\label{tcom}
\lim_{\bm t\to\infty}
\lV 
\begin{gathered}
\trlt{\rho}{1}\trlt{\sigma}{2}\lmk\pi_0(A)\rmk\\
-\trlt{\sigma}{2}\trlt{\rho}{1}\lmk\pi_0(A)\rmk
\end{gathered}
\rV=0,
\end{align}
because the support of $A$ gets out of $\Lambda_i +t_i\bm e_{\Lambda_i}$
eventually. 
For any $A\in\caA_{loc}$
we have
\begin{align}
\begin{split}
&\epsilon_+^{(\Lambda_0)}(\rho,\sigma)\trl{\rho}{0}\trl{\sigma}{0}\lmk\pi_0(A)\rmk\\
&=\lim_{\bm t\to\infty}
\lmk {W^{\bm t}}_{\sigma\Lambda_0\Lambda_2}\otimes {W^{\bm t}}_{\rho\Lambda_0\Lambda_1}\rmk^*
\lmk
 {W^{\bm t}}_{\rho\Lambda_0\Lambda_1}\otimes  {W^{\bm t}}_{\sigma\Lambda_0\Lambda_2}
\rmk\trl{\rho}{0}\trl{\sigma}{0}
\lmk\pi_0(A)\rmk\\
&=\lim_{\bm t\to\infty}
\lmk {W^{\bm t}}_{\sigma\Lambda_0\Lambda_2}\otimes {W^{\bm t}}_{\rho\Lambda_0\Lambda_1}\rmk^*
\trlt{\rho}{1}\trlt{\sigma}{2}
\lmk\pi_0(A)\rmk
\lmk
 {W^{\bm t}}_{\rho\Lambda_0\Lambda_1}\otimes  {W^{\bm t}}_{\sigma\Lambda_0\Lambda_2}
\rmk,
\end{split}
\end{align}
and
\begin{align}
\begin{split}
&\trl{\sigma}{0}\trl{\rho}{0}
\lmk\pi_0(A)\rmk
\epsilon_+^{(\Lambda_0)}(\rho,\sigma)\\
&=\lim_{\bm t\to\infty}\trl{\sigma}{0}\trl{\rho}{0}
\lmk\pi_0(A)\rmk
\lmk {W^{\bm t}}_{\sigma\Lambda_0\Lambda_2}\otimes {W^{\bm t}}_{\rho\Lambda_0\Lambda_1}\rmk^*
\lmk
 {W^{\bm t}}_{\rho\Lambda_0\Lambda_1}\otimes  {W^{\bm t}}_{\sigma\Lambda_0\Lambda_2}
\rmk\\
&=\lim_{\bm t\to\infty}\lmk {W^{\bm t}}_{\sigma\Lambda_0\Lambda_2}\otimes {W^{\bm t}}_{\rho\Lambda_0\Lambda_1}\rmk^*
\trlt{\sigma}{2}\trlt{\rho}{1}
\lmk\pi_0(A)\rmk
\lmk
 {W^{\bm t}}_{\rho\Lambda_0\Lambda_1}\otimes  {W^{\bm t}}_{\sigma\Lambda_0\Lambda_2}
\rmk.
\end{split}
\end{align}
From this and  (\ref{tcom}), we obtain 
\begin{align}
\epsilon_+^{(\Lambda_0)}(\rho,\sigma)\trl{\rho}{0}\trl{\sigma}{0}\lmk\pi_0(A)\rmk
=\trl{\sigma}{0}\trl{\rho}{0}
\lmk\pi_0(A)\rmk
\epsilon_+^{(\Lambda_0)}(\rho,\sigma)
\end{align}
for any $A\in \caA_{\rm loc}$.
 Using the $\sigma$w-continuity of $\trl{\rho}{0}\trl{\sigma}{0}$,
 $\trl{\sigma}{0}\trl{\rho}{0}$ (Lemma \ref{ccon}) and Lemma \ref{lemhoshi}
 we can replace $\pi_0(A)$
 by general $x\in \btv$, proving the claim.
\end{proof}
\begin{lem}\label{lem31}
Consider Setting \ref{setni} and assume the approximate Haag duality.
Let $\theta\in \bbR$, $\varphi\in (0,\pi)$.
For any $\rho,\sigma,\rho',\sigma'\in \caO_0$, $\Lambda_0\in\ctv$,
and $R\in \lmk \trl{\rho} 0, \trl{\rho'}0\rmk$,
$S\in \lmk \trl{\sigma} 0, \trl{\sigma'}0\rmk$,
we have
\begin{align}
\epsilon_+^{(\Lambda_0)}(\rho',\sigma')(R\otimes S)
=(S\otimes R) \epsilon_+^{(\Lambda_0)}(\rho,\sigma).
\end{align}
\end{lem}
\begin{proof}
Let 
 $\Lambda_1,\Lambda_2\in\ctv$,
be cones satisfying $\Lambda_2\perp_{(\theta,\varphi)} \Lambda_1$, $\Lambda_2\leftarrow_{(\theta,\varphi)}\Lambda_1$.
We use the notation in Lemma \ref{lem28}.
Note that
\begin{align}
\begin{split}
&{W^{\bm t}}_{\rho'\Lambda_0\Lambda_1}R {W^{\bm t}}_{\rho\Lambda_0\Lambda_1}^*\in \lmk
\trlt{\rho}1,\trlt{\rho'}1\rmk,\\
&{W^{\bm t}}_{\sigma'\Lambda_0\Lambda_2}S{W^{\bm t}}_{\sigma\Lambda_0\Lambda_2}^*\in \lmk
\trlt{\sigma}2,\trlt{\sigma'}2\rmk.
\end{split}
\end{align}
Because $\Lambda_1\perp_{(\theta,\varphi)} \Lambda_2$ by Lemma \ref{lem26}, we have
\begin{align}
\lim_{\bm t\to\infty}\lV {W^{\bm t}}_{\rho'\Lambda_0\Lambda_1}R {W^{\bm t}}_{\rho\Lambda_0\Lambda_1}^*
\otimes {W^{\bm t}}_{\sigma'\Lambda_0\Lambda_2}S{W^{\bm t}}_{\sigma\Lambda_0\Lambda_2}^*
-{W^{\bm t}}_{\sigma'\Lambda_0\Lambda_2}S{W^{\bm t}}_{\sigma\Lambda_0\Lambda_2}^*
\otimes {W^{\bm t}}_{\rho'\Lambda_0\Lambda_1}R {W^{\bm t}}_{\rho\Lambda_0\Lambda_1}^*\rV=0.
\end{align}
From this and Lemma \ref{lem19}, we have
\begin{align}
\begin{split}
&
\lim_{\bm t\to\infty}\lV \lmk {W^{\bm t}}_{\rho'\Lambda_0\Lambda_1}\otimes {W^{\bm t}}_{\sigma'\Lambda_0\Lambda_2}\rmk
(R\otimes S) 
\lmk {W^{\bm t}}_{\rho\Lambda_0\Lambda_1}\otimes {W^{\bm t}}_{\sigma\Lambda_0\Lambda_2}\rmk^*
-\lmk  {W^{\bm t}}_{\sigma'\Lambda_0\Lambda_2}\otimes
{W^{\bm t}}_{\rho'\Lambda_0\Lambda_1}\rmk
(S\otimes R)
\lmk {W^{\bm t}}_{\sigma\Lambda_0\Lambda_2}\otimes {W^{\bm t}}_{\rho\Lambda_0\Lambda_1} \rmk^*
\rV
\\
&
=\lim_{\bm t\to\infty}\lV
{W^{\bm t}}_{\rho'\Lambda_0\Lambda_1}R {W^{\bm t}}_{\rho\Lambda_0\Lambda_1}^*
\otimes {W^{\bm t}}_{\sigma'\Lambda_0\Lambda_2}S{W^{\bm t}}_{\sigma\Lambda_0\Lambda_2}^*
-{W^{\bm t}}_{\sigma'\Lambda_0\Lambda_2}S{W^{\bm t}}_{\sigma\Lambda_0\Lambda_2}^*
\otimes {W^{\bm t}}_{\rho'\Lambda_0\Lambda_1}R {W^{\bm t}}_{\rho\Lambda_0\Lambda_1}^*
\rV=0.
\end{split}
\end{align}
This proves the Lemma.
\end{proof}
\begin{lem}\label{lem32}
Consider Setting \ref{setni} and assume the approximate Haag duality.
Let $\theta\in \bbR$, $\varphi\in (0,\pi)$, and
 $\rho,\sigma\in \caO_0$,
$\Lambda_0,\Lambda_1\in \ctv$,
and $V_{\rho\Lambda_i}\in \caV_{\rho\Lambda_i}$, $V_{\sigma\Lambda_i}\in \caV_{\sigma\Lambda_i}$,
$i=0,1$.
Set $W_{\rho\Lambda_0\Lambda_1}:=V_{\rho\Lambda_1}V_{\rho\Lambda_0}^*$,
$W_{\sigma\Lambda_0\Lambda_1}:=V_{\sigma\Lambda_1}V_{\sigma\Lambda_0}^*$.
Fix some $\{\bar V_{\eta,\Lambda_i}\in \caV_{\eta,\Lambda_i}\mid \eta\in\caO_0\}$,
$i=0,1$.
Then we have
\begin{align}
&W_{\rho\Lambda_0\Lambda_1}\otimes W_{\sigma\Lambda_0\Lambda_1}
\in \caV_{\rho\;\;\circ_{\comp{}{0}}\;\;\sigma,\Lambda_1},\label{w321}\\
&1\in \caV_{\rho\;\;\circ_{\comp{}{0}}\;\;\sigma,\Lambda_0}\label{w322}.
\end{align}
\end{lem}
\begin{rem}\label{lem32rem}
By this Lemma and Definition \ref{epdef} and Lemma \ref{lem28},
we may set
\[
W^{\bm t}_{\rho\;\;\circ_{\comp{}{0}}\;\;\sigma,\Lambda_0,\Lambda_1}
:=W_{\rho\Lambda_0\Lambda_1}^{\bm t}\otimes W_{\sigma\Lambda_0\Lambda_1}^{\bm t},
\]
in the definition of $\epsilon_+^{(\Lambda_0)}\lmk \rho\;\;\circ_{\comp{}{0}}\;\;\sigma,\cdot\rmk$
and $\epsilon_+^{(\Lambda_0)}\lmk \cdot, \rho\;\;\circ_{\comp{}{0}}\;\;\sigma\rmk$.
\end{rem}
\begin{proof}
For any $A\in\caA_{\Lambda_1^c}$, we have
\begin{align}
\begin{split}
&\Ad\lmk W_{\rho\Lambda_0\Lambda_1}\otimes W_{\sigma\Lambda_0\Lambda_1}\rmk
\rho\;\;\circ_{\comp{}{0}}\;\;\sigma(A)\\
&=\Ad\lmk W_{\rho\Lambda_0\Lambda_1}\trl{\rho}0\lmk W_{\sigma\Lambda_0\Lambda_1}\rmk\rmk
\trl{\rho}0\trl{\sigma}0\lmk\pi_0(A)\rmk\\
&=\Ad\lmk W_{\rho\Lambda_0\Lambda_1}\rmk \trl{\rho}0
\Ad\lmk W_{\sigma\Lambda_0\Lambda_1}\rmk\trl{\sigma}0\lmk\pi_0(A)\rmk\\
&=\trl{\rho}1\trl{\sigma}1\lmk \pi_0(A)\rmk
=\pi_0(A)
\end{split}
\end{align}
In the third equality, we used Lemma \ref{lem10} (iii).
In the fourth equality, we used Lemma \ref{lem9} (d).
Hence we obtain (\ref{w321}).
The latter one (\ref{w322})
is proven in Lemma \ref{lem14}.
\end{proof}

\begin{lem}\label{lem416}
Consider Setting \ref{setni} and assume the approximate Haag duality.
Let $\theta\in \bbR$, $\varphi\in (0,\pi)$, $\Lambda_0\in \ctv$ and
fix $\{\bar V_{\eta\Lambda_0}\in \caV_{\eta\Lambda_0}\mid \eta\in\caO_0\}$.
For any $\rho,\sigma,\tau\in \caO_0$,
we have
\begin{align}
&\epsilon_+^{(\Lambda_0)}\lmk \rho\;\;\circ_{\comp {}0}\sigma, \tau\rmk
=\lmk
\epsilon_+^{(\Lambda_0)}\lmk\rho,\tau\rmk\otimes 1_{\trl{\sigma}0}
\rmk
\lmk
 1_{\trl{\rho}0}\otimes \epsilon_+^{(\Lambda_0)}\lmk\sigma,\tau\rmk
\rmk,\label{fstb}\\
&\epsilon_+^{(\Lambda_0)}\lmk \rho, \sigma\;\;\circ_{\comp {}0}\tau\rmk
=\lmk
 1_{\trl{\sigma}0}\otimes \epsilon_+^{(\Lambda_0)}\lmk\rho,\tau\rmk
\rmk
\lmk
\epsilon_+^{(\Lambda_0)}\lmk\rho,\sigma\rmk\otimes 1_{\trl{\tau}0}
\rmk.\label{sndb}
\end{align}
\end{lem}
\begin{proof}
First we prove (\ref{fstb}).
We may take 
$\Lambda_1,\Lambda_2\in\ctv$,
satisfying $\Lambda_2\leftarrow_{(\theta,\varphi)}\Lambda_1$
so that there are some $\tilde \Lambda_1,\tilde\Lambda_2\in\ctv$ such that
 \begin{align}
 \begin{split}
 &(\Lambda_1)_\varepsilon, \Lambda_0\subset \tilde \Lambda_1,\quad\Lambda_2\subset \tilde \Lambda_2\\
&\arg\lmk\tilde \Lambda_{1}\rmk_{2\varepsilon+\delta+\epsilon+\epsilon'} \cap
\arg \lmk \tilde\Lambda_2\rmk_{2\varepsilon+\delta+\epsilon+\epsilon'}=\emptyset,\\
&\lmk \tilde \Lambda_1-R_{|\arg\tilde \Lambda_1|, \varepsilon}\bm e_{\tilde \Lambda_1}\rmk_{\varepsilon}
\subset \tilde \Lambda_2^c,\quad 
 \lmk \tilde \Lambda_2-R_{|\arg\tilde \Lambda_2|, \varepsilon}\bm e_{\tilde \Lambda_2}\rmk_{\varepsilon}
\subset \tilde \Lambda_1^c,\\
&\lmk \tilde \Lambda_1\rmk_{2(\varepsilon+\epsilon+\epsilon'+\delta)},
\lmk \tilde \Lambda_2\rmk_{2(\varepsilon+\epsilon+\epsilon'+\delta)}\in\ctv,
 \end{split}
 \end{align} 
 for some $\delta,\varepsilon>0$ and $\epsilon,\epsilon'>0$ small enough.
 
Fix some 
$V_{\rho\Lambda_0}\in\caV_{\rho\Lambda_0}$,
$V_{\sigma\Lambda_0}\in\caV_{\sigma\Lambda_0}$,
$V_{\tau\Lambda_0}\in\caV_{\rho\Lambda_0}$,
$V_{\rho\Lambda_i+t_i\bm e_{\Lambda_i}}\in\caV_{\rho\Lambda_i+t_i\bm e_{\Lambda_i}}$,
$V_{\sigma\Lambda_i+t_i\bm e_{\Lambda_i}}\in\caV_{\sigma\Lambda_i+t_i\bm e_{\Lambda_i}}$,
$V_{\tau\Lambda_i+t_i\bm e_{\Lambda_i}}\in\caV_{\tau\Lambda_i+t_i\bm e_{\Lambda_i}}$,
for $t_i\ge 0$ $i=1,2$ and set
 ${W^{\bm t}}_{\rho\Lambda_i\Lambda_j}:=V_{\rho\ltj{j}{t_j}}V_{\rho\ltj{i}{t_i}}^*$,
 ${W^{\bm t}}_{\sigma\Lambda_i\Lambda_j}:=V_{\sigma\ltj{j}{t_j}}V_{\sigma\ltj{i}{t_i}}^*$,
 ${W^{\bm t}}_{\tau\Lambda_i\Lambda_j}:=V_{\tau\ltj{j}{t_j}}V_{\tau\ltj{i}{t_i}}^*$,
 if $i,j=1,2$.
 Set also
  ${W^{\bm t}}_{\rho\Lambda_0\Lambda_i}:=V_{\rho\ltj{i}{t_i}}V_{\rho\Lambda_0}^*$,
  for $i=1,2$.
By Remark \ref{lem32rem}, we may set
\begin{align}\label{compw}
\begin{split}
&{W^{\bm t}}_{\rho\;\;\circ_{\comp{}{0}}\;\;\sigma,\Lambda_0,\Lambda_1}
:={W^{\bm t}}_{\rho\Lambda_0\Lambda_1}\otimes {W^{\bm t}}_{\sigma\Lambda_0\Lambda_1}
={W^{\bm t}}_{\rho\Lambda_0\Lambda_1}\trl\rho0\lmk  {W^{\bm t}}_{\sigma\Lambda_0\Lambda_1}\rmk\\
&{W^{\bm t}}_{\sigma\;\;\circ_{\comp{}{0}}\;\;\tau,\Lambda_0,\Lambda_2}
:={W^{\bm t}}_{\sigma\Lambda_0\Lambda_2}\otimes {W^{\bm t}}_{\tau\Lambda_0\Lambda_2}
={W^{\bm t}}_{\sigma\Lambda_0\Lambda_2}\trl\sigma 0\lmk {W^{\bm t}}_{\tau\Lambda_0\Lambda_2}\rmk,
\end{split}
\end{align}
in the definition of $\epsilon_+^{(\Lambda_0)}\lmk \rho\;\;\circ_{\comp{}{0}}\;\;\sigma,\tau\rmk$
and $\epsilon_+^{(\Lambda_0)}\lmk  \rho,\sigma\;\;\circ_{\comp{}{0}}\;\;\tau\rmk$.

Fix any $\epsilon''>0$.
Choose $s\ge R_{|\arg\tilde \Lambda_1|+2(\varepsilon+\delta+\epsilon),\frac{\epsilon'}2},
R_{|\arg\tilde\Lambda_1|,\varepsilon}$
such that
\begin{align}
2f_{|\arg\tilde\Lambda_1|+2(\varepsilon+\delta+\epsilon),\frac{\epsilon'}2, \frac{\epsilon'}2}(s)<\epsilon'',
\quad
2f_{|\arg\tilde \Lambda_1|,\varepsilon,\delta}(s)<\epsilon''.
\end{align}
Because 
\begin{align}
W_{\sigma\Lambda_0\Lambda_1}^{\bm t}
\in\pi_0\lmk\caA_{\lmk
\lmk \Lambda_1+t_1\bm e_{\Lambda_1}\rmk\cup\Lambda_0
\rmk^c}\rmk'
\subset \amf\lmk\tilde \Lambda_1\rmk,
\end{align}
from Lemma \ref{lem2t6},
there is a unitary $\tilde W_{\sigma\Lambda_0\Lambda_1}^{\bm t s}\in \pi_0\lmk 
\caA_{\lmk \tilde \Lambda_1\rmk_{\varepsilon+\delta}-s\bm e_{\tilde \Lambda_1}}\rmk''$
such that
\begin{align}
\lV
W_{\sigma\Lambda_0\Lambda_1}^{\bm t}-
\tilde W_{\sigma\Lambda_0\Lambda_1}^{\bm t s}
\rV
\le2 f_{|\arg\tilde\Lambda_1|,\varepsilon,\delta}(s)
<\epsilon''.
\end{align}
From Lemma \ref{lem33},
we have
\begin{align}
\trlt{\rho}1\lmk\tilde W_{\sigma\Lambda_0\Lambda_1}^{\bm t s}\rmk
\in\amf\lmk
\lmk \tilde \Lambda_1\rmk_{\varepsilon+\delta+\epsilon}-s\bm e_{\tilde\Lambda_1}
\rmk.
\end{align}
Then by Lemma \ref{lem2t6}, there is a unitary
\begin{align}
X_{\sigma\Lambda_0\Lambda_1}^{\bm t s}
\in \pi_0
\lmk
\caA_{\lmk\tilde \Lambda_{1}\rmk_{\varepsilon+\delta+\epsilon+\epsilon'}-2s\bm e_{\tilde \Lambda_1}}
\rmk''
\end{align}
such that 
\begin{align}
\lV
X_{\sigma\Lambda_0\Lambda_1}^{\bm t s}-\trlt{\rho}1\lmk\tilde W_{\sigma\Lambda_0\Lambda_1}^{\bm t s}\rmk
\rV\le 2f_{|\arg\tilde\Lambda_1|+2(\varepsilon+\delta+\epsilon),\frac{\epsilon'}2, \frac{\epsilon'}2}(s)<\epsilon''.
\end{align}
For $t_2$ large enough, we have
\begin{align}
\lmk\tilde \Lambda_{1}\rmk_{\varepsilon+\delta+\epsilon+\epsilon'}-2s\bm e_{\tilde \Lambda_1}
\subset 
\lmk \Lambda_2+t_2\bm e_{\Lambda_2}\rmk^c,
\end{align}
and 
\begin{align}
\trlt{\tau}2\lmk X_{\sigma\Lambda_0\Lambda_1}^{\bm t s}
\rmk=X_{\sigma\Lambda_0\Lambda_1}^{\bm t s}.
\end{align}
Hence we obtain
\begin{align}
\lim\sup_{\bm t\to\infty}
\lV
\trlt{\tau}2\trlt{\rho}1
\lmk
W_{\sigma\Lambda_0\Lambda_1}^{\bm t}
\rmk
-\trlt{\rho}1\lmk W_{\sigma\Lambda_0\Lambda_1}^{\bm t}\rmk
\rV\le 2\epsilon''.
\end{align}
As this holds for any $\epsilon''>0$,
we have
\begin{align}
\lim_{\bm t\to\infty}
\lV
\trlt{\tau}2\trlt{\rho}1
\lmk
W_{\sigma\Lambda_1\Lambda_0}^{\bm t}
\rmk
-\trlt{\rho}1\lmk W_{\sigma\Lambda_1\Lambda_0}^{\bm t}\rmk
\rV=0.
\end{align}
Similarly, we have
\begin{align}
\lim_{\bm t\to\infty}
\lV
\trlt{\tau}2
\lmk
W_{\sigma\Lambda_0\Lambda_1}^{\bm t}
\rmk
-W_{\sigma\Lambda_0\Lambda_1}^{\bm t}
\rV=0.
\end{align}

Hence we have
\begin{align}\label{comcom1}
\begin{split}
&\lim_{\bm t\to\infty}
\lV 
\begin{gathered}
\trlt{\tau}2\trlt{\rho}{1}\lmk\wodt\sigma01\rmk\\-  \trlt{\rho}1\trlt{\tau}2\lmk \wodt\sigma01\rmk
\end{gathered}
\rV=0.
\end{split}
\end{align}

By the definition of $\epsilon_+^{(\Lambda_0)}$
and the fact that $\wodt\rho ij$ etc are intertwiners, we obtain
\begin{align}
\begin{split}
&\lmk
\epsilon_+^{(\Lambda_0)}\lmk\rho,\tau\rmk\otimes 1_{\trl{\sigma}0}
\rmk
\lmk
 1_{\trl{\rho}0}\otimes \epsilon_+^{(\Lambda_0)}\lmk\sigma,\tau\rmk
\rmk\\
&=\lim_{\bm t\to\infty}\trl{\tau}0\lmk \wodt\rho01^*\rmk \wodt\tau02^*
\wodt\rho01\\&\quad
\trl\rho0\lmk
\wodt\tau02\trl\tau0\lmk\wodt\sigma01^*\rmk\wodt\tau02^*\wodt\sigma01
\rmk\\
&\quad
\trl\rho0\trl\sigma0\lmk\wodt\tau02\rmk\\
&=\lim_{\bm t\to\infty}\wodt\tau02^*\trlt\tau2\lmk \wodt\rho 01^*\rmk\\
&\quad \trlt\rho1\lmk\trlt\tau2(\wodt\sigma01^*)\rmk
\trlt\rho1\lmk
\wodt\sigma01\rmk
\wodt\rho01
\\&\quad
\trl\rho0\trl\sigma0\lmk\wodt\tau02\rmk\\
&=\lim_{\bm t\to\infty}
\wodt\tau02^*\trlt\tau2\lmk \wodt\rho 01^*\rmk\\
&\quad \trlt\tau2\trlt\rho1\lmk\wodt\sigma01^*\rmk
\trlt\rho1\lmk
\wodt\sigma01\rmk
\wodt\rho01
\\&\quad
\trl\rho0\trl\sigma0\lmk\wodt\tau02\rmk\\
&=\lim_{\bm t\to\infty}
\wodt\tau02^*
\trlt\tau2\lmk \wodt\rho 01^*
\trlt\rho1\lmk\wodt\sigma01^*\rmk
\rmk\\
&\trlt\rho1\lmk
\wodt\sigma01\rmk
\wodt\rho01
\quad
\trl\rho0\trl\sigma0\lmk\wodt\tau02\rmk\\
&=\lim_{\bm t\to\infty}
\trl\tau0\lmk
\trl\rho0\lmk\wodt\sigma01^*\rmk \wodt\rho 01^*
\rmk
\wodt\tau02^*\wodt\rho01
\trl\rho0\lmk
\wodt\sigma01\rmk
\\&\quad
\trl\rho0\trl\sigma0\lmk\wodt\tau02\rmk\\
&=\lim_{\bm t\to\infty}
\trl\tau0\lmk
\lmk \wodt{\rho\circ_{\Lambda_0}\sigma}01\rmk^*
\rmk
\wodt\tau02^*
\wodt{\rho\circ_{\Lambda_0}\sigma}01
\trl\rho0\trl\sigma0\lmk\wodt\tau02\rmk\\
&=\lim_{\bm t\to\infty}\lmk
\wodt\tau02\otimes \wodt{\rho\circ_{\Lambda_0}\sigma}01
\rmk^*
\lmk
\wodt{\rho\circ_{\Lambda_0}\sigma}01\otimes \wodt\tau02
\rmk\\
&=\epsilon_+^{(\Lambda_0)}
\lmk\rho\circ_{\Lambda_0}\sigma,\tau\rmk.
\end{split}
\end{align}
In the third equality, we used (\ref{comcom1}).
In the sixth equality, we used (\ref{compw}).
This proves the first equality (\ref{fstb}).

The second one (\ref{sndb}) can be proven analogously.
As in the first case, we can choose $\Lambda_1,\Lambda_2$
so that
\begin{align}\lV
\begin{gathered}
\trlt\rho1\lmk\trlt\sigma2\lmk\wodt\tau02\rmk\rmk\\
-
\trlt\sigma2\trlt\rho1\lmk\wodt\tau02\rmk
\end{gathered}
\rV\to 0,\quad \bm t\to\infty.
\end{align}

We have 
\begin{align}
\begin{split}
&\lmk
 1_{\trl{\sigma}0}\otimes \epsilon_+^{(\Lambda_0)}\lmk\rho,\tau\rmk
\rmk
\lmk
\epsilon_+^{(\Lambda_0)}\lmk\rho,\sigma\rmk\otimes 1_{\trl{\tau}0}
\rmk\\
&=\lim_{\bm t\to\infty}\trl{\sigma}0\trl\tau0\lmk\lmk \wodt\rho01\rmk^*\rmk
\trl\sigma0\lmk\lmk \wodt\tau02\rmk^*\rmk\\
&\quad\trl\sigma0\lmk \trlt\rho1\lmk\wodt\tau02\rmk\rmk
\lmk\wodt\sigma 02\rmk^*\wodt\rho01\trl\rho0\lmk \wodt\sigma02\rmk\\
&=\lim_{\bm t\to\infty}\trl{\sigma}0\trl\tau0\lmk\lmk \wodt\rho01\rmk^*\rmk
\trl\sigma0\lmk\lmk \wodt\tau02\rmk^*\rmk\\
&\quad\lmk\wodt\sigma 02\rmk^*
\trlt\sigma2\lmk \trlt\rho1\lmk\wodt\tau02\rmk\rmk
\wodt\rho01\trl\rho0\lmk \wodt\sigma02\rmk\\
&=\lim_{\bm t\to\infty}\trl{\sigma}0\trl\tau0\lmk\lmk \wodt\rho01\rmk^*\rmk
\trl\sigma0\lmk\lmk \wodt\tau02\rmk^*\rmk\\
&\quad\lmk\wodt\sigma 02\rmk^*
 \trlt\rho1\lmk\trlt\sigma2\lmk\wodt\tau02\rmk\rmk
\wodt\rho01\trl\rho0\lmk \wodt\sigma02\rmk\\
&=\lim_{\bm t\to\infty}\trl{\sigma}0\trl\tau0\lmk\lmk \wodt\rho01\rmk^*\rmk
\lmk \wodt{\sigma\circ_{\Lambda_0}\tau} 02\rmk^*
\\
&\quad
 \wodt\rho01\trl{\rho}0\lmk \wodt{\sigma\circ_{\Lambda_0}\tau}02\rmk
 \\
&=\lim_{\bm t\to\infty}
\lmk
\wodt{\sigma\circ_{\Lambda}\tau}02
\otimes
\wodt\rho01 
\rmk^* \lmk
\wodt\rho01 \otimes \wodt{\sigma\circ_{\Lambda}\tau}02
\rmk\\
&=
\epsilon_+^{(\Lambda_0)}\lmk \rho, \sigma\;\;\circ_{\comp {}0}\tau\rmk.
\end{split}
\end{align}

\end{proof}

\section{Braided $C^*$-tensor category}\label{braidsec}
Let us consider the following setting.
For the definition of interactions and the dynamics given by them, see \cite{BR2}.
\begin{setting}\label{catset}
Let $\Phi$ be a uniformly bounded finite range interaction on $\at$ and $\tau_{\Phi}$ the $C^*$-dynamics given by $\Phi$.
Let $\omega$ be a pure $\tau_\Phi$-ground state satisfying the gap condition:
there is a $\gamma>0$ such that
\begin{align}
-i\omega\lmk B^* \delta_\Phi(B)\rmk\ge \gamma\omega(B^*B),
\end{align}
for any $B\in \at$ in the domain of the generator $\delta_\Phi$ of $\tau_{\Phi}$ with $\omega(B)=0$.
Let $(\caH,\pi_0,\Omega)$ be a GNS representation of $\omega$.
We use the notations in Setting \ref{setni} for this $(\caH,\pi_0)$.
We assume that $\pi_0$ has a nontrivial sector theory.
We also assume the approximate Haag duality for $\pi_0$.
Fix some $\theta\in\bbR$, $\varphi\in (0,\pi)$ and $\Lambda_0\in\ctv$.
We set
\begin{align}
\caO_{\Lambda_0}:=
\left\{
\rho\in\caO_0\mid
\rho\vert_{\Lambda_{0}^c}=\pi_0\vert_{\Lambda_{0}^c}
\right\}.
\end{align}
We define the multiplication of $\rho,\sigma\in \caO_{\Lambda_0}$
by $\rho\;\;\circ_{\comp{}{0}}\;\;\sigma$ in Definition \ref{compdef},
with $\bar V_{\eta,\Lambda_0}:=\unit_{\caH}\in \caV_{\eta,\Lambda_0}$,
for each $\eta\in\caO_{\Lambda_0}$.
Note from Lemma \ref{lem12} that in fact 
$\rho\;\;\circ_{\comp{}{0}}\;\;\sigma\in \caO_{\Lambda_0}$.
\end{setting}
In this section we show the following theorem.
For definition of a braided $C^*$-tensor category, see 
 \cite{NT}.
\begin{thm}\label{bracatthm}
Consider Setting \ref{catset}.
Then $\caO_{\Lambda_0}$ and its morphisms form
a braided $C^*$-tensor category.
\end{thm}
We already have tensor (i.e., composition) and braiding.
What is missing is direct sums and subrepresentations.
In order to do that we need to prove that $\pi_0(\caA_{\Lambda})''$
is properly infinite.
This requires different arguments than those in AQFT,
because $\Omega$ is not cyclic for $\pi_0(\caA_{\Lambda})''$.
We use the gap condition.
\begin{lem}\label{lem36}
Let $\Phi$ be a uniformly bounded finite range interaction on $\at$
and $\tau_{\Phi}$ the $C^*$-dynamics given by $\Phi$.
Let $\omega$ be a pure $\tau_\Phi$-ground state satisfying the gap condition:
there is a $\gamma>0$ such that
\begin{align}\label{gapcon}
-i\omega\lmk B^* \delta_\Phi(B)\rmk\ge \gamma\omega(B^*B),
\end{align}
for any $B\in \at$ in the domain of the generator $\delta_\Phi$ of $\tau_{\Phi}$ with $\omega(B)=0$.
Let $(\caH,\pi_0,\Omega)$ be a GNS representation of $\omega$.
Then for any cone $\Lambda$ in $\bbZ^2$, $\pi_0(\caA_\Lambda)''$ does not have
any normal tracial state.
\end{lem}
\begin{proof}
Suppose that there is a cone $\Lambda$ such that $\pi_0(\caA_\Lambda)''$ has
a normal tracial state $\tau$.
We derive a contradiction out of it.

For any $x\in\bbZ^2$, there are unit vectors $\xi_x,\eta_x\in \bbC^d$
such that 
\begin{align}\label{xchoice}
\omega\lmk X_{\eta_x,\xi_x}^{(x)}\rmk=0,\quad
\omega\lmk X_{\xi_x,\xi_x}^{(x)}\rmk\ge \frac 1d.
\end{align}
Here $X_{\zeta_1,\zeta_2}^{(x)}$ denotes an element in $\caA_{\{x\}}\simeq\Mat_d$
such that $X_{\zeta_1,\zeta_2}^{(x)}\zeta=\braket{\zeta_2}{\zeta}\zeta_1$.
To see this, let $\{e_{ij}^{(x)}\}$ be a system of matrix units in $\caA_{\{x\}}$.
Because $\sum_{i=1}^d\omega(e_{ii}^{(x)})=1$,
there is $i=1,\ldots,d$ such that $\omega(e_{ii}^{(x)})\ge \frac 1d$.
We set $\xi_x\in \bbC^d$ a unit vector corresponding to this projection
$e_{ii}^{(x)}$, i.e., $X_{\xi_x,\xi_x}^{(x)}=e_{ii}^{(x)}$.
Let $\eta_x$ be a unit vector perpendicular to $D\xi_x$,
where $D$ is the reduced density matrix of $\omega$ onto $\caA_{\{x\}}$.
Then we see these vectors satisfy the condition (\ref{xchoice}).

We also note that because our interaction is uniformly bounded and finite range,
there is some $c_1>0$ such that 
\begin{align}
\lV\delta_{\Phi}(B)\rV\le c_1\lV B\rV, \quad \text{for all}\quad
x\in\bbZ^2,\quad B\in \caA_{\{x\}}.
\end{align}
From now on, we fix $0<\varepsilon<\frac{\gamma}{4c_1 d}$.

Note that because $\omega$ is pure, $\pza{}$ is a factor.
To see this, note that
\begin{align}
\pza{}\cap\pzac{}\;\;\subset\;\; \pi_0(\caA_{\Lambda^c})'\cap \pzac{}\;\;\subset\;\;
\pi_0(\at)'=\bbC\unit.
\end{align}
By definition, $\tau\circ\pi_0\vert_{\caA_\Lambda}$ is a $\pi_0\vert_{\caA_\Lambda}$-normal state.
Because $\pza{}$ is a factor, this means that $\tau\circ\pi_0\vert_{\caA_\Lambda}$ and $\omega\vert_{\caA_\Lambda}$
are quasi-equivalent.
From Corollary 2.6.11 of \cite{BR1}, this means
there exists a finite subset $S$ of $\Lambda\cap\bbZ^2$
such that 
\begin{align}\label{149}
\lv
\tau\circ\pi_0(B)-\omega(B)
\rv<\varepsilon\lV B\rV,\quad
B\in \caA_{\Lambda\setminus S},
\end{align}
for our fixed $\varepsilon$.

Suppose that $\Phi$ has a range less than $m\in\bbN$.
Let $S_m$ be a set of points in $\bbZ^2$
whose distance from $S\cup\Lambda^c$ is less than or equal to $m$.
Let $x\in \bbZ^2\setminus \lmk S_m\rmk$.
For this $x$, $X_{\eta_x,\xi_x}^{(x)},  \lmk X_{\eta_x,\xi_x}^{(x)}\rmk^*$ belongs to the domain of $\delta_\Phi$
and $\omega\lmk X_{\eta_x,\xi_x}^{(x)}\rmk=\omega\lmk \lmk X_{\eta_x,\xi_x}^{(x)}\rmk^*\rmk=0$.
Therefore, by the gap condition (\ref{gapcon}), we have
\begin{align}\label{151}
\begin{split}
&-i\omega\lmk \lmk X_{\eta_x,\xi_x}^{(x)}\rmk^* \delta_{\Phi}\lmk X_{\eta_x,\xi_x}^{(x)}\rmk\rmk
+i\omega \lmk\delta_{\Phi}\lmk X_{\eta_x,\xi_x}^{(x)}\rmk \lmk X_{\eta_x,\xi_x}^{(x)}\rmk^*\rmk \\
&=-i\omega\lmk \lmk X_{\eta_x,\xi_x}^{(x)}\rmk^* \delta_{\Phi}\lmk X_{\eta_x,\xi_x}^{(x)}\rmk\rmk
-i\omega \lmk X_{\eta_x,\xi_x}^{(x)} \delta_{\Phi}\lmk X_{\eta_x,\xi_x}^{(x)}\rmk^*\rmk 
\\
&\ge \gamma\omega\lmk \lmk X_{\eta_x,\xi_x}^{(x)}\rmk^*  X_{\eta_x,\xi_x}^{(x)}\rmk
+\gamma\omega\lmk X_{\eta_x,\xi_x}^{(x)}\lmk X_{\eta_x,\xi_x}^{(x)}\rmk^*  \rmk
\ge \gamma\omega\lmk X_{\xi_x,\xi_x}^{(x)}\rmk
\ge \frac{\gamma} d.
\end{split}
\end{align}
In the equality, we used the fact that $\omega$ is $\tau_\Phi$-invariant.
On the other hand,
by tracial property of $\tau$, we have
\begin{align}\label{152}
\begin{split}
&\lv
-i\omega\lmk \lmk X_{\eta_x,\xi_x}^{(x)}\rmk^* \delta_{\Phi}\lmk X_{\eta_x,\xi_x}^{(x)}\rmk\rmk
+i\omega \lmk\delta_{\Phi}\lmk X_{\eta_x,\xi_x}^{(x)}\rmk \lmk X_{\eta_x,\xi_x}^{(x)}\rmk^*\rmk
\rv\\
&\le
\lv
-i\omega\lmk \lmk X_{\eta_x,\xi_x}^{(x)}\rmk^* \delta_{\Phi}\lmk X_{\eta_x,\xi_x}^{(x)}\rmk\rmk
+i\tau\circ\pi_0\lmk \lmk X_{\eta_x,\xi_x}^{(x)}\rmk^* \delta_{\Phi}\lmk X_{\eta_x,\xi_x}^{(x)}\rmk\rmk
\rv\\
&+\lv
-i\tau\circ\pi_0\lmk \lmk X_{\eta_x,\xi_x}^{(x)}\rmk^* \delta_{\Phi}\lmk X_{\eta_x,\xi_x}^{(x)}\rmk\rmk
+i\tau\circ\pi_0\lmk\delta_{\Phi}\lmk X_{\eta_x,\xi_x}^{(x)}\rmk \lmk X_{\eta_x,\xi_x}^{(x)}\rmk^*\rmk
\rv\\
&+
\lv-i\tau\circ\pi_0\lmk\delta_{\Phi}\lmk X_{\eta_x,\xi_x}^{(x)}\rmk \lmk X_{\eta_x,\xi_x}^{(x)}\rmk^*\rmk
+i\omega \lmk\delta_{\Phi}\lmk X_{\eta_x,\xi_x}^{(x)}\rmk \lmk X_{\eta_x,\xi_x}^{(x)}\rmk^*\rmk
\rv\\
&\le
\varepsilon \lV \lmk X_{\eta_x,\xi_x}^{(x)}\rmk^* \delta_{\Phi}\lmk X_{\eta_x,\xi_x}^{(x)}\rmk \rV
+\varepsilon \lV\delta_{\Phi}\lmk X_{\eta_x,\xi_x}^{(x)}\rmk \lmk X_{\eta_x,\xi_x}^{(x)}\rmk^*
\rV\\
&\le 2\varepsilon c_1.
\end{split}
\end{align}
Note that $\delta_{\Phi}\lmk X_{\eta_x,\xi_x}^{(x)}\rmk$ belongs to
$\caA_{\Lambda\setminus S}$ because of our choice of $x\in \bbZ^2\setminus S_m$.
We applied (\ref{149}).
Combining (\ref{151}) and (\ref{152}), we obtain
$\frac{\gamma} d\le 2\varepsilon c_1$, which contradicts to
our choice of $\varepsilon$.
\end{proof}
The following Theorem is proven in \cite{NaOg}.
\begin{thm}\label{no}
Consider Setting \ref{setni}. Suppose that 
$\pi_0$ has a nontrivial sector theory.
Then for any cone $\Lambda$, $\pza{}$ is not of type $I$.
\end{thm}
\begin{lem}\label{lem37}
In Setting \ref{catset},
for any cone $\Lambda$, $\pza{}$ is either type $II_\infty$
or type $III$ factor.
Furthermore, $\amf(\Lambda)=\pi_0(\caA_{\Lambda^c})'$
is either type $II_\infty$
or type $III$ factor.
\end{lem}
\begin{proof}
Let $\Lambda$ be a cone.
Recall from the proof of Lemma \ref{lem36} that $\pza{}$ is a factor.
By Corollary 1.20 V of \cite{takesaki},
$\pza{}$ should be either type $I$, type $II_1$, type $II_\infty$ or type $III$ factor.
From Theorem \ref{no}, it is not of type $I$.
From Lemma \ref{lem36}, $\pza{}$ does not have a normal tracial state.
From Theorem 2.4 V \cite{takesaki}, this means that $\pza{}$ is not finite.
Therefore, it is not type $II_1$.

Because $\pi_0(\al_{\Lambda^c})''$ is factor, $\pi_0(\al_{\Lambda^c})'$ is factor. 
By Corollary 2.24 V of \cite{takesaki}, $\amf(\Lambda)=\pi_0(\al_{\Lambda^c})'$ 
is not type $I$ because $\pi_0(\al_{\Lambda^c})''$ is not.
If $\pi_0(\al_{\Lambda^c})'$ is type $II_1$, 
it admits sufficiently many finite normal trace. (Theorem 2.4 V\cite{takesaki}.)
From 
$
\pza{}\subset \pi_0(\al_{\Lambda^c})'
$, this means $\pza{}$ admits sufficiently many finite normal trace, which is not true because 
$\pza{}$ is not finite.
Therefore, $\amf(\Lambda)=\pi_0(\caA_{\Lambda^c})'$ is 
either type $II_\infty$
or type $III$ factor.
\end{proof}
\begin{lem}\label{lem38}
For any $\rho,\sigma\in \caO_{\Lambda_0}$, we have
\begin{align}
(\rho,\sigma)\subset 
\lmk \trlz\rho,\trlz\sigma\rmk\subset\caB_{(\theta,\varphi)}.
\end{align}
\end{lem}
\begin{proof}
For any $R\in (\rho,\sigma)$, we have
\begin{align}
R\trlz\rho\lmk \pi_0(A)\rmk=R\rho(A)=\sigma(A)R=\trlz\sigma\lmk \pi_0(A)\rmk R,
\end{align}
for all $A\in\at$
using Lemma \ref{lem9} (ii) with $\bar V_{\rho\Lambda_0}=\unit$.
By the $\sigma$w-continuity of $\trlz\rho$, $\trlz\sigma$ on 
$\pi_0(\caA_{\Lambda})''$ this is extended to $\caB_0$ of Lemma \ref{lemhoshi}, 
and from 
Lemma \ref{lemhoshi}, this extends to $\btv$.
The last inclusion is by Lemma \ref{lem17}.
\end{proof}

\begin{lem}\label{directsum}
Consider Setting \ref{catset}.
For any $\rho,\sigma\in\caO_{\Lambda_{0}}$,
there exists a $\tau\in \caO_{\Lambda_{0}}$, 
and isometries $u\in (\rho,\tau)$,
$v\in (\sigma,\tau)$ such that $uu^{*}+vv^{*}=\unit$.
\end{lem}
\begin{proof}
From Lemma \ref{lem37}, $\pza{}$ is properly intinite for any cone $\Lambda$.
Therefore, there exist a projection $p_{\Lambda}\in \pza{}$
such that $p_{\Lambda}\sim \unit- p_{\Lambda}\sim \unit$
from Proposition 1.36 V\cite{takesaki}, i.e.,
there are isometries $u_{\Lambda}, v_{\Lambda}\in \pza{}$
such that $u_{\Lambda} u_{\Lambda}^{*}=p_{\Lambda}$ and $v_{\Lambda} v_{\Lambda}^{*}=\unit- p_{\Lambda}$.
Define $\tau : \at\to \caB(\caH_{0})$ by 
\begin{align}
\tau(A):=u_{\Lambda_{0}}\rho(A) u_{\Lambda_{0}}^{*}
+v_{\Lambda_{0}}\sigma(A) v_{\Lambda_{0}}^{*},\quad A\in \at.
\end{align}
Because the range of  $u_{\Lambda_0}$ and $ v_{\Lambda_0}$ are orthogonal
and they are isometries, $\tau$ is a $*$-homomorphism.

To see $\tau\in\caO_{\Lambda_0}$, fix some
$V_{\rho,\Lambda}\in\caV_{\rho,\Lambda}$ and
$V_{\sigma,\Lambda}\in\caV_{\sigma,\Lambda}$
for each cone $\Lambda$.
For $\Lambda_0$, we may and we do set $V_{\rho,\Lambda_0}=\unit$, $V_{\sigma,\Lambda_0}=\unit$, because
$\rho,\sigma\in \caO_{\Lambda_0}$.
Set
\begin{align}
W_\Lambda:=u_{\Lambda} V_{\rho\Lambda} u_{\Lambda_0}^*
+v_\Lambda V_{\sigma\Lambda} v_{\Lambda_0}^*.
\end{align}
By our choice $V_{\rho,\Lambda_0}=\unit$, $V_{\sigma,\Lambda_0}=\unit$,
we have $W_{\Lambda_0}=\unit$.
For general $\Lambda$, $W_\Lambda$ is a unitary:
\begin{align}
W_\Lambda^* W_\Lambda=
u_{\Lambda_0} V_{\rho\Lambda}^* u_{\Lambda}^*
u_{\Lambda} V_{\rho\Lambda} u_{\Lambda_0}^*
+v_{\Lambda_0}
V_{\sigma\Lambda}^* v_{\Lambda}^*v_\Lambda V_{\sigma\Lambda} v_{\Lambda_0}^*
=u_{\Lambda_0}u_{\Lambda_0}^*+v_{\Lambda_0}v_{\Lambda_0}^*=\unit,\\
W_\Lambda W_{\Lambda}^*=u_{\Lambda} V_{\rho\Lambda}^* u_{\Lambda_0}^*
u_{\Lambda_0} V_{\rho\Lambda} u_{\Lambda}^*
+v_{\Lambda}
V_{\sigma\Lambda}^* v_{\Lambda_0}^*v_\Lambda V_{\sigma\Lambda_0} v_{\Lambda}^*
=u_{\Lambda}u_{\Lambda}^*+v_{\Lambda}v_{\Lambda}^*=\unit.
\end{align}
For any cone $\Lambda$, we have  $\Ad\lmk W_{\Lambda}\rmk\circ\tau\vert_{\caA_{\Lambda^c}}
=\pi_0\vert_{\caA_{\Lambda^c}}$.
In fact, for any $A\in \caA_{\Lambda^c}$, we have
\begin{align}
\begin{split}
&\Ad\lmk W_{\Lambda}\rmk\circ\tau(A)
=\Ad\lmk u_{\Lambda} V_{\rho\Lambda} u_{\Lambda_0}^*
+v_\Lambda V_{\sigma\Lambda} v_{\Lambda_0}^*
\rmk
\lmk
u_{\Lambda_{0}}\rho(A) u_{\Lambda_{0}}^{*}
+v_{\Lambda_{0}}\sigma(A) v_{\Lambda_{0}}^{*}
\rmk\\
&=
\Ad\lmk
u_{\Lambda} V_{\rho\Lambda} u_{\Lambda_0}^*
\rmk
\lmk
u_{\Lambda_{0}}\rho(A) u_{\Lambda_{0}}^{*}
\rmk
+
\Ad\lmk
v_\Lambda V_{\sigma\Lambda} v_{\Lambda_0}^*
\rmk
\lmk
v_{\Lambda_{0}}\sigma(A) v_{\Lambda_{0}}^{*}
\rmk\\
&=\Ad\lmk
u_{\Lambda} V_{\rho\Lambda} 
\rmk
\lmk
\rho(A) 
\rmk
+
\Ad\lmk
v_\Lambda V_{\sigma\Lambda}
\rmk
\lmk
\sigma(A)
\rmk\\
&=
\Ad\lmk
u_{\Lambda}
\rmk
\lmk
\pi_0(A) 
\rmk
+
\Ad\lmk
v_\Lambda 
\rmk
\lmk
\pi_0(A)
\rmk\\
&=\lmk u_\Lambda u_\Lambda^*+v_\Lambda v_\Lambda^*\rmk\lmk \pi_0(A)\rmk
=\pi_0(A).
\end{split}
\end{align}
In the fourth equality, we used $u_{\Lambda}, v_{\Lambda}\in \pza{}$
and $A\in \caA_{\Lambda^c}$.
In particular, 
because $W_{\Lambda_0}=\unit$, we have $\tau\vert_{\caA_{\Lambda_0^c}}
=\pi_0\vert_{\caA_{\Lambda_0^c}}$.
This proves $\tau\in \caO_{\Lambda_0}$.

We have $u_{\Lambda_0}\in (\rho,\tau)$ because 
\begin{align}
\tau(A) u_{\Lambda_0}=\lmk u_{\Lambda_{0}}\rho(A) u_{\Lambda_{0}}^{*}
+v_{\Lambda_{0}}\sigma(A) v_{\Lambda_{0}}^{*}\rmk u_{\Lambda_0}
=u_{\Lambda_0}\rho(A) u_{\Lambda_0}^* u_{\Lambda_0}
=u_{\Lambda_0}\rho(A)
\end{align}
for any $A\in \at$.
Similarly, $v_{\Lambda_0}\in (\sigma,\tau)$.
Setting $u:=u_{\Lambda_0}$ and $v:=v_{\Lambda_0}$
we obtain the isometries satisfying the required condition.
\end{proof}
%
\begin{lem}\label{lem41}
Consider Setting \ref{catset}.
For any $\rho\in\caO_{\Lambda_{0}}$ and a nonzero projection $p\in (\rho,\rho)$,
there exists a $\tau\in \caO_{\Lambda_{0}}$
and an isometry $v\in (\tau,\rho)$,
such that $vv^*=p$.
\end{lem}

In order to prove this, we first prepare following general Lemmas.
The following Lemma is well known, and can be derived as a refinement of the proof of Lemma 2.5.2
of \cite{Lin}.
\begin{lem}\label{lin}
For any $\varepsilon>0$, there exists $\delta_\varepsilon>0$
satisfying the following.:
For any unital $C^*$-algebra $\mathfrak B$ and projections $p,q\in\mathfrak B$
such that $\lV q-qp\rV<\delta_\varepsilon$,
there is a partial isometry $w\in \mathfrak B$ such that
\begin{align}
ww^*=q,\quad w^*w\le p,\quad \lV w-q\rV<\varepsilon.
\end{align}
\end{lem}
\begin{lem}\label{nmd}
There is some $\frac 12>\delta>0$ satisfying the following.:
Let $\caH$ be a separable Hilbert space, and $\caN,\caM$
be infinite factors on $\caH$ such that $\caN\subset \caM$.
Let $p$ be a nonzero projection in $\caN'$, and $u$ a unitary on $\caH$
such that $\tilde p:=u pu^*\in \caM$ and $\lV u-\unit \rV<\delta$.
Then $\tilde p$ is equivalent to $\unit$ in $\caM$.
\end{lem}
\begin{proof}
We set  $\delta:=\frac 18\cdot  \min \left\{\delta_{\frac 18}, \frac 18\right\}$,
with $\delta_\varepsilon$ in the Lemma \ref{lem41}.
Let us consider $\caH,\caN,\caM,p,u$ as above, satisfying the condition for this $\delta$.
We would like to show that $\tilde p:=u pu^*$  is equivalent to $\unit$ in $\caM$.
If $p=\unit$, this is trivial. We assume $p\neq\unit$.

Because $\caN$ is an infinite factor, $\unit$ is properly infinite in $\caN$.
Therefore, from Proposition 1.36 V\cite{takesaki}, there exists a projection $F$ in $\caN$ and an isometry
$v\in\caN$
such that $\unit\neq F$ and $vv^*=F$.
As $p\in\caN'$ is a nonzero projection and $\caN$ is a factor,
the map
\begin{align}
\caN\ni x\mapsto xp\in \caN p
\end{align}
is a $*$-isomorphim.
Therefore, $vv^*p=Fp\neq p$.
Because $F,v\in \caN$ and $p\in\caN'$,  $p-vpv^*=p(1-vv^*)=p(1-F)$ is a nonzero projection.
In particular, we have $\lV p-vpv^*\rV=1$.
Because $(v\tilde p)^*v\tilde p=\tilde p v^*v\tilde p=\tilde p$,
$v\tilde p$ is a partial isometry in $\caM$ and $v\tilde pv^*$ is a projection in
$\caM$.
Because $\lV u-\unit \rV<\delta$, we have
\begin{align}
\lV v\tilde p v^*(\unit-\tilde p)\rV
\le 2\lV\tilde p-p\rV+\lV vpv^*(\unit-p)\rV
\le 4\lV u-\unit\rV
<\delta_{\frac 18}.
\end{align}
Therefore, from Lemma \ref{lin}, there exists a partial isometry $w\in\caM$
such that $ww^*=v\tilde p v^*$, $w^* w\le \tilde p$
and $\lV w-v\tilde p v^*\rV<\frac 18$.
Set $w_0:=w^* v\tilde p\in \caM$.
We have
\begin{align}
w_0^*w_0=\tilde p v^* ww^* v\tilde p=\tilde p v^* v\tilde p v^* v\tilde p
=\tilde p,
\end{align}
and $w_0$ is a partial isometry.
We also have
\begin{align}
w_0w_0^*=w^* v\tilde p v^* w\le w^* w\le \tilde p.
\end{align}
Therefore, if $\tilde p-w_0w_0^*\neq 0$, then 
$\tilde p$ is infinite.
In fact,
because
\begin{align}
\begin{split}
\lV
w_0w_0^*-\tilde p
\rV=\lV w^* v\tilde p v^* w-\tilde p\rV
=\lV
\lmk w- v\tilde p v^*\rmk^* v\tilde p v^* w+
v\tilde p v^*(w-v\tilde p v^*)+v\tilde p v^*-\tilde p
\rV\\
\ge \lV v\tilde p v^*-\tilde p\rV-2\lV w- v\tilde p v^*\rV
\ge \lV vp v^*-p\rV-4\lV u-\unit\rV-2\lV w- v\tilde p v^*\rV
\ge \frac 12,
\end{split}
\end{align}
we have $\tilde p-w_0w_0^*\neq 0$, and 
$\tilde p$ is infinite.
We used $\lV p-vpv^*\rV=1$ in the last inequality.
Because $\caH$ is separable and $\caM$ is an infinite factor, from Corollary 6.3.5 of
\cite{KR}, $\tilde p$ is equivalent to $\unit$ in $\caM$.
\end{proof}
\begin{lem}\label{nmr}
Let $\caN,\caM,\caR$ be von Neumann algebras 
 acting on a Hilbert space $\caH$ such that $\caN,\caM\subset\caR$ and 
$\caN\subset \caM'$.
Suppose that $\caM$ is an infinite factor.
Then any nonzero projection in $\caN$ is infinite in $\caR$. 
\end{lem}
\begin{proof}
Let $p\in\caN$ be a nonzero projection.
Since $\caM$ is a factor, and $p\in \caN\subset \caM'$ a projection,
the map
\begin{align}\label{aap}
\caM \ni a\mapsto ap\in\caM p
\end{align}
is a $*$-isomorphism.

Because $\caM$ is an infinite factor, by  Proposition 1.36 V\cite{takesaki}, 
there is a projection $q\in \caM$ such that
$q\sim (\unit-q)\sim\unit$ in $\caM$.
Hence there is some $v\in \caM$ such that $v^*v=\unit$
and $vv^*=q$.
For this $q\in \caM$, $pq\in \caR$ is a projection.
Because of the injectivity of the map (\ref{aap}), we have
$p\neq pq$.
Set $w:=pv\in \caR$.
We have 
\begin{align}
ww^*=pv v^* p=pqp=pq,\quad
w^*w=v^*pv=v^*vp=p.
\end{align}
Hence $p$ and $pq$ are equivalent in $\caR$.
Because $p\neq pq$, $p$ is infinite in $\caR$.
\end{proof}
\begin{lem}\label{nmr2}
Let $\caN,\caM,\caR$ be von Neumann algebras acting on a separable Hilbert space $\caH$
such that $\caN,\caM\subset\caR$ and $\caN\subset \caM'$.
Suppose that $\caM$ is an infinite factor and $\caR$ is a factor.
Let $w$ be a unitary on $\caH$ and $u\in\caN$ an isometry
such that $w^*uu^*w\in\caN$.
Then there is a unitary $W$ on $\caH$
such that
\begin{align}
\Ad\lmk Wu^*\rmk(x)=\Ad\lmk u^* w\rmk(x),\quad x\in\caR'.
\end{align}
\end{lem}
\begin{proof}
From Lemma \ref{nmr},
non-zero projections $w^*uu^*w, uu^*\in \caN$
are infinite in $\caR$.
Because $\caR$ is a factor and $\caH$ is separable,
from Corollary 6.3.5 \cite{KR}, 
$w^*uu^*w$ and $uu^*$ are equivalent in $\caR$.
Namely, there is some $a\in\caR$ such that $w^*uu^*w=a^*a$
and $uu^*=aa^*$.

Set $W:=u^*wa^*u$. This $W$ is unitary because
\begin{align}
\begin{split}
W^*W=u^*aw^*uu^*wa^*u
=u^* a a^*aa^* u=u^*u=\unit,\\
WW^*=u^*wa^*uu^*aw^*u
=u^*wa^*aa^*aw^*u=u^*ww^*u=u^*u=\unit.
\end{split}
\end{align}
For any $x\in \caR'$, we have
\begin{align}
\begin{split}
\Ad\lmk Wu^*\rmk(x)
=\Ad\lmk u^*wa^*u u^*\rmk(x)
=\Ad\lmk u^*wa^*a a^*\rmk(x)
=\Ad\lmk u^*wa^*\rmk(x)\\
=\Ad\lmk u^* w\rmk\lmk a^*xa\rmk
=\Ad\lmk u^* w\rmk\lmk a^*a x\rmk
=\Ad\lmk u^*w\rmk(x).
\end{split}
\end{align}

\end{proof}

\begin{proofof}[Lemma \ref{lem41}]

For each cone $\Lambda$, we fix some $V_{\rho,\Lambda}\in\caV_{\rho,\Lambda}$.
For $\Lambda_0$, we may and we do set $V_{\rho,\Lambda_0}=\unit$.
Set $p_{\Lambda}:=\Ad(V_{\rho,\Lambda})(p)$.
Note that $p_{\Lambda}$ belongs to $\pi_0\lmk\caA_{\Lambda^c}\rmk'$ for each cone $\Lambda$, because
\begin{align}
p\cdot \Ad\lmk V_{\rho\Lambda}^*\rmk\lmk \pi_0(A)\rmk=p\rho(A)=\rho(A) p=
\Ad\lmk V_{\rho\Lambda}^*\rmk\lmk \pi_0(A)\rmk\cdot p,\quad A\in \caA_{\Lambda^c}.
\end{align}

Let $\delta>0$ be the number given in Lemma \ref{nmd}.
For each $\varphi\in (0,2\pi)$, $\varepsilon>0$ with $\varphi+8\varepsilon<2\pi$,
fix some $t_{\varphi,\varepsilon}\ge R_{\varphi,\varepsilon}$
such that $f_{\varphi,\varepsilon,\varepsilon}(t_{\varphi,\epsilon})<\delta$.
(Recall the definition of approximate Haag duality Definition \ref{assum7}.)
For each cone $\Lambda$, fix some  $0<\varepsilon_\Lambda<\min\{
\frac{2\pi-|\arg\Lambda|}{16}, \frac{|\arg\Lambda|}{16}
\}$
and
set
\begin{align}
\Gamma_{\Lambda}:=\Lambda_{-4\varepsilon_\Lambda}
+t_{|\arg\Lambda|-8\varepsilon_\Lambda,\varepsilon_\Lambda}\bm e_{\Lambda}.
\end{align}
We choose a cone $D_{\Lambda}$ so that 
$D_{\Lambda}\subset 
\lmk \Gamma_\Lambda-R_{|\arg\Gamma_\Lambda|,\varepsilon_\Lambda}
\bm e_{\Gamma_{\Lambda}}\rmk_{\varepsilon_\Lambda}\cap\Gamma_\Lambda^c$.
Note that
\begin{align}
D_\Lambda
\subset 
\lmk \Gamma_\Lambda-R_{|\arg\Gamma_\Lambda|,\varepsilon_\Lambda}
\bm e_{\Gamma_{\Lambda}}\rmk_{\varepsilon_\Lambda}\cap\Gamma_\Lambda^c
\subset \Lambda_{-3\varepsilon_\Lambda}
+\lmk
t_{|\arg\Lambda|-8\varepsilon_{\Lambda},\varepsilon_\Lambda}
-R_{|\arg\Gamma_\Lambda|,\varepsilon_\Lambda}
\rmk\bm e_{\Lambda}
\subset\Lambda_{-3\varepsilon_\Lambda}.
\end{align}
Note that $|\arg\Gamma_\Lambda|=|\arg\Lambda|-8\varepsilon_{\Lambda}$.
By the approximate Haag duality,
 there is a unitary 
$U_{\Gamma_\Lambda,\varepsilon_\Lambda}\in \caU(\caH)$
satisfying
\begin{align}
\pi_0\lmk\caA_{\Gamma_\Lambda^c}\rmk'\subset 
\Ad\lmk U_{\Gamma_\Lambda,\varepsilon_\Lambda}\rmk\lmk 
\pi_0\lmk \caA_{\lmk \Gamma_\Lambda-R_{|\arg\Gamma_\Lambda|,\varepsilon_\Lambda}\bm e_{\Gamma_\Lambda}\rmk_{\varepsilon_\Lambda}}\rmk''
\rmk.
\end{align}
Furthermore, there is a unitary 
\begin{align}
\hat U_\Lambda:=\tilde U_{\Gamma_\Lambda,\varepsilon_\Lambda,\varepsilon_\Lambda,t_{|\arg\Lambda|-8\varepsilon_{\Lambda},\varepsilon_\Lambda}}\in
  \pi_0\lmk \caA_{(\Gamma_\Lambda)_{2\varepsilon_\Lambda}-
  t_{|\arg\Lambda|-8\varepsilon_{\Lambda},\varepsilon_\Lambda}
  \bm e_{\Lambda}}\rmk''
  =\pi_0\lmk\caA_{\Lambda_{-2\varepsilon_\Lambda}}\rmk''
\end{align}
 satisfying
\begin{align}\label{uappro1}
\lV
U_{\Gamma_\Lambda,\varepsilon_\Lambda}-\hat U_{\Lambda}
\rV\le f_{\arg|\Gamma_\Lambda|,\varepsilon_\Lambda,\varepsilon_\Lambda}
\lmk t_{|\arg\Lambda|-8\varepsilon_{\Lambda},\varepsilon_\Lambda}\rmk
<\delta<\frac 12.
\end{align}

Now we apply Lemma \ref{nmd}
to the separable (because $\caH$ is a GNS Hilbert space of a state on $\at$) Hilbert space $\caH$,
infinite factors
$\caN:=\pi_0\lmk\caA_{D_\Lambda}\rmk''$
$\caM:=\pi_0\lmk\caA_{\Lambda_{-2\varepsilon_\Lambda}}\rmk''$,
$p_{\Gamma_\Lambda}
\in \pi_0\lmk\caA_{\Gamma_\Lambda^c}\rmk'\subset \pi_0\lmk\caA_{D_\Lambda}\rmk'= \caN'$,
$u:=\hat U_{\Lambda} U_{\Gamma_\Lambda,\varepsilon_\Lambda}^*\in \caU(\caH)$.
Note that 
\begin{align}
\Ad\lmk  \hat U_{\Lambda} U_{\Gamma_\Lambda,\varepsilon_\Lambda}^*
\rmk
(p_{\Gamma_\Lambda})\in \Ad\lmk \hat U_{\Lambda} \rmk
\lmk \pi_0\lmk \caA_{ \lmk \Gamma_\Lambda\rmk_{\varepsilon_\Lambda}-R_{|\arg \Gamma_\Lambda|,\varepsilon_\Lambda}\bm e_{\Gamma_\Lambda}}\rmk\rmk''
\subset\pi_0\lmk\caA_{\Lambda_{-2\varepsilon_\Lambda}}\rmk''=\caM.
\end{align}
Applying  Lemma \ref{nmd}, we conclude that 
$\Ad\lmk \hat U_{\Lambda} U_{\Gamma_\Lambda,\varepsilon_\Lambda}^*\rmk(p_{\Gamma_\Lambda})$ is equivalent to $\unit$ in $\pi_0\lmk\caA_{\Lambda_{-2\varepsilon_\Lambda}}\rmk''$.
Hence for each cone $\Lambda$,
there is an isometry $u_{\Lambda}\in \pi_0\lmk\caA_{\Lambda_{-2\varepsilon_\Lambda}}\rmk''$
such that $u_{\Lambda}u_{\Lambda}^*=\Ad\lmk w_{\Lambda}\rmk(p_{\Gamma_\Lambda})$,
where we set $w_{\Lambda}:=
 \hat U_{\Lambda} U_{\Gamma_\Lambda,\varepsilon_\Lambda}^*$.
Note that $\lV w_{\Lambda}-\unit\rV<\frac 12$ from (\ref{uappro1}).

For each cone $\Lambda$, we set
\begin{align}
T_{\Lambda}:=\Ad\lmk u_{\Lambda}^*
w_{\Lambda}
V_{\rho\Gamma_{\Lambda}}\rmk\circ\rho :\at\to\caB(\caH).
\end{align}
It is a $*$-homomorphism
because
\begin{align}
\begin{split}
&T_{\Lambda}(A)T_{\Lambda}(B)=
 u_{\Lambda}^*
w_{\Lambda}
V_{\rho\Gamma_{\Lambda}}\rho(A)
V_{\rho\Gamma_{\Lambda}}^*
w_{\Lambda}^*
 u_{\Lambda}
 u_{\Lambda}^*
w_{\Lambda}
V_{\rho\Gamma_{\Lambda}}\rho(B)
V_{\rho\Gamma_{\Lambda}}^*
w_{\Lambda}^*
 u_{\Lambda}\\
 &=u_{\Lambda}^*
w_{\Lambda}
V_{\rho\Gamma_{\Lambda}}\rho(A)
V_{\rho\Gamma_{\Lambda}}^*
p_{\Gamma_\Lambda}
V_{\rho\Gamma_{\Lambda}}\rho(B)
V_{\rho\Gamma_{\Lambda}}^*
w_{\Lambda}^*
 u_{\Lambda}\\
 &=u_{\Lambda}^*
w_{\Lambda}
V_{\rho\Gamma_{\Lambda}}\rho(A)
p \rho(B)
V_{\rho\Gamma_{\Lambda}}^*
w_{\Lambda}^*
 u_{\Lambda}\\
 &=u_{\Lambda}^*
w_{\Lambda}
V_{\rho\Gamma_{\Lambda}}\rho(AB)p
V_{\rho\Gamma_{\Lambda}}^*
w_{\Lambda}^*
 u_{\Lambda}
 \\
 &=u_{\Lambda}^*
w_{\Lambda}
V_{\rho\Gamma_{\Lambda}}\rho(AB)
V_{\rho\Gamma_{\Lambda}}^*
w_{\Lambda}^*
 u_{\Lambda}
 =T_{\Lambda}(AB),
\end{split}
\end{align}
for any $A,B\in\at$.

For any cone $\Lambda$,
there is a unitary $W_\Lambda$ such that 
\begin{align}\label{waru}
T_{\Lambda}\vert_{\caA_{\Lambda^c}}=\Ad(W_\Lambda)\circ
\pi_0\vert_{\caA_{\Lambda^c}}.
\end{align}
To see this,
we choose a cone $C_\Lambda$ such that
\begin{align}\label{c3d}
C_\Lambda\subset \lmk \Lambda_{-2\varepsilon_\Lambda}\rmk^c\cap \Lambda.
\end{align}
We use Lemma \ref{nmr2} for infinite factors
$\caN:=\pi_0\lmk\caA_{\lmk \Lambda_{-2\varepsilon_\Lambda}\rmk^c}\rmk'$, 
$\caM:=\pi_0\lmk \caA_{C_\Lambda}\rmk''$,
$\caR:=\pi_0\lmk\caA_{\Lambda^c}\rmk'$
acting on the separable Hilbert space $\caH$
and a unitary $w_\Lambda$ on  $\caH$ and an isometry $u_{\Lambda}\in \pi_0\lmk\caA_{\Lambda_{-2\varepsilon_\Lambda}}\rmk''\subset  \pi_0\lmk\caA_{\lmk \Lambda_{-2\varepsilon_\Lambda}\rmk^c}\rmk' =\caN$.
Because of (\ref{c3d}), we have $\caN\subset \caM'$
and $\caN,\caM\subset \caR$.
We have 
\begin{align}
w_\Lambda^*u_{\Lambda}u_{\Lambda}^*w_{\Lambda}
=p_{\Gamma_\Lambda}\in \pi_0\lmk\caA_{\Gamma_\Lambda^c}\rmk'
\subset \pi_0\lmk \caA_{\lmk \Lambda_{-2\varepsilon_\Lambda}\rmk^c}\rmk'.
\end{align}
Applying Lemma \ref{nmr2}, there is a unitary $W_\Lambda$
on $\caH$ such that
\begin{align}\label{kro}
\Ad\lmk W_\Lambda u_\Lambda^*\rmk(x)
=\Ad\lmk u_\Lambda^* w_\Lambda\rmk(x),\quad
x\in \caR'=\pi_0\lmk\caA_{\Lambda^c}\rmk''.
\end{align}
Note that for any $A\in \caA_{\Lambda^c}$, we have
\begin{align}
\begin{split}
T_{\Lambda}(A)
=\Ad\lmk u_{\Lambda}^*
w_{\Lambda}
V_{\rho\Gamma_{\Lambda}}\rmk\circ\rho(A)
=\Ad\lmk u_{\Lambda}^*
w_{\Lambda}
V_{\rho\Gamma_{\Lambda}} V_{\rho \Lambda}^*
\rmk\circ\pi_0(A)
=\Ad\lmk u_{\Lambda}^*
w_{\Lambda}
\rmk\circ\pi_0(A)\\
=\Ad\lmk W_\Lambda u_\Lambda^*\rmk\circ\pi_0(A)
=\Ad\lmk W_\Lambda \rmk\lmk u_\Lambda^*\pi_0(A) u_\Lambda\rmk
=\Ad\lmk W_\Lambda \rmk\lmk u_\Lambda^* u_\Lambda\pi_0(A)\rmk
=\Ad\lmk W_\Lambda \rmk\lmk \pi_0(A)\rmk.
\end{split}
\end{align}
Here we used that
\begin{align}
V_{\rho\Gamma_{\Lambda}} V_{\rho \Lambda}^*\in \pi_0
\lmk\caA_{\Lambda^c}\rmk'
\end{align}
in the third equality and (\ref{kro}) in the fourth equality
and the fact $u_\Lambda\in  \pi_0\lmk\caA_{\lmk \Lambda_{-2\varepsilon_\Lambda}\rmk^c}\rmk'
\subset \pi_0\lmk \caA_{\Lambda^c}\rmk'$ in the sixth equality.
This proves (\ref{waru}).

Now we set 
\begin{align}
\tau:=\Ad\lmk W_{\Lambda_0}^*\rmk\circ T_{\Lambda_0}
\end{align}
with $W_\Lambda$ in (\ref{waru}).
By definition, $\tau$ is a $*$-representation and we have $\tau\vert_{\caA_{\Lambda_0^c}}=\pi_0\vert_{\caA_{\Lambda_0^c}}$ by (\ref{waru}).
For each cone $\Lambda$, set 
\begin{align}
X_\Lambda:= W_{\Lambda}^*
u_{\Lambda}^*
w_{\Lambda}
V_{\rho\Gamma_{\Lambda}}
V_{\rho\Gamma_{\Lambda_0}}^*
w_{\Lambda_0}^*
u_{\Lambda_0}
 W_{\Lambda_0}.
\end{align}
We have
\begin{align}
\begin{split}
&X_\Lambda X_\Lambda^*
=W_{\Lambda}^*
u_{\Lambda}^*
w_{\Lambda}
V_{\rho\Gamma_{\Lambda}}
V_{\rho\Gamma_{\Lambda_0}}^*
w_{\Lambda_0}^*
u_{\Lambda_0}
 W_{\Lambda_0}
 W_{\Lambda_0}^*
 u_{\Lambda_0}^*
 w_{\Lambda_0}
 V_{\rho\Gamma_{\Lambda_0}}
 V_{\rho\Gamma_{\Lambda}}^*
 w_{\Lambda}^*
 u_{\Lambda}W_{\Lambda}\\
 &=
 W_{\Lambda}^*
u_{\Lambda}^*
w_{\Lambda}
V_{\rho\Gamma_{\Lambda}}
V_{\rho\Gamma_{\Lambda_0}}^*
p_{\Gamma_{\Lambda_0}}
 V_{\rho\Gamma_{\Lambda_0}}
 V_{\rho\Gamma_{\Lambda}}^*
 w_{\Lambda}^*
 u_{\Lambda}W_{\Lambda}\\
 &=
 W_{\Lambda}^*
u_{\Lambda}^*
w_{\Lambda}
p_{\Gamma_{\Lambda}}
 w_{\Lambda}^*
 u_{\Lambda}W_{\Lambda}=
 W_{\Lambda}^*
u_{\Lambda}^*
u_{\Lambda}
u_{\Lambda}^*
 u_{\Lambda}W_{\Lambda}=\unit.
\end{split}
\end{align}
Similarly we have $X_\Lambda^* X_\Lambda=\unit$
and $X_\Lambda$ is a unitary.
From the definition, we see that
\begin{align}
\Ad(X_\Lambda)\circ\tau
=\Ad(W_\Lambda^*)T_{\Lambda},
\end{align}
because
\begin{align}
\begin{split}
\Ad(X_\Lambda)\circ\tau
=\Ad\lmk W_{\Lambda}^*
u_{\Lambda}^*
w_{\Lambda}
V_{\rho\Gamma_{\Lambda}}
V_{\rho\Gamma_{\Lambda_0}}^*
w_{\Lambda_0}^*
u_{\Lambda_0}
 W_{\Lambda_0} W_{\Lambda_0}^*u_{\Lambda_0}^*
w_{\Lambda_0}
V_{\rho\Gamma_{\Lambda_0}}\rmk
\circ\rho\\
=\Ad\lmk W_{\Lambda}^*
u_{\Lambda}^*
w_{\Lambda}
V_{\rho\Gamma_{\Lambda}}
V_{\rho\Gamma_{\Lambda_0}}^*
p_{\Gamma_{\Lambda_0}}
V_{\rho\Gamma_{\Lambda_0}}\rmk
\circ\rho
=\Ad\lmk W_{\Lambda}^*
u_{\Lambda}^*
w_{\Lambda}
V_{\rho\Gamma_{\Lambda}}
p
\rmk
\circ\rho\\
=\Ad\lmk W_{\Lambda}^*
u_{\Lambda}^*
w_{\Lambda}p_{\Gamma_\Lambda}
V_{\rho\Gamma_{\Lambda}}
\rmk
\circ\rho
=\Ad\lmk W_{\Lambda}^*
u_{\Lambda}^*
w_{\Lambda}
V_{\rho\Gamma_{\Lambda}}
\rmk
\circ\rho=\Ad(W_\Lambda^*)T_{\Lambda}.
\end{split}
\end{align}
Therefore, from (\ref{waru}), we have 
$\Ad(X_\Lambda)\circ\tau\vert_{\caA_{\Lambda^c}}=\pi_0\vert_{\caA_{\Lambda^c}}$.
Hence we conclude $\tau\in \caO_{\Lambda_0}$.

Finally, we set
\begin{align}
v:=V_{\rho\Gamma_{\Lambda_0}}^*
w_{\Lambda_0}^*
u_{\Lambda_0}
 W_{\Lambda_0}
\end{align}
and prove 
\begin{align}
v\in (\tau,\rho),\quad v^*v=\unit,\quad vv^*=p.
\end{align}
For any $A\in\at$,
\begin{align}\begin{split}
v \tau(A)
=V_{\rho\Gamma_{\Lambda_0}}^*
w_{\Lambda_0}^*
u_{\Lambda_0}
 W_{\Lambda_0}
 W_{\Lambda_0}^* 
 u_{\Lambda_0}^*
w_{\Lambda_0}
V_{\rho\Gamma_{{\Lambda_0}}}\rho(A)
V_{\rho\Gamma_{\Lambda_0}}^*
w_{\Lambda_0}^*
u_{\Lambda_0}
 W_{\Lambda_0}
\\ 
=\rho(A)p
V_{\rho\Gamma_{\Lambda_0}}^*
w_{\Lambda_0}^*
u_{\Lambda_0}
 W_{\Lambda_0}
 =\rho(A)
V_{\rho\Gamma_{\Lambda_0}}^* p_{\Gamma_{\Lambda_0}}
w_{\Lambda_0}^*
u_{\Lambda_0}
 W_{\Lambda_0}
 =\rho(A)
V_{\rho\Gamma_{\Lambda_0}}^*
w_{\Lambda_0}^*
u_{\Lambda_0}
 W_{\Lambda_0}
=\rho(A)v.\end{split}
\end{align}
Hence $v\in (\tau,\rho)$.
Furthermore, we have
\begin{align}
\begin{split}
&vv^*
=
V_{\rho\Gamma_{\Lambda_0}}^*
w_{\Lambda_0}^*
u_{\Lambda_0}
u_{\Lambda_0}^*
w_{\Lambda_0}
 V_{\rho\Gamma_{\Lambda_0}}
 =V_{\rho\Gamma_{\Lambda_0}}^*
 p_{\Lambda_0}
 V_{\rho\Gamma_{\Lambda_0}}
 =p,\\
 &v^*v
 =
 W_{\Lambda_0}^*
u_{\Lambda_0}^*
u_{\Lambda_0}
 W_{\Lambda_0}
 =
 \unit.
 \end{split}
\end{align}
This completes the proof.
\end{proofof}

\begin{proofof}[Theorem \ref{catset}]
We prove that $\caO_{\Lambda_0}$ and its morphisms form a braided $C^*$-tensor
category \cite{NT}. 

For any $\rho,\sigma\in \caO_{\Lambda_0}$,
$(\rho,\sigma)$ is a Banach space.
For any $\rho,\sigma, \gamma\in \caO_{\Lambda_0}$,
the map 
\begin{align}
(\sigma,\gamma)\times (\rho,\sigma)\ni (S, R)\mapsto
SR\in (\rho,\gamma)
\end{align}
is bilinear and $\lV SR\rV\le \lV S\rV \lV R\rV$ holds.
For any $R\in (\rho,\sigma)$, its adjoint $R^*$ belongs to
$(\sigma,\rho)$ and $R^{**}=R$,
$\lV R^* R\rV=\lV R\rV^2$
(in particular, $(\rho,\rho)$ is a $C^*$-algebra), and $R^*R\in (\rho,\rho)$ is positive.
We define the tensor product of $\caO_{\Lambda_0}$ by
\begin{align}
\rho\otimes \sigma :=\rho\circ_{(\theta,\varphi),\Lambda_0, \{ V_{\eta\Lambda_0}=\unit\}}\sigma,\quad
\rho,\sigma \in \caO_{\Lambda_0}.
\end{align}
Tensor product of intertwiners are given by 
\begin{align}
R_1\otimes R_2=R_1\trlz\rho(R_2),\quad
R_1\in (\rho,\rho'),\quad R_2\in (\sigma,\sigma'),\quad
\rho,\rho',\sigma,\sigma'\in \caO_{\Lambda_0}
\end{align}
as in Lemma \ref{lem18}.
Note that it is bi-linear.
Morphisms $\alpha_{\rho, \sigma,\gamma}:=\unit_{\caB(\caH)}$
 play the role of the associativity morphisms.
 (Recall Lemma \ref{lem14}.)
 The irreducible representation $\pi_0$ corresponds to the unit object, because
 $\trlz{\pi_0}=\id$.
 (This can be see from Lemma \ref{lem9} (a).)
  Furthermore, setting
$\lambda_\sigma=\rho_{\sigma}=\unit_{\caB(\caH)}$,
we have natural unitary isomorphisms
$\lambda_\sigma : \pi_o\otimes\sigma\to \sigma$
and $\rho_\sigma: \sigma\otimes \pi_0\to \sigma$.
It is clear that the associativity morphism satisfies the pentagonal diagram,
and the natural unitary isomorphisms satisfy the triangle diagram.
By Lemma \ref{lem18},
morphisms satisfy $(R\otimes S)^*=R^*\otimes S^*$.
The existence of the direct sums and subobjects are proven in Lemma \ref{directsum} and 
Lemma \ref{lem41} respectively.
The unit object $\pi_0$ is simple, i.e. $(\pi_0,\pi_0)=\bbC\unit_{\caB(\caH)}$
because it is irreducible.
As a subset of $\at\times\caH\times\caH$, $\caO_{\Lambda_0}$ is a set.
Definition \ref{epdef} and Lemma \ref{lem416} proves that
$\epsilon_+^{\Lambda_0}$ gives the braiding.
\end{proofof}

\section{Stability of the braiding structure}\label{stasec}
\begin{thm}\label{monoidalthm}
Let $\Phi_1,\Phi_2$ be uniformly bounded finite range interactions on $\caA_{\bbZ^2}$
and $\omega_{\Phi_i}$ $i=1,2$ pure $\tau_{\Phi_i}$-ground states satisfying the gap condition.
Suppose that the GNS  representation of $\omega_{\Phi_1}$ has a non-trivial sector theory and
satisfies the approximate Haag duality.
Suppose that there is a approximately-factorizable automorphism $\alpha$
such that $\omega_{\Phi_2}=\omega_{\Phi_1}\circ\alpha$.
Then the following hold.
\begin{description}
\item[(i)]The GNS  representation of $\omega_{\Phi_2}$ has a non-trivial sector theory and
satisfies the approximate Haag duality.
\item[(ii)]Fix some $\theta\in\bbR$, $\varphi\in (0,\pi)$
and $\Lambda_0^{(i)}\in \ctv$, for $i=1,2$.
From (i) and Theorem \ref{catset}, we obtain braided $C^*$-tensor categories $C_1=(\caO_1,Mor(\caO_1))$,
$C_2=(\caO_2,Mor(\caO_2))$ associated 
to $\omega_{\Phi_1}$ $\Lambda_0^{(1)}$ and $\omega_{\Phi_2}$, $\Lambda_0^{(2)}$ respectively.
The braided $C^*$-tensor category $C_1$, $C_2$ are unitarily monoidaly equivalent.
\end{description}
\end{thm}
\begin{proof}
(i) is from \cite{NaOg} and Lemma \ref{staah}. We prove (ii).
Let $(\caH_1, \pi_1)$ be a GNS representation for $\omega_{\Phi_1}$.
Then $(\caH_2,\pi_2):=(\caH_1,\pi_1\alpha)$ is a GNS representation of $\omega_{\Phi_2}$.
For any cone $\Lambda$ and $\varepsilon>0$,
there are automorphism $\beta_{\Lambda}, \tilde \beta_{\Lambda}
\in\Aut\lmk\caA_{\Lambda}\rmk$,
$\beta_{\Lambda^c}, \tilde \beta_{\Lambda^c}\in \Aut\lmk\caA_{\Lambda^c}\rmk$
and $\Xi_{\Lambda, \varepsilon}, \tilde \Xi_{\Lambda, \varepsilon}\in \Aut\lmk\caA_{\Lambda_{\varepsilon}\cap ((\Lambda^c)_\varepsilon)}\rmk$
and unitaries $v_{\Lambda\varepsilon},\tilde v_{\Lambda\varepsilon}\in\at$
such that $\alpha=\Ad\lmk v_{\Lambda,\varepsilon}\rmk\circ
\Xi_{\Lambda, \varepsilon}\circ \lmk\beta_{\Lambda}\otimes \beta_{\Lambda^c}\rmk$,
$\alpha^{-1}=\Ad\lmk \tilde v_{\Lambda,\varepsilon}\rmk\circ
\tilde \Xi_{\Lambda, \varepsilon}\circ \lmk\tilde \beta_{\Lambda}\otimes \tilde
\beta_{\Lambda^c}\rmk$.

For each $\eta\in\caO_i$ and $\Lambda\in\ctv$, we fix some $\bar V_{\eta,\Lambda}\in \caV_{\eta,\Lambda}$. We may and we do choose  $\bar V_{\eta,\Lambda_0^{(i)}}=\unit$, $\eta\in\caO_i$
and $\bar V_{\pi_i,\Lambda}=\unit$ for $i=1,2$ and any cone $\Lambda$.
For each objects $\rho_1\in \caO_1$, $\rho_2\in \caO_2$ and a cone $\Lambda$,
we fix $0<2\varepsilon<|\arg\Lambda|$ and
set
\begin{align}\label{eq61}
Y_{\rho_1}^\Lambda
:= \rho_1(v_{\Lambda,\varepsilon}) \bar V_{\rho_1,\Lambda_{-\varepsilon}}^*
\pi_1\lmk v_{\Lambda,\varepsilon}^*\rmk,\quad
\tilde Y_{\rho_2}^\Lambda
:= \rho_2(\tilde v_{\Lambda,\varepsilon}) \bar V_{\rho_2,\Lambda_{-\varepsilon}}^*
\pi_2\lmk \tilde v_{\Lambda,\varepsilon}^*\rmk.
\end{align}
Due to the choice of $\bar V_{\pi_i,\Lambda}=\unit$, we have
$Y_{\pi_1}^\Lambda=\unit=\tilde Y_{\pi_2}^{\Lambda}$.
We have 
\begin{align}
\Ad\lmk \lmk Y_{\rho_1}^{\Lambda}\rmk^*\rmk\circ\rho_1\circ\alpha\vert_{\caA_{\Lambda^c}}=\pi_2\vert_{\caA_{\Lambda^c}}
\end{align}
because
for any $A\in \caA_{\Lambda^c}$
\begin{align}
\begin{split}
&\rho_1\circ\alpha(A)
=\rho_1\circ \Ad\lmk v_{\Lambda,\varepsilon}\rmk\circ
\Xi_{\Lambda, \varepsilon}\circ \lmk\beta_{\Lambda}\otimes \beta_{\Lambda^c}\rmk(A)
=\Ad\lmk\rho_1\lmk v_{\Lambda,\varepsilon}\rmk\rmk
\rho_1\lmk
\Xi_{\Lambda, \varepsilon}\circ \lmk\beta_{\Lambda}\otimes \beta_{\Lambda^c}\rmk(A)
\rmk\\
&=
\Ad\lmk\rho_1\lmk v_{\Lambda,\varepsilon}\rmk \bar V_{\rho_1\Lambda_{-\varepsilon}}^*\rmk
\pi_1\lmk
\Xi_{\Lambda, \varepsilon}\circ \lmk\beta_{\Lambda}\otimes \beta_{\Lambda^c}\rmk(A)
\rmk
=\Ad\lmk \rho_1(v_{\Lambda,\varepsilon}) \bar V_{\rho_1,\Lambda_{-\varepsilon}}^*
\pi_1\lmk v_{\Lambda,\varepsilon}^*\rmk\rmk
\lmk \pi_2(A)\rmk.
\end{split}
\end{align}
Here we used $\Xi_{\Lambda, \varepsilon}\circ \lmk\beta_{\Lambda}\otimes \beta_{\Lambda^c}\rmk(A)
\in \pi_1(\caA_{\lmk \Lambda_{-\varepsilon}\rmk^c})$.
Similarly, we get 
$\Ad\lmk \lmk \tilde Y_{\rho_2}^{\Lambda}\rmk^*\rmk\circ\rho_2\circ\alpha^{-1}\vert_{\caA_{\Lambda^c}}=\pi_1\vert_{\caA_{\Lambda^c}}$.

From this, for each $\rho_1\in\caO_1$ and $\rho_2\in\caO_2$
we get maps $F:\caO_1\to\caO_2$, $G: \caO_2\to \caO_1$ defined by
\begin{align}\label{188}
F(\rho_1):=\Ad\lmk \lmk Y_{\rho_1}^{\Lambda_0^{(2)}}\rmk^*\rmk\circ\rho_1\circ\alpha\in\caO_2,\quad
G(\rho_2):=\Ad\lmk \lmk \tilde Y_{\rho_2}^{\Lambda_0^{(1)}}\rmk^*\rmk\circ\rho_2\circ\alpha^{-1}
\in \caO_1.
\end{align}
Note that
\begin{align}\label{189}
GF(\rho_1)
=\Ad\lmk
\lmk
\tilde Y_{F(\rho_1)}^{\Lambda_0^{(1)}}
\rmk^* 
\lmk Y_{\rho_1}^{\Lambda_0^{(2)}}\rmk^*
\rmk\circ\rho_1.
\end{align}
We also have $F(\pi_1)=\pi_1\alpha=\pi_2$ and
$G(\pi_2)=\pi_2\alpha^{-1}=\pi_1$.
For morphisms $R_1\in (\rho_1,\sigma_1)$ with $\rho_1,\sigma_1\in \caO_1$, 
we have
\begin{align}
F(R_1):=\lmk Y_{\sigma_1}^{\Lambda_0^{(2)}}\rmk^*
R_1 Y_{\rho_1}^{\Lambda_0^{(2)}}\in \lmk F(\rho_1),F(\sigma_1)\rmk.
\end{align}
In fact for any $A\in \at$, we have
\begin{align}
\begin{split}
F(R_1)\lmk F(\rho_1)(A)\rmk
=\lmk Y_{\sigma_1}^{\Lambda_0^{(2)}}\rmk^*
R_1 Y_{\rho_1}^{\Lambda_0^{(2)}}
 \lmk Y_{\rho_1}^{\Lambda_0^{(2)}}\rmk^*\lmk \rho_1\circ\alpha\lmk A\rmk\rmk Y_{\rho_1}^{\Lambda_0^{(2)}}
 =\lmk Y_{\sigma_1}^{\Lambda_0^{(2)}}\rmk^*
\lmk \sigma_1\circ\alpha\lmk A\rmk \rmk R_1Y_{\rho_1}^{\Lambda_0^{(2)}}\\
=\Ad\lmk \lmk Y_{\sigma_1}^{\Lambda_0^{(2)}}\rmk^*\rmk
\lmk \sigma_1\circ\alpha\lmk A\rmk \rmk 
\cdot \lmk Y_{\sigma_1}^{\Lambda_0^{(2)}}\rmk^*
R_1Y_{\rho_1}^{\Lambda_0^{(2)}}
=\lmk F(\sigma_1)(A)\rmk F(R_1).
\end{split}
\end{align}
Similarly, for morphisms $R_2\in (\rho_2,\sigma_2)$ with $\rho_2,\sigma_2\in \caO_2$, 
we have
\begin{align}
G(R_2):=\lmk \tilde Y_{\sigma_2}^{\Lambda_0^{(1)}}\rmk^*
R_2 \tilde Y_{\rho_2}^{\Lambda_0^{(1)}}\in \lmk G(\rho_2),G(\sigma_2)\rmk.
\end{align}
Note that $F,G$ are linear on morphisms and 
 $F(R_1)^*=F(R_1^*)$ and $G(R_2)^*=G(R_2^*)$.
Hence we obtained a functor $F$ from $C_1$ to $C_2$
and  a functor $G$ from $C_2$ to $C_1$.
Any element of $(\pi_i,\pi_i)$ is of the form $c\unit$ because $\pi_i$
is irreducible
and
$F(c\unit)=c\unit$.

Next we would like to make $F,G$ tensor functors.
To do so, first we note the domain $\caB_{(\theta,\varphi)}^{(i)}:=\overline{\cup_{\Lambda\in \ctv} \pi_i\lmk \caA_{\Lambda^c}\rmk'}$
of $T_{\rho_i}^{(\theta,\varphi), \Lambda_0^{(i)} \{\unit\}}$
for $i=1,2$ are the same.
Recall from Lemma \ref{lemhoshi}
that $\btv^{(i)}=\overline{\cup_{\Lambda\in \ctv}\pi_i(\caA_\Lambda)''}$.
In fact 
for any $\Lambda\in \ctv$, choose $\varepsilon>0$ with $\Lambda_\varepsilon\in\ctv$, and
we have
\begin{align}\label{eq68}
\pi_2\lmk\caA_{\Lambda}\rmk''
=\lmk \pi_1\circ\alpha\lmk\caA_{\Lambda}\rmk\rmk''
=\Ad\lmk \pi_1(v_{\Lambda,\varepsilon})\rmk
\lmk
\lmk \pi_1\lmk \Xi_{\Lambda\varepsilon}\beta_{\Lambda}\lmk\caA_{\Lambda}\rmk\rmk\rmk''
\rmk
\subset \Ad\lmk \pi_1(v_{\Lambda,\varepsilon})\rmk
\lmk \pi_1\lmk \caA_{\Lambda_\varepsilon}\rmk''\rmk
\subset \caB_{(\theta,\varphi)}^{(1)}.
\end{align}
Hence we have $\caB_{(\theta,\varphi)}^{(2)}\subset \caB_{(\theta,\varphi)}^{(1)}$.
Similarly, we have $\caB_{(\theta,\varphi)}^{(1)}\subset \caB_{(\theta,\varphi)}^{(2)}$.
From now on we use the notation $\btv:= \caB_{(\theta,\varphi)}^{(1)}=\caB_{(\theta,\varphi)}^{(2}$.

Note that $\wrl{\rho_1}{2}, \wrlt{\rho_2}{1}\in \btv$.
This is because of (\ref{eq61}) and
\begin{align}
\begin{split}
&\rho_1\lmk \at\rmk\in\trlzi{\rho_1}{1}\circ\pi_1(\at)
\subset\btv,\quad
\pi_1\lmk \at \rmk\in\btv,\\
&\bar V_{\rho_1\lmk \Lambda_0^{(2)}\rmk_{-\varepsilon}}^*=\bar V_{\rho_1
\Lambda_0^{(1)}} \bar V_{\rho_1\lmk \Lambda_0^{(2)}\rmk_{-\varepsilon}}^*
\in \btv.
\end{split}
\end{align}

Next we prove
\begin{align}\label{trh}
\trlzi{F(\rho_1)}{2}=\Ad\lmk \lmk Y_{\rho_1}^{\Lambda_0^{(2)}}\rmk^*\rmk
\trlzi{\rho_1}{1},\quad \text{for all} \quad \rho_1\in\caO_{1}.
\end{align}
Set $X:=\Ad\lmk \lmk Y_{\rho_1}^{\Lambda_0^{(2)}}\rmk^*\rmk\trlzi{\rho_1}{1}$.
By the definition of $F$, and $\trlzi{\rho_1}{1}\circ\pi_1=\rho_1$,
 we have $X\circ\pi_2=X\circ\pi_1\circ\alpha=F(\rho_1)$.
 For any cone $\Lambda\in\ctv$, choose $\varepsilon>0$
 with $\Lambda_\varepsilon\in \ctv$ and 
 we have
 $
 \pi_2(\caA_{\Lambda})''
\subset  \Ad\lmk\pi_1\lmk  v_{\Lambda\varepsilon}\rmk\rmk
\lmk \pi_1\lmk \caA_{\Lambda_\varepsilon}\rmk''\rmk
 $ from (\ref{eq68}).
 Because $X$ is $\sigma$w-continuous on $\pi_1(\caA_{\Lambda_\varepsilon})''$,
 it is  $\sigma$w-continuous on $\pi_2(\caA_{\Lambda})''$.
 From Lemma \ref{lem9},
 we have $\trlzi{F(\rho_1)}{2}=X$, proving (\ref{trh}).
 Similarly, we have
 \begin{align}\label{trh2}
\trlzi{G(\rho_2)}{1}=\Ad\lmk \lmk Y_{\rho_2}^{\Lambda_0^{(1)}}\rmk^*\rmk
\trlzi{\rho_2}{2},\quad \text{for all} \quad \rho_2\in\caO_{2}.
\end{align}

From $\wrl{\rho_1}{2}, \wrlt{\rho_2}{1}\in \btv$, we can define unitaries
\begin{align}\label{199}
F_2(\rho_1,\sigma_1)
:=\lmk \wrl{\rho_1\otimes\sigma_1}2\rmk^* \trlzi{\rho_1}1\lmk \wrl{\sigma_1}2\rmk\wrl{\rho_1}2,\quad
\rho_1,\sigma_1\in \caO_1.
\end{align}
We then find using (\ref{trh}) that 
for any $\rho_1,\sigma_1\in \caO_1$ that
\begin{align}\label{200}
\begin{split}
&F(\rho_1)\otimes F(\sigma_1)
=\trlzi{F(\rho_1)}2\circ \trlzi{F(\sigma_1)}2\circ\pi_2\\
&=\Ad\lmk \lmk \wrl{\rho_1}2\rmk^*\rmk
\circ
\trlzi{\rho_1}1\Ad\lmk
\lmk \wrl{\sigma_1}2\rmk^*
\rmk
\circ \trlzi{\sigma_1}1\circ\pi_2\\
&=\Ad\lmk \lmk\wrl{\rho_1}2\rmk^* \trlzi{\rho_1}1 \lmk \wrl{\sigma_1}2\rmk^*\rmk\circ
\trlzi{\rho_1}1\trlzi{\sigma_1}1\pi_1\circ\alpha\\
&=\Ad\lmk F_2(\rho_1,\sigma_1)^*\rmk\circ F\lmk\rho_1\otimes \sigma_1\rmk.
\end{split}
\end{align}
Similarly, for $\rho_2,\sigma_2\in \caO_2$,
we have
\begin{align}\label{201}
G(\rho_2)\otimes G(\sigma_2)
=\Ad\lmk G_2(\rho_2,\sigma_2)^*\rmk\lmk G\lmk \rho_2\otimes \sigma_2\rmk\rmk
\end{align}
with
\begin{align}\label{202}
G_2(\rho_2,\sigma_2)
:=\lmk \wrlt{\rho_2\otimes\sigma_2}1\rmk^* \trlzi{\rho_2}2\lmk \wrlt{\sigma_2}1\rmk\wrlt{\rho_2}1,\quad
\rho_2,\sigma_2\in \caO_2.
\end{align}
We have
\begin{align}\label{203}
\begin{split}
&F_2\lmk \rho_1\otimes\sigma_1,\tau_1\rmk F_2(\rho_1,\sigma_1)\\
&=\lmk \wrl{\rho_1\otimes\sigma_1\otimes \tau_1}2\rmk^* \trlzi{\rho_1\otimes \sigma_1}1\lmk \wrl{\tau_1}2\rmk\wrl{\rho_1\otimes\sigma_1}2\\
&\lmk \wrl{\rho_1\otimes\sigma_1}2\rmk^* \trlzi{\rho_1}1\lmk \wrl{\sigma_1}2\rmk\wrl{\rho_1}2\\
&=\lmk \wrl{\rho_1\otimes\sigma_1\otimes \tau_1}2\rmk^* \trlzi{\rho_1\otimes \sigma_1}1\lmk \wrl{\tau_1}2\rmk \trlzi{\rho_1}1\lmk \wrl{\sigma_1}2\rmk\wrl{\rho_1}2,
\end{split}
\end{align}
and
\begin{align}\label{204}
\begin{split}
&F_2\lmk\rho_1,\sigma_1\otimes \tau_1\rmk
\trlzi{F(\rho_1)}2\lmk F_2(\sigma_1,\tau_1)\rmk\\
&=\lmk \wrl{\rho_1\otimes\sigma_1\otimes \tau_1}2\rmk^* \trlzi{\rho_1}1\lmk \wrl{\sigma_1\otimes \tau_1}2\rmk\wrl{\rho_1}2
\trlzi{F(\rho_1)}2
\lmk \lmk \wrl{\sigma_1\otimes \tau_1}2\rmk^* \trlzi{\sigma_1}1\lmk \wrl{\tau_1}2\rmk\wrl{\sigma_1}2\rmk\\
&=
\lmk \wrl{\rho_1\otimes\sigma_1\otimes \tau_1}2\rmk^* \trlzi{\rho_1}1\lmk \wrl{\sigma_1\otimes \tau_1}2\rmk
\trlzi{\rho_1}1
\lmk \lmk \wrl{\sigma_1\otimes \tau_1}2\rmk^* 
\trlzi{\sigma_1}1\lmk \wrl{\tau_1}2\rmk\wrl{\sigma_1}2\rmk
\wrl{\rho_1}2\\
&=
\lmk \wrl{\rho_1\otimes\sigma_1\otimes \tau_1}2\rmk^* 
\trlzi{\rho_1}1
\lmk 
\trlzi{\sigma_1}1\lmk \wrl{\tau_1}2\rmk\wrl{\sigma_1}2\rmk
\wrl{\rho_1}2\\
&=
\lmk \wrl{\rho_1\otimes\sigma_1\otimes \tau_1}2\rmk^* 
\trlzi{\rho_1\otimes\sigma_1}1\lmk \wrl{\tau_1}2\rmk
\trlzi{\rho_1}1\lmk\wrl{\sigma_1}2\rmk
\wrl{\rho_1}2.
\end{split}
\end{align}
We used (\ref{trh}).
In the last line we used Lemma \ref{lem14}.
Note that the right hand side of (\ref{203}) and (\ref{204})
are the same.

Hence the following holds:
\[
\begin{CD}
 \lmk F(\rho_1)\otimes F(\sigma_1)\rmk\otimes F(\tau_1) @>{F_2\otimes \iota}>> 
 F(\rho_1\otimes \sigma_1)\otimes F(\tau_1) @>{F_2}>>F\lmk\lmk\rho_1\otimes\sigma_1\rmk\otimes \tau_1\rmk\\
  @V{\unit=a'}VV   @. @VV{F(a)=\unit}V \\
    F(\rho_1)\otimes \lmk F(\sigma_1)\otimes F(\tau_1) \rmk    @> {\iota\otimes F_2}>>  
    F(\rho_1)\otimes F(\sigma_1\otimes \tau_1) @> {F_2}>> F\lmk\rho_1\otimes \lmk \sigma_1\otimes \tau_1\rmk\rmk
\end{CD},
\]
where $a'$ and $F(a)$ are associators which are $\unit$.

Note from Lemma \ref{lem9} (a)
that $\trlzi{\pi_1}1=\id$.
For any $\rho_1\in\caO_1$, we then have
\begin{align}
F_2\lmk \pi_1,\rho_1\rmk
=\lmk \wrl{\pi_1\otimes\rho_1}2\rmk^* \trlzi{\pi_1}1
\lmk \wrl{\rho_1}2\rmk \wrl{\pi_1}2=\unit.
\end{align}
Hence the following holds:
\[
\begin{CD}
F(\pi_1)\otimes F(\rho_1) @>{F_2}>>F(\pi_1\otimes \rho_1)\\
@A {\unit=F_0\otimes\iota}AA @VV{F(\lambda_1)=\unit}V\\
\pi_2\otimes F(\rho_1)@>{\lambda_2=\unit}>> F(\rho_1)
\end{CD}
\]
Similarly, we have
\begin{align}
F_2(\rho_1,\pi_1)=\lmk \wrl{\rho_1\otimes\pi_1}2\rmk^*
\trlzi{\rho_1}1(  \wrl{\pi_1}2 ) \lmk \wrl{\rho_1}2\rmk^*=\unit.
\end{align}
Hence $F$ is a unitary tensor functor.
Similarly, $G$ is a unitary tensor functor.

From (\ref{trh}),(\ref{trh2}),, we have
\begin{align}
\eta(\rho_1):=\lmk \wrlt{F(\rho_1)}1\rmk^* \lmk \wrl{\rho_1}2\rmk^*\in \lmk
\rho_1,GF(\rho_1)\rmk\cap \lmk \trlzi{\rho_1}1,\trlzi{GF(\rho_1)}1\rmk.
\end{align}
Note that $\eta(\pi_1)=\unit$.

Furthermore, we have
\begin{align}
\begin{split}
&GF(\rho_1)\otimes GF(\sigma_1)
=\Ad\lmk G_2\lmk F(\rho_1), F(\sigma_1)\rmk^*\rmk
\circ G\lmk F(\rho_1)\otimes F(\sigma_1)\rmk\\
&=\Ad\lmk G_2\lmk F(\rho_1), F(\sigma_1)\rmk^*\rmk
\Ad\lmk \lmk \wrlt{F(\rho_{1})\otimes F(\sigma_{1})}1\rmk^*\rmk
\Ad\lmk F_2(\rho_1,\sigma_1)^*\rmk F(\rho_1\otimes \sigma_1)\circ\alpha^{-1}
\\
&=\Ad\lmk
\begin{gathered}
\lmk \wrlt{F(\rho_1)}1\rmk^{*}
\trlzi{F(\rho_1)}2\lmk \lmk \wrlt{F(\sigma_1)}1\rmk^{*}\rmk
\lmk \wrlt{F(\rho_1)\otimes F(\sigma_1)}1\rmk\\
\lmk \wrlt{F(\rho_{1})\otimes F(\sigma_{1})}1\rmk^{*}
\lmk \wrl{\rho_1}2\rmk^{*} \trlzi{\rho_1}1\lmk \lmk \wrl{\sigma_1}2\rmk^{*}\rmk
\lmk \wrl{\rho_1\otimes\sigma_1}2\rmk
\end{gathered}
\rmk\\
&\quad F(\rho_{1}\otimes \sigma_{1})\circ\alpha^{-1}\\
&=\Ad\lmk
\begin{gathered}
\lmk \wrlt{F(\rho_1)}1\rmk^{*}
\trlzi{F(\rho_1)}2\lmk \lmk \wrlt{F(\sigma_1)}1\rmk^{*}\rmk
\lmk \wrl{\rho_1}2\rmk^{*} \trlzi{\rho_1}1\lmk \lmk \wrl{\sigma_1}2\rmk^{*}\rmk
\lmk \wrl{\rho_1\otimes\sigma_1}2\rmk
\end{gathered}
\rmk\\
&\quad F(\rho_{1}\otimes \sigma_{1})\circ\alpha^{-1}\\
&=\Ad\lmk
\begin{gathered}
\lmk \wrlt{F(\rho_1)}1\rmk^{*}
\lmk \wrl{\rho_1}2\rmk^{*}
\trlzi{\rho_1}1\lmk \lmk \wrlt{F(\sigma_1)}1\rmk^{*}\rmk
 \trlzi{\rho_1}1\lmk \lmk \wrl{\sigma_1}2\rmk^{*}\rmk
\lmk \wrl{\rho_1\otimes\sigma_1}2\rmk\\
\lmk \wrl{\rho_1\otimes\sigma_1}2\rmk^{*}
\end{gathered}
\rmk\\
&\quad (\rho_{1}\otimes \sigma_{1})\\
&=\Ad\lmk
\begin{gathered}
\lmk \wrlt{F(\rho_1)}1\rmk^{*}
\lmk \wrl{\rho_1}2\rmk^{*}
\trlzi{\rho_1}1\lmk \lmk \wrlt{F(\sigma_1)}1\rmk^{*}
 \lmk \wrl{\sigma_1}2\rmk^{*}\rmk
\end{gathered}
\rmk\\
&\quad (\rho_{1}\otimes \sigma_{1})\\
&=\Ad\lmk
\begin{gathered}
\lmk \wrlt{F(\rho_1)}1\rmk^{*}
\lmk \wrl{\rho_1}2\rmk^{*}
\trlzi{\rho_1}1\lmk \lmk \wrlt{F(\sigma_1)}1\rmk^{*}
 \lmk \wrl{\sigma_1}2\rmk^{*}\rmk\\
 \wrl{\rho_{1}\otimes \sigma_{1}}2
\wrlt {F(\rho_{1}\otimes \sigma_{1})}1
\end{gathered}
\rmk GF(\rho_{1}\otimes \sigma_{1})\\
&=
\Ad\lmk \lmk \eta(\rho_{1})\otimes \eta(\sigma_{1})\rmk \eta(\rho_{1}\otimes \sigma_{1})^{*}\rmk
GF(\rho_{1}\otimes \sigma_{1}).
\end{split}
\end{align}
We used (\ref{201}), (\ref{188}), (\ref{200}), (\ref{202}), (\ref{199}), (\ref{189})
occasionally.
In the fifth equation we used (\ref{trh}).
Setting 
\begin{align}
(GF)_{2}(\rho_{1},\sigma_{1}):=\eta(\rho_{1}\otimes \sigma_{1})
\lmk \eta(\rho_{1})\otimes \eta(\sigma_{1})\rmk^{*},
\end{align}
we conclude
\begin{align}
\Ad\lmk (GF)_{2}(\rho_{1},\sigma_{1})\rmk\circ \lmk GF(\rho_1)\otimes GF(\sigma_1)\rmk
=GF(\rho_1\otimes\sigma_1).
\end{align}
Because
\begin{align}
 (GF)_{2}\lmk \rho_{1},\sigma_{1}\rmk\cdot
\lmk \eta(\rho_{1})\otimes \eta(\sigma_{1})\rmk
=\eta(\rho_{1}\otimes \sigma_{1}),
\end{align}
we see that
\[
\begin{CD}
\rho_1\otimes\sigma_1@>{id}>>\rho_1\otimes\sigma_1\\
@V{\eta\otimes\eta}VV @VV{\eta}V\\
GF(\rho_1)\otimes GF(\sigma_1)@>{(GF)_2} >> GF\lmk \rho_1\otimes\sigma_1\rmk.
\end{CD}
\]
holds.
Hence we see $GF\simeq \id$.
Similarly we have $FG\simeq \id$.
Hence $C_{1}$ and $C_{2}$ are unitarily monoidally equivalent.

Finally, we show that $F$ is braided.
Let $\rho_1,\sigma_1\in\caO_1$.
We would like to show 
\begin{align}\label{barcd}
\begin{CD}
F(\rho_1)\otimes F(\sigma_1)@>{F_2(\rho_1,\sigma_1)}>>F\lmk \rho_1\otimes\sigma_1\rmk\\
@V{\epsilon_+^{\Lambda_0^{(2)}}\lmk F(\rho_1), F(\sigma_1)\rmk}
VV @VV{F\lmk \epsilon_+^{\Lambda_0^{(1)}}(\rho_1,\sigma_1)\rmk}V\\
F(\sigma_1)\otimes F(\rho_1)@>{F_2(\sigma_1,\rho_1)} >> F\lmk \sigma_1\otimes\rho_1\rmk.
\end{CD}
\end{align}
Fix
$\Lambda_1,\Lambda_2\in\ctv$,
satisfying $\Lambda_1\perp_{(\theta,\varphi)} \Lambda_2$, 
$\Lambda_2\leftarrow_{(\theta,\varphi)}\Lambda_1$.
For each $t_j\ge 0$, $j=1,2$, we fix $\bar V_{\rho_1\ltj{j}{t_j}}\in \caV_{\rho_1\ltj{j}{t_j}}$, $\bar V_{\sigma_1\ltj{j}{t_j}}\in \caV_{\sigma_1\ltj{j}{t_j}}$
and 
$\bar V_{F(\rho_1)\ltj{j}{t_j}}\in \caV_{F(\rho_1)\ltj{j}{t_j}}$, $\bar V_{F(\sigma_1)\ltj{j}{t_j}}\in \caV_{F(\sigma_1)\ltj{j}{t_j}}$.
Set
\begin{align}
\begin{split}
&W_{\rho_1\Lambda_0^{(1)}\Lambda_1}^{\bm t}:=\bar V_{\rho_1\Lambda_1+t_1 \bm e_{\Lambda_1}}\bar V_{\rho_1\Lambda_0^{(1)}}^*=\bar V_{\rho_1\Lambda_1+t_1 \bm e_{\Lambda_1}}
,\quad
W_{\sigma_1\Lambda_0^{(1)}\Lambda_2}^{\bm t}:=\bar V_{\sigma_1\Lambda_2+t_2\bm e_{\Lambda_2}}\bar V_{\sigma_1\Lambda_0^{(1)}}^*=\bar V_{\sigma_1\Lambda_2+t_2\bm e_{\Lambda_2}},\\
&W_{F(\rho_1)\Lambda_0^{(2)}\Lambda_1}^{\bm t}:=\bar V_{F(\rho_1)\Lambda_1+t_1 \bm e_{\Lambda_1}}\bar V_{F(\rho_1)\Lambda_0^{(2)}}^*=\bar V_{F(\rho_1)\Lambda_1+t_1 \bm e_{\Lambda_1}},\quad
W_{F(\sigma_1)\Lambda_0^{(2)}\Lambda_2}^{\bm t}:=\bar V_{F(\sigma_1)\Lambda_2+t_2\bm e_{\Lambda_2}}\bar V_{F(\sigma_1)\Lambda_0^{(2)}}^*=\bar V_{F(\sigma_1)\Lambda_2+t_2\bm e_{\Lambda_2}}.
\end{split}
\end{align}
With these notation, (\ref{barcd}) means
\begin{align}
\begin{split}
&
\lmk \wrl{\sigma_1\otimes\rho_1}2\rmk^* \trlzi{\sigma_1}1
\lmk \wrl{\rho_1}2\rmk \wrl{\sigma_1}2\\
&\lim_{\bm t\to\infty}
\trlzi{F(\sigma_1)}{2}\lmk
\lmk
W_{F(\rho_1){\Lambda_0^{(2)}}\Lambda_1}^{\bm t}
\rmk^*\rmk
\lmk  W_{F(\sigma_1){\Lambda_0^{(2)}}\Lambda_2}^{\bm t}\rmk^*
W_{F(\rho_1){\Lambda_0^{(2)}}\Lambda_1}^{\bm t}
\trlzi{F(\rho_1)}2\lmk
W_{F(\sigma_1){\Lambda_0^{(2)}}\Lambda_2}^{\bm t}
\rmk
\\
&=\lmk \wrl{\sigma_1\otimes\rho_1}2\rmk^* \trlzi{\sigma_1}1
\lmk \wrl{\rho_1}2\rmk \wrl{\sigma_1}2
\lim_{\bm t\to\infty}
\lmk W_{F(\sigma_1){\Lambda_0^{(2)}}\Lambda_2}^{\bm t}\otimes W_{F(\rho_1){\Lambda_0^{(2)}}\Lambda_1}^{\bm t}\rmk^*
\lmk
W_{F(\rho_1){\Lambda_0^{(2)}}\Lambda_1}^{\bm t}\otimes  W_{F(\sigma_1){\Lambda_0^{(2)}}\Lambda_2}^{\bm t}
\rmk\\
&=F_2(\sigma_1,\rho_1)\epsilon_+^{\Lambda_0^{(2)}}\lmk F(\rho_1), F(\sigma_1)\rmk\\
&=F\lmk \epsilon_+^{\Lambda_0^{(1)}}(\rho_1,\sigma_1)\rmk F_2(\rho_1,\sigma_1)\\
&= 
\lmk Y_{\sigma_1\otimes \rho_1}^{{\Lambda_0}^{(2)}}\rmk^*
\lim_{\bm t\to\infty}
\lmk W_{\sigma_1{\Lambda_0^{(1)}}\Lambda_2}^{\bm t}\otimes W_{\rho_1{\Lambda_0^{(1)}}\Lambda_1}^{\bm t}\rmk^*
\lmk
W_{\rho_1{\Lambda_0^{(1)}}\Lambda_1}^{\bm t}\otimes  W_{\sigma_1{\Lambda_0^{(1)}}\Lambda_2}^{\bm t}
\rmk
Y_{\rho_1\otimes \sigma_1}^{{\Lambda_0^{(2)}}}
\lmk \wrl{\rho_1\otimes\sigma_1}2\rmk^* \trlzi{\rho_1}1\lmk \wrl{\sigma_1}2\rmk\wrl{\rho_1}2\\
&=
\lmk Y_{\sigma_1\otimes \rho_1}^{{\Lambda_0}^{(2)}}\rmk^*
\lim_{\bm t\to\infty}
\trlzi{\sigma_1}1
\lmk\lmk W_{\rho_1{\Lambda_0^{(1)}}\Lambda_1}^{\bm t}\rmk^*\rmk
 \lmk W_{\sigma_1{\Lambda_0^{(1)}}\Lambda_2}^{\bm t}\rmk^*
\lmk
W_{\rho_1{\Lambda_0^{(1)}}\Lambda_1}^{\bm t}
\trlzi{\rho_1}1\lmk  W_{\sigma_1{\Lambda_0^{(1)}}\Lambda_2}^{\bm t}\rmk
\rmk\\
&Y_{\rho_1\otimes \sigma_1}^{{\Lambda_0^{(2)}}}
\lmk \wrl{\rho_1\otimes\sigma_1}2\rmk^* \trlzi{\rho_1}1\lmk \wrl{\sigma_1}2\rmk\wrl{\rho_1}2
\end{split}
\end{align}
This can be rewritten to
\begin{align}\label{sfa}
\begin{split}
\lim_{\bm t\to\infty}
\lV\begin{gathered}
 \lmk W_{\sigma_1{\Lambda_0^{(1)}}\Lambda_2}^{\bm t}\rmk
\trlzi{\sigma_1}1
\lmk W_{\rho_1{\Lambda_0^{(1)}}\Lambda_1}^{\bm t}\rmk
 \trlzi{\sigma_1}1
\lmk \wrl{\rho_1}2\rmk \wrl{\sigma_1}2\\
\trlzi{F(\sigma_1)}{2}\lmk
\lmk
W_{F(\rho_1){\Lambda_0^{(2)}}\Lambda_1}^{\bm t}
\rmk^*\rmk
\lmk  W_{F(\sigma_1){\Lambda_0^{(2)}}\Lambda_2}^{\bm t}\rmk^*
\\
-
\lmk
W_{\rho_1{\Lambda_0^{(1)}}\Lambda_1}^{\bm t}
\trlzi{\rho_1}1\lmk  W_{\sigma_1{\Lambda_0^{(1)}}\Lambda_2}^{\bm t}\rmk
\rmk\\
\trlzi{\rho_1}1\lmk \wrl{\sigma_1}2\rmk\wrl{\rho_1}2
\trlzi{F(\rho_1)}2\lmk \lmk
W_{F(\sigma_1){\Lambda_0^{(2)}}\Lambda_2}^{\bm t}\rmk^*
\rmk
\lmk W_{F(\rho_1){\Lambda_0^{(2)}}\Lambda_1}^{\bm t}\rmk^*
\end{gathered}\rV=0.
\end{split}
\end{align}
Using (\ref{trh}) and $W_{\rho_1{\Lambda_0^{(1)}}\Lambda_1}^{\bm t}\in \lmk
\trlzi{\rho_1}1,\trlt{\rho_1}1\rmk$ etc,
the first term in the norm is equal to 
\begin{align}
 \lmk W_{\sigma_1{\Lambda_0^{(1)}}\Lambda_2}^{\bm t}\rmk\wrl{\sigma_1}2
 \lmk
W_{F(\sigma_1){\Lambda_0^{(2)}}\Lambda_2}^{\bm t}
\rmk^*
\trlt{F(\sigma_1)}2
\lmk W_{\rho_1{\Lambda_0^{(1)}}\Lambda_1}^{\bm t}\wrl{\rho_1}2
  \lmk W_{F(\rho_1){\Lambda_0^{(2)}}\Lambda_1}^{\bm t}\rmk^*\rmk.
\end{align}
The second term is equal to
\begin{align}
\begin{split}
&W_{\rho_1{\Lambda_0^{(1)}}\Lambda_1}^{\bm t}
\wrl{\rho_1}2\lmk W_{F(\rho_1){\Lambda_0^{(2)}}\Lambda_1}^{\bm t}\rmk^*
\trlt{F(\rho_1)}1\lmk  W_{\sigma_1{\Lambda_0^{(1)}}\Lambda_2}^{\bm t}
 \wrl{\sigma_1}2\lmk
W_{F(\sigma_1){\Lambda_0^{(2)}}\Lambda_2}^{\bm t}\rmk^*
\rmk.
\end{split}
\end{align}
Setting
\begin{align}
{X^{t_1}_{\rho_1\Lambda_1}}
:= W_{\rho_1\Lambda_0^{(1)}\Lambda_1}^{\bm t}
\wrl{\rho_1}2
\lmk W_{F(\rho_1)\Lambda_0^{(2)}\Lambda_1}^{\bm t}\rmk^*,\quad
{X^{t_2}_{\sigma_1\Lambda_2}}
:= W_{\sigma_1\Lambda_0^{(1)}\Lambda_2}^{\bm t}
\wrl{\sigma_1}2
\lmk W_{F(\sigma_1)\Lambda_0^{(2)}\Lambda_2}^{\bm t}\rmk^*,
\end{align}
(\ref{sfa}) becomes
\begin{align}\label{art}
\lim_{t\to\infty}
\lV
{X^{t_2}_{\sigma_1\Lambda_2}}\trlt{F(\sigma_1)}2
\lmk 
{X^{t_1}_{\rho_1\Lambda_1}}\rmk
-{X^{t_1}_{\rho_1\Lambda_1}}\trlt{F(\rho_1)}1
\lmk {X^{t_2}_{\sigma_1\Lambda_2}}
\rmk
\rV=0.
\end{align}
This is the equation we have to show now.

In order to do so, fix some $\delta>0$
such that $\arg (\Lambda_1)_{64\delta}\cap \arg (\Lambda_2)_{64\delta}=\emptyset$
and $ (\Lambda_1)_{64\delta},  (\Lambda_2)_{64\delta}\in\ctv$.
We set 
$\amf_{i}(\Lambda):=\pi_{i}\lmk \caA_{\Lambda^{c}}\rmk'$,
$i=1,2$.

Fix any $\epsilon>0$.
Recall $g_{\varphi,\delta,\delta'}$ from Definition \ref{qfdef}.
Fix some $s\ge 0$ such that
\begin{align}\label{schoice}
\begin{split}
&s\ge R_{|\arg\Lambda_{i}|+10\delta,\delta}\\
& g_{|\arg\Lambda_1|+2\delta,\delta,\delta}(s)<\frac{\epsilon}{32},\\
&g_{|\arg \Lambda_i|+12\delta,\delta,\delta}\lmk \frac s2\rmk
+2g_{|\arg \Lambda_i|+16\delta,\delta,\delta}\lmk \frac s2\rmk
+f_{|\arg \Lambda_i|+12\delta,\delta,\delta}\lmk \frac s2\rmk<\frac{\epsilon}{32},\\
&g_{|\arg \Lambda_i|+22\delta,\delta,\delta}\lmk \frac s2\rmk
+2g_{|\arg \Lambda_i|+26\delta,\delta,\delta}\lmk \frac s2\rmk
+f_{|\arg \Lambda_i|+22\delta,\delta,\delta}\lmk \frac s2\rmk<\frac{\epsilon}{32},
\end{split}
\end{align}
for $i=1,2$.
Now for any $t_1\ge 0$, we claim that
there is some unitary
$z^{t_1s}_{\rho_1\Lambda_1}\in \amf_{2}\lmk
\lmk\Lambda_1\rmk_{9\delta}+(t_1-2s)\bm e_{\Lambda_1}
\rmk$
such that
\begin{align}\label{oboe}
\lV
z^{t_1s}_{\rho_1\Lambda_1}-{X^{t_1}_{\rho_1\Lambda_1}}
\rV\le \frac{\epsilon}4.
\end{align}
To prove this, set $\Lambda_{j}^t:=\Lambda_j+t\bm e_{\Lambda_j}$
for each $t\ge 0$ and $j=1,2$.
Then we have
\begin{align}
\begin{split}
&
\Ad\lmk {X^{t_1}_{\rho_1\Lambda_1}}
\rmk\pi_2\vert_{\caA_{\lmk \lmk \Lambda_{1}^{t_1}\rmk_\delta\rmk^c}}
=
\Ad\lmk W_{\rho_1\Lambda_0^{(1)}\Lambda_1}^{\bm t}
\wrl{\rho_1}2
\lmk W_{F(\rho_1)\Lambda_0^{(2)}\Lambda_1}^{\bm t}\rmk^*\rmk\pi_2\vert_{\caA_{\lmk \lmk \Lambda_{1}^{t_1}\rmk_\delta\rmk^c}}\\
&=
\Ad\lmk W_{\rho_1\Lambda_0^{(1)}\Lambda_1}^{\bm t}
\wrl{\rho_1}2
\rmk
F(\rho_1)\vert_{\caA_{\lmk \lmk \Lambda_{1}^{t_1}\rmk_\delta\rmk^c}}\\
&=
\Ad\lmk W_{\rho_1\Lambda_0^{(1)}\Lambda_1}^{\bm t}
\wrl{\rho_1}2
\rmk
\trlzi{F(\rho_1)}2\circ\pi_2
\vert_{\caA_{\lmk \lmk \Lambda_{1}^{t_1}\rmk_\delta\rmk^c}}\\
&=
\Ad\lmk W_{\rho_1\Lambda_0^{(1)}\Lambda_1}^{\bm t}
\rmk
\trlzi{\rho_1}1\circ\pi_2
\vert_{\caA_{\lmk \lmk \Lambda_{1}^{t_1}\rmk_\delta\rmk^c}}\\
&=
\trlta{\rho_1}1\circ\pi_2
\vert_{\caA_{\lmk \lmk \Lambda_{1}^{t_1}\rmk_\delta\rmk^c}}=
\trlta{\rho_1}1\circ\pi_1\alpha
\vert_{\caA_{\lmk \lmk \Lambda_{1}^{t_1}\rmk_\delta\rmk^c}}\\
&=
\Ad\lmk
\trlta{\rho_1}1\pi_1
\lmk v_{{\lmk {\Lambda_1^{t_1}}\rmk_{\delta}}\delta}\rmk
\rmk
\circ
\trlta{\rho_1}1\pi_1
\lmk
\Xi_{{\lmk {\Lambda_1^{t_1}}\rmk_{\delta}}, \delta}\circ \beta_{({\lmk {\Lambda_1^{t_1}}\rmk_{\delta}})^c}
\rmk\vert_{\caA_{\lmk \lmk \Lambda_{1}^{t_1}\rmk_\delta\rmk^c}}\\
&=
\Ad\lmk
\trlta{\rho_1}1\pi_1
\lmk v_{{\lmk {\Lambda_1^{t_1}}\rmk_{\delta}}\delta}\rmk
\rmk
\pi_1
\circ
\lmk
\Xi_{{\lmk {\Lambda_1^{t_1}}\rmk_{\delta}}, \delta}\circ \beta_{({\lmk {\Lambda_1^{t_1}}\rmk_{\delta}})^c}
\rmk\vert_{\caA_{\lmk \lmk \Lambda_{1}^{t_1}\rmk_\delta\rmk^c}}\\
&=
\Ad\lmk
\trlta{\rho_1}1\pi_1
\lmk v_{{\lmk {\Lambda_1^{t_1}}\rmk_{\delta}}\delta}\rmk\cdot
\pi_1 \lmk v^*_{{\lmk {\Lambda_1^{t_1}}\rmk_{\delta}}\delta}\rmk
\rmk
\pi_1
\alpha\vert_{\caA_{\lmk \lmk \Lambda_{1}^{t_1}\rmk_\delta\rmk^c}}
\\
&=
\Ad\lmk
\trlta{\rho_1}1\pi_1
\lmk v_{{\lmk {\Lambda_1^{t_1}}\rmk_{\delta}}\delta}\rmk\cdot
\pi_1 \lmk v^*_{{\lmk {\Lambda_1^{t_1}}\rmk_{\delta}}\delta}\rmk
\rmk
\pi_2\vert_{\caA_{\lmk \lmk \Lambda_{1}^{t_1}\rmk_\delta\rmk^c}}.
\end{split}
\end{align}
Hence we have
\begin{align}\label{accord}
 {X^{t_1}_{\rho_1\Lambda_1}}
 =\trlta{\rho_1}1\pi_1
\lmk v_{{\lmk {\Lambda_1^{t_1}}\rmk_{\delta}}\delta}\rmk\cdot
\pi_1 \lmk v^*_{{\lmk {\Lambda_1^{t_1}}\rmk_{\delta}}\delta}\rmk
\cdot y^{t_1}_{\rho_1\Lambda_1},
\end{align}
with some unitary $ y^{t_1}_{\rho_1\Lambda_1}\in \amf_{2}(\lmk \Lambda_{1}^{t_1}\rmk_\delta)$.
Recall from Definition \ref{qfdef} that
there is a unitary
\begin{align}\label{accord2}
v_{{\lmk {\Lambda_1^{t_1}}\rmk_{\delta}},\delta,\delta,s}'\in
\caA_{\lmk {\Lambda_1^{t_1}}\rmk_{3\delta}-s\bm e_{\Lambda_1}}
\end{align}
such that
\begin{align}\label{accord3}
\lV
v_{{\lmk {\Lambda_1^{t_1}}\rmk_{\delta}},\delta}-v_{{\lmk {\Lambda_1^{t_1}}\rmk_{\delta}},\delta,\delta,s}'
\rV
\le g_{|\arg\Lambda_1|+2\delta,\delta,\delta}(s)<\frac{\epsilon}{32}.
\end{align}
Because ${\Lambda_1^{t_1}}\subset \lmk {\Lambda_1^{t_1}}\rmk_{3\delta}-s\bm e_{\Lambda_1}$,
we have
\begin{align}\label{accord4}
\begin{split}
&\pi_1 \lmk
v_{{\lmk {\Lambda_1^{t_1}}\rmk_{\delta}},\delta,\delta,s}'\rmk,
\trlta{\rho_1}1\lmk
\pi_1\lmk v_{{\lmk {\Lambda_1^{t_1}}\rmk_{\delta}},\delta,\delta,s}'\rmk\rmk\\
&\in 
\trlta{\rho_1}1\lmk\pi_1
\lmk
\caA_{\lmk {\Lambda_1^{t_1}}\rmk_{3\delta}-s\bm e_{\Lambda_1}}
\rmk\rmk
\subset
\amf_{1}\lmk
{\lmk {\Lambda_1^{t_1}}\rmk_{4\delta}-s\bm e_{\Lambda_1}}
\rmk.
\end{split}
\end{align}
from Lemma \ref{lem33}.
Note from Lemma \ref{intr} and our choice of $s$ (\ref{schoice})
there is some unitary $w_{s} $
such that 
\begin{align}\label{inin}
\begin{split}
\lV w_{s}-\unit\rV<\frac{\epsilon}{32},\\
\Ad\lmk w_{s}\rmk
\lmk
\amf_{1}\lmk
{\lmk {\Lambda_1^{t_1}}\rmk_{4\delta}-s\bm e_{\Lambda_1}}
\rmk
\rmk
\subset \amf_{2}\lmk
{\lmk {\Lambda_1^{t_1}}\rmk_{9\delta}-2s\bm e_{\Lambda_1}}
\rmk.
\end{split}
\end{align}

From (\ref{accord}), (\ref{accord2}), (\ref{accord3}), (\ref{accord4}),(\ref{inin}),
we conclude the claim (\ref{oboe}).
Similarly, 
 for any $t\ge 0$,
there is some unitary
$z^{t_2s}_{\sigma_1\Lambda_2}\in \amf_{2}\lmk
\lmk\Lambda_2\rmk_{9\delta}+(t_2-2s)\bm e_{\Lambda_2}
\rmk$
such that
\begin{align}\label{oboe2}
\lV
z^{t_2s}_{\sigma_1\Lambda_2}-{X^{t_2}_{\sigma_1\Lambda_2}}
\rV\le \frac{\epsilon}4.
\end{align}

From this, we can estimate
\begin{align}\label{last}
\begin{split}
&\lV
{X^{t_2}_{\sigma_1\Lambda_2}}\trlta{F(\sigma_1)}2
\lmk 
{X^{t_1}_{\rho_1\Lambda_1}}\rmk
-{X^{t_1}_{\rho_1\Lambda_1}}\trlta{F(\rho_1)}1
\lmk {X^{t_2}_{\sigma_1\Lambda_2}}
\rmk
\rV\\
&\le
{\epsilon}
+
\lV
z^{t_2s}_{\sigma_1\Lambda_2}\trlta{F(\sigma_1)}2
\lmk 
z^{t_1s}_{\rho_1\Lambda_1}\rmk
-z^{t_1s}_{\rho_1\Lambda_1}\trlta{F(\rho_1)}1
\lmk z^{t_2s}_{\sigma_1\Lambda_2}
\rmk
\rV\\
&\le
{\epsilon}
+
\lV
\left.\trlta{F(\sigma_1)}2
\right\vert_{\amf_{2}\lmk\lmk\Lambda_1\rmk_{9\delta}+(t_1-2s)\bm e_{\Lambda_1}\rmk}
-\id_{\amf_{2}\lmk\lmk\Lambda_1\rmk_{9\delta}+(t_1-2s)\bm e_{\Lambda_1}\rmk}
\rV\\
&+
\lV
\left.\trlta{F(\rho_1)}1
\right\vert_{\amf_{2}\lmk\lmk\Lambda_2\rmk_{9\delta}+(t_2-2s)\bm e_{\Lambda_2}\rmk}
-\id_{\amf_{2}\lmk\lmk\Lambda_2\rmk_{9\delta}+(t_2-2s)\bm e_{\Lambda_2}\rmk}
\rV\\
&+
\sup_{\substack{x\in \amf_{2}\lmk\lmk\Lambda_2\rmk_{9\delta}+(t_{2}-2s)\bm e_{\Lambda_2}\rmk, \\y\in \amf_{2}\lmk
\lmk\Lambda_1\rmk_{9\delta}+(t_1-2s)\bm e_{\Lambda_1}
\rmk
,\\ \lV x\rV,\lV y\rV\le 1}}
\lV
xy-yx
\rV.
\end{split}
\end{align}
%
%
%
%
%
From Lemma \ref{lem319}, Lemma \ref{lem320},
the last three terms in the last line of (\ref{last}) converges to $0$
as $t_1,t_2\to\infty$.
Therefore, for $t_1,t_2$ large enough,
we have 
\begin{align}
&\lV
{X^{t_2}_{\sigma_1\Lambda_2}}\trlta{F(\sigma_1)}2
\lmk 
{X^{t_1}_{\rho_1\Lambda_1}}\rmk
-{X^{t_1}_{\rho_1\Lambda_1}}\trlta{F(\rho_1)}1
\lmk {X^{t_2}_{\sigma_1\Lambda_2}}
\rmk
\rV\le \epsilon.
\end{align}
This proves 
(\ref{art}) hence (\ref{barcd}).
Hence our monoidal equivalence $F$ is braided.
\end{proof}
{\bf Acknowledgment.}\\
This work was supported by JSPS KAKENHI Grant Number 16K05171 and 19K03534.
It was also supported by JST CREST Grant Number JPMJCR19T2.

\appendix
\section{Cones}\label{apcone}
We cpllect basic notation and properties of cones.
For $I\subset \bbR$, we set
\[
\bbA_I:=\left\{e^{it}\mid t\in I\right\}\subset \bbT.
\]
We say $\bbA\subset \bbT$ is a closed interval if it is $\bbA=\bbA_I$
with a closed interval $I$ of $\bbR$. 
For $\theta\in\bbR$, we set $\bm e_{\theta}=(\cos\theta,\sin\theta)$.
For $\bm a\in \bbR$, $\theta\in\bbR$ and $\varphi\in (0,\pi)$,
set 
\begin{align*}
\Lambda_{\bm a, \theta,\varphi}
:=&\left\{
\bm x\in\bbR^{2}\mid (\bm x-\bm a)\cdot \bm e_{\theta}>\cos\varphi\cdot \lV \bm x-{\bm a}\rV
\right\}\\
=&\bm a+\left\{
t\bm e_{\beta}\mid t>0,\quad \beta\in (\theta-\varphi,\theta+\varphi)\right\}.
\end{align*}
We call a subset of $\bbR^{2}$ with this form a {\it cone}.
For a cone $\Lambda=\Lambda_{\bm a, \theta,\varphi}$ given above,
we set
\begin{align*}
\arg\lmk \Lambda\rmk:=
\left\{ e^{it}\in\bbT \mid t\in [\theta-\varphi,\theta+\varphi]\right\},\quad
|\arg\lmk \Lambda\rmk|:=2\varphi,\quad \text{and}\quad
\bm a_{\Lambda}:=\bm a,\quad
\bm e_{\Lambda}:=\bm e_{\theta}.
\end{align*}
For $\varepsilon>0$ and $\Lambda=\Lambda_{\bm a, \theta,\varphi}$ with $\varphi+\varepsilon<\pi$
we denote the ``fattened'' cone by
\begin{align*}
\Lambda_{\varepsilon}:=\Lambda_{\bm a, \theta,\varphi+\varepsilon}.
\end{align*}
If $0<\varphi-\varepsilon$, we also set
\begin{align*}
\Lambda_{-\varepsilon}:=\Lambda_{\bm a, \theta,\varphi-\varepsilon}.
\end{align*}

\begin{lem}\label{lem1}
Let $\kappa\in\bbR$.
Let $\bm a_{1},\bm a_{2}\in \bbR^{2}$, $\theta_{1},\theta_{2}\in \bbR$ and
$\varphi_{1},\varphi_{2}\in (0,\pi)$
with 
\begin{align*}
\kappa\le \theta_{1}-\varphi_{1}, \theta_{2}-\varphi_{2},\quad\text{and}\quad
\theta_{1}+\varphi_{1}, \theta_{2}+\varphi_{2}<\kappa+2\pi.
\end{align*}
Set
\begin{align*}
\eta:=\max\left\{ \theta_{1}+\varphi_{1}, \theta_{2}+\varphi_{2}\right\},\\
\zeta:=\min\left\{ \theta_{1}-\varphi_{1}, \theta_{2}-\varphi_{2}\right\},
\end{align*}
and 
\begin{align*}
\theta:=\frac{\eta+\zeta}2,\quad \varphi:=\frac{\eta-\zeta}2.
\end{align*}
Assume that $\varphi\in (0,\pi)$.
Then there is some $R_{0}\in\bbR_{\ge 0}$
such that 
\begin{align*}
\Lambda_{\bm a_{1},\theta_{1},\varphi_{1}}, 
\Lambda_{\bm a_{2},\theta_{2},\varphi_{2}}\subset 
\Lambda_{-R\bm e_{\theta}, \theta,\varphi},
\end{align*}
for all $R\ge R_{0}$.
\end{lem}
\begin{proof}
By the definition of $\theta,\varphi$, we have
\begin{align}\label{bib}
(\theta_{i}-\varphi_{i}, \theta_{i}+\varphi_{i})\subset (\theta-\varphi,\theta+\varphi),\quad i=1,2.
\end{align}
Any element of $\Lambda_{\bm a_{i},\theta_{i},\varphi_{i}}$, $i=1,2$
is of the form
\begin{align*}
\bm a_{i}+t\bm e_{\beta}, \quad \beta \in (\theta_{i}-\varphi_{i}, \theta_{i}+\varphi_{i}),\quad
0\le t\in \bbR.
\end{align*}
We have to show that there is some $R_{0}>0$ such that
\begin{align*}
\lmk
\bm a_{i}+t\bm e_{\beta}-(-R\bm e_{\theta})
\rmk
\cdot \bm e_{\theta}
> \cos\varphi\cdot \lV \bm a_{i}+t\bm e_{\beta}-(-R\bm e_{\theta})\rV
\end{align*}
for any $\beta \in (\theta_{i}-\varphi_{i}, \theta_{i}+\varphi_{i})$,
$0\le t\in \bbR$, $i=1,2$ and $R\ge R_{0}$.

Suppose that $\varphi\in (0,\frac\pi2]$.
Note that $-\cos\varphi+\cos(\beta-\theta)\ge 0$ 
because of (\ref{bib}).
Using this (in the third line), we get
\begin{align*}
\begin{split}
&\lmk
\bm a_{i}+t\bm e_{\beta}-(-R\bm e_{\theta})
\rmk\cdot \bm e_{\theta}-\cos\varphi\cdot \lV \bm a_{i}+t\bm e_{\beta}-(-R\bm e_{\theta})\rV\\
&\ge \bm a_{i}\cdot e_{\theta}+t\cos(\beta-\theta)+R
-\cos\varphi\lmk \lV \bm a_{i}\rV+t+R\rmk\\
&\ge 
\bm a_{i}\cdot \bm e_{\theta}-\cos\varphi \lV \bm a_{i}\rV+ 
t\lmk -\cos\varphi+\cos(\beta-\theta)\rmk + R(1-\cos\varphi)\\
&\ge 
\bm a_{i}\cdot \bm e_{\theta}-\cos\varphi \lV \bm a_{i}\rV+ 
R(1-\cos\varphi)
\end{split}
\end{align*}
The last term is positive if
\begin{align*}
R>(1-\cos\varphi)^{-1}\max_{i=1,2}
\left\{
-\bm a_{i}\cdot \bm e_{\theta}+\cos\varphi \lV \bm a_{i}\rV
\right\}.
\end{align*}
If $\varphi\in (\frac\pi 2,\pi)$, note that there is some $R_0\in\bbR_+$ such that
\begin{align}
\lmk-R\bm e_\theta-\bm a_i\rmk\cdot \bm e_\theta
<\cos\varphi \lV R\bm e_\theta+\bm a_i\rV,
\end{align}
for all $R\ge R_0$ and $i=1,2$.
If $R\ge R_0$ then
we have
\begin{align}\label{ihi}
\lmk
\bm a_{i}+t\bm e_{\beta}-(-R\bm e_{\theta})
\rmk
\cdot \bm e_{\theta}
=
\lmk R\bm e_\theta+\bm a_i\rmk \cdot \bm e_\theta
+t\bm e_\beta\cdot\bm e_\theta
>-\cos\varphi \lV R\bm e_\theta+\bm a_i\rV+t\cos(\beta-\theta).
\end{align}
On the other hand, we have
\begin{align}
\lV \bm a_{i}+t\bm e_{\beta}-(-R\bm e_{\theta})\rV
\ge \lV t\bm e_\beta\rV-\lV R\bm e_\theta+\bm a_i\rV
=t-\lV R\bm e_\theta+\bm a_i\rV.
\end{align}
From $\cos\varphi<0$, we then have
\begin{align}
\cos\varphi\lmk t-\lV R\bm e_\theta+\bm a_i\rV\rmk\ge 
\cos\varphi\cdot \lV \bm a_{i}+t\bm e_{\beta}-(-R\bm e_{\theta})\rV.
\end{align}
Substituting this to (\ref{ihi}), we have
\begin{align}
\begin{split}
\lmk
\bm a_{i}+t\bm e_{\beta}-(-R\bm e_{\theta})
\rmk
\cdot \bm e_{\theta}
>-\cos\varphi \lV R\bm e_\theta+\bm a_i\rV+t\cos(\beta-\theta)\\
\ge \cos\varphi\lmk -\lV R\bm e_\theta+\bm a_i\rV+t\rmk
\ge \cos\varphi\cdot \lV \bm a_{i}+t\bm e_{\beta}-(-R\bm e_{\theta})\rV,
\end{split}
\end{align}
proving the claim.

\end{proof}
Let us give another version of the Lemma.
\begin{lem}\label{lem1p}
For any $\bm a,\bm a_{1}\in \bbR^{2}$, $\theta,\theta_{1}\in \bbR$ and
$\varphi,\varphi_{1}\in (0,\pi)$
with 
\begin{align*}
\theta-\varphi\le \theta_{1}-\varphi_{1}<\theta_{1}+\varphi_{1}\le\theta+\varphi.
\end{align*}
Then there is some $R_{0}\in\bbR_{\ge 0}$
such that 
\begin{align*}
\Lambda_{R\bm e_{\theta_{1}}+\bm a_{1},\theta_{1},\varphi_{1}}
\subset 
\Lambda_{\bm a, \theta,\varphi},
\end{align*}
for all $R\ge R_{0}$.
\end{lem}
\begin{proof}
Set $\tilde \theta:=-\theta_{1}$, $\tilde \varphi:=\pi-\varphi_{1}$,
$\tilde\theta_{1}:=-\theta$, $\tilde\varphi_{1}:=\pi-\varphi$.
We have
\begin{align*}
\tilde\theta-\tilde\varphi\le \tilde\theta_1-\tilde\varphi_1<
\tilde\theta_1+\tilde\varphi_1\le \tilde\theta+\tilde\varphi.
\end{align*}
Then applying Lemma \ref{lem1}, (with $\theta_{1},\varphi_{1},\theta_{2},\varphi_{2}$, $\bm a_1$ 
replaced by
$\tilde \theta_{1}$, $\tilde\varphi_{1}$, $\tilde\theta$, $\tilde\varphi$, $-\bm a_1+\bm a$ respectively),
there is some $R_{0}\in \bbR_{\ge 0}$ such that
\begin{align*}
\Lambda_{-\bm a_{1}+\bm{a},-\theta,\pi-\varphi}
=\Lambda_{-\bm a_{1}+\bm{a},\tilde \theta_{1},\tilde \varphi_{1}}, 
\subset 
\Lambda_{-R\bm e_{\tilde \theta}, \tilde \theta,\tilde \varphi}
=\Lambda_{R\bm e_{\theta_1}, -\theta_1,\pi-\varphi_1},
\end{align*}
for all $R\ge R_0$.
Taking its complement, we obtain
\begin{align*}
\Lambda_{R\bm e_{\theta_1}, \theta_1,\varphi_1}
\subset 
\Lambda_{-\bm a_{1}+\bm{a},\theta,\varphi},
\end{align*}
for all $R\ge R_0$, proving the claim.

\end{proof}

Let $\theta,\varphi\in \bbR$ with $0<\varphi<\pi$.
We define a set of cones $\caC_{(\theta,\varphi)}$ by
\begin{align}\label{ctvdef}
\caC_{(\theta,\varphi)}:=\left\{
\Lambda\mid \arg\Lambda\cap\bbA_{[\theta-\varphi,\theta+\varphi]}=\emptyset
\right\}.
\end{align}
\begin{lem}\label{lem4}
Let $\theta\in\bbR$ and $\varphi\in (0,\pi)$.
Then $\caC_{(\theta,\varphi)}$ is a upward filtering set
with respect to the inclusion order.
Furthermore,
\begin{align}
\cup_{R\ge 0} \lmk \Lambda-R\bm e_{\Lambda}\rmk =\bbR^2.
\end{align}
\end{lem}
\begin{proof}
Immediate from Lemma \ref{lem1}.
\end{proof}

\end{document}